\documentclass[a4paper,11pt]{amsart}
\usepackage{amsfonts,amssymb,epsfig,latexsym}
\usepackage{amsmath}
\usepackage[latin1]{inputenc}
\usepackage{color,layout}
\usepackage{ifthen}

\usepackage{pdfsync}
\usepackage{graphicx}
\usepackage{color}
\usepackage{pdfsync}
\date{}

\newcommand{\be}{\begin{equation}}
\newcommand{\ee}{\end{equation}}

\def\la{\langle}
\def\ra{\rangle}

\def\R{\mathbb{R}}
\def\C{\mathbb{C}}
\def\Z{\mathbb{Z}}
\def\N{\mathbb{N}}

\renewcommand{\Re}{\text{{\rm Re}\;}}
\renewcommand{\Im}{\text{{\rm Im}\;}}

\newtheorem{theorem}{Theorem}[section]
\newtheorem{lemma}[theorem]{Lemma}
\newtheorem{proposition}[theorem]{Proposition}

\theoremstyle{definition}

\newtheorem{remark}[theorem]{Remark}

\numberwithin{equation}{section}

\title[Multidimensional phase space tunneling]{Resonance widths in a case of multidimensional phase space tunneling}
\author{
Alain GRIGIS${}^1$ \& 
Andr\'e MARTINEZ$ {}^2$
 }

\begin{document}

\maketitle 
\addtocounter{footnote}{1}
\footnotetext{{\tt\small Universit\'{e} Paris 13,  
D\'epartement de Math\'ematiques,  LAGA UMR CNRS 7539, Av. J.-B. Cl\'ement, 93430 Villetaneuse, France, grigis@math.univ-paris13.fr}}  
\addtocounter{footnote}{1}
\footnotetext{{\tt\small Universit\`a di Bologna,  
Dipartimento di Matematica, Piazza di Porta San Donato, 40127
Bologna, Italy, 
martinez@dm.unibo.it }\\Partly supported by Universit\`a di Bologna, Funds for Selected Research Topics}  
\begin{abstract}
We consider a semiclassical $2\times 2$ matrix Schr\"odinger operator of the form $P=-h^2\Delta {\bf I}_2 +
{\rm diag}(x_n-\mu , \tau V_2(x)) +hR(x,hD_x)$, where $\mu$ and $\tau$ are two  small positive constants, $V_2$ is real-analytic and admits a non degenerate minimum at 0, and $R=(r_{j,k}(x,hD_x))_{1\leq j,k\leq 2}$ is a symmetric off-diagonal $2\times 2$ matrix of first-order differential operators with analytic coefficients. Then, denoting by $e_1$ the first eigenvalue of $-\Delta +  \la \tau V_2''(0)x,x\ra /2$, and under some ellipticity condition on $r_{1,2}=r_{2,1}^*$, we show that, for any $\mu$ sufficiently small, and for $0<\tau \leq\tau(\mu)$ with some $\tau(\mu)>0$,  the unique resonance $\rho$ of $P$ such that $\rho = \tau V_2(0) + (e_1+r_{2,2}(0,0))h + {\mathcal O}(h^2)$ (as $h\rightarrow 0_+$) satisfies, 
$$
\Im \rho = -h^{\frac32}f(h,\ln\frac1{h})e^{-2S/h},
$$
where $f(h,\ln\frac1{h}) \sim \sum_{0\leq m\leq\ell} f_{\ell,m}h^\ell(\ln\frac1{h})^m$ is a symbol with $f_{0,0}>0$, and $S$ is the imaginary part of the complex action along some convenient closed path containing $(0,0)$ and consisting of a union of complex nul-bicharacteristics of $p_1:=\xi^2 - x_n-\mu$ and $p_2:=\xi^2 +\tau V_2(x)$ (broken instanton). This broken instanton is described in terms of the outgoing and incoming complex Lagrangian manifolds associated with $p_2$ at the point 
$(0,0)$, and their intersections with the characteristic set $p_1^{-1}(0)$ of $p_1$.
\end{abstract}  


\baselineskip = 18pt 
\vfill\eject

\section{Introduction}

We consider the semiclassical $2\times 2$ matrix Schr\"odinger operator,
$$
P= 
\left(\begin{array}{cc}
P_1 & 0\\
0 & P_2
\end{array}\right) + hR(x,hD_x)
$$
with,
\begin{eqnarray*}
&& P_1 := -h^2\Delta +x_n-\mu;\\
&& P_2 := -h^2\Delta +\tau V_2(x),
\end{eqnarray*}
where $x=(x', x_n)=(x_1,\dots ,x_n)$ is the current variable in $\R^n$ ($n\geq 2$),
$h>0$ denotes the semiclassical parameter, $R(x,hD_x)=(r_{j,k}(x,hD_x))_{1\leq j,k\leq 2}$ is a  formally self-adjoint off-diagonal $2\times 2$ matrix of first-order
semiclassical differential operators, 
 and $\mu , \tau$ are two positive extra-parameters that will
be taken sufficiently small later on. Moreover, we assume that the smooth potential $V_2$ is
bounded on $\R^n$  and
satisfies,
\be
\label{assV2}
V_2\geq 0\; ; \; V_2^{-1}(0) =\{ 0\} \; ; \;  {\rm Hess}V_2(0) >0 \; ; \; 
\liminf_{\vert x\vert\rightarrow\infty} V_2 >0.
\ee
In particular, the classically allowed region  at energy 0 associated with $P_2$ (that is, $\{ V_2(x)\leq 0\} = \{0\}$) is included into that associated with $P_1$ (that coincides with $\R^{n-1}\times (-\infty, \mu]$). For this reason, the possible tunneling between the two operators cannot be reduced to a tunneling in the $x$-variables (as it usually happens, for instance, in scalar multiple-wells problems -- see, e.g., \cite{HeSj1, Ma1} -- or for shape resonances -- see, e.g., \cite{HeSj2}), and actually, such a situation persists under any linear symplectic change of coordinates. Instead, one has to consider the corresponding energy shells which, in our case, are given by $\{(x,\xi)\in \R^{2n}\,; \, \xi^2 + x_n =\mu \}$ for $P_1$, and by $\{(x,\xi)\in \R^{2n}\,; \, \xi^2 + V_2(x) =0\}=\{(0,0)\}$ for $P_2$. Since they are disjoint, one expects some exponentially small tunneling between them. Indeed, some previous works already exists on such situations (see, e.g., \cite{Ba, Ma2, Ma3, Ma4, Ma5, MaNaSo, NaSo}), but, except in the one dimensional case (\cite{Ba, Na}), no optimal bounds are obtained. Here, we plan to give such an optimal bound in terms of the width of the lower resonance appearing in our model.
\vskip 0.3 cm

Since we plan to study the resonances of $P$ near 0, we also assume that $V_2$ and the coefficients of 
$R$ are  globally analytic on $\R^n$, and extend to bounded holomorphic functions in a complex strip of the form
${\mathcal S}_\delta :=\{ \vert \Im x\vert < \delta \}$ for some $\delta >0$ (actually, in the sequel, we will need a little bit more). Then, the resonances of $P$ near 0 are defined, e.g., as the eigenvalues close to  0 of  the complex translated operator $P_\theta$ ($0<\theta <\delta$) obtained by taking $x_n$ on $\R -i\theta$. 

As in \cite{Ma2}, one can construct formal WKB eigenfunctions of $P_\theta$ with energy
close to 0 and, using the 
ellipticity of $P_{1,\theta}$, one can justify these formal constructions as in \cite{HeSj1}, and show that the eigenvalues $\rho_j$ of $P_\theta$ in a complex domain of the form $ [0, Ch]-i[0, \varepsilon]$ ($\varepsilon >0$ small enough, $C>0$ arbitrarily large) admit asymptotic expansions as $h\rightarrow 0_+$, of the form,
$$
\rho_j \sim he_j +h\sum_{k\geq 1}\rho_{j,k}h^{k/2},
$$
where $\rho_{j,k}\in \R$ and $e_j$ is the $j$-th eigenvalue of the
harmonic oscillator $-\Delta + \la V_2''(0)x,x\ra /2$. Moreover, still as in \cite{Ma2}, one can show that the width of $\rho_j$ is exponentially small, that is, there exists $S_j>0$ such that,
$$
\vert\Im\rho_j\vert ={\mathcal O}(e^{-S_j/h})
$$ uniformly as $h\rightarrow 0$. However, the techniques used in \cite{Ma2} do not permit to have a reasonably good estimate on the best value that one can take for $S_j$, that is, on the quantity
$$
{\bf S}_j := -\lim_{h\rightarrow 0}h\ln \vert\Im\rho_j\vert .
$$
 Better estimates can be found in \cite{Ma4}, but they are not optimal, as can be deduced from the works \cite{Ba, Na} where  optimal results are obtained in the one dimensional case.

The object of this paper is to compute the exact value of ${\bf S}_1$ (when $\mu$ is taken sufficiently small and $\tau$ is sufficiently small with respect to $\mu$), and show that ${\bf S}_1=S$ where $S$ is the imaginary part of some action integral along the complex bicharacteristics of $p_1(x,\xi):= \xi^2+x_n-\mu$ and $p_2(x,\xi):= \xi^2+ \tau V_2(x)$ (see (\ref{actionS}) below). Actually, our result is even more precise, since it gives a complete asymptotic expansion of $e^{S/h}\Im \rho_1 $ (see Theorem \ref{mainth}). As a by-product of the techniques, we also obtain the estimate ${\bf S}_j \geq S$ for all $j\geq 2$ (see Remark \ref{autreSj}).\vskip 0.2cm

Our strategy consists in first making some (global) symplectic change of variables, in order to recover the geometrical situation of \cite{GrMa} (where tunneling was in the position variables only), and then to adapt the argument used in \cite{GrMa}. However, since this symplectic change is quadratic in the dual variable $\xi$,   a great amount of 
new technical difficulties appears, due to the bad behavior of the symbols as  $|\xi|\to \infty$. In order to overcome them, it has been necessary for us to considerably improve the estimates of \cite{GrMa} and to introduce some kind of exotic pseudodifferential calculus.\vskip 0.5cm
In the next section (Section2), we describe the geometric objects that enter in the statement of our main result (Section 3). Section 4 is devoted to the reduction of the problem to a geometric situation close to that of \cite{GrMa}, and to the construction of a WKB solution. In order to compare this asymptotic solution with the actual resonant state, in Section 5 we develop a  generalization of Agmon estimates \cite{Ag}, adapted to the peculiarities of the reduced operator (in particular, its somehow exotic aspects from a pseudodifferential point of view). In Section 6, the comparison between asymptotic and actual solution is studied, in the different regions of the $x$-space. Finally, the proof is completed in Section 7. Besides, various annex results and technical proofs are given in the six appendices.
\vskip 0.2cm
{\bf Acknowledgments} The authors would like to thank Johannes Sj\"ostrand for many enlightening discussions on the subject.

\section{Geometrical Preliminaries}

This section is devoted to the construction of the broken instanton that is used to express  ${\bf S}_1$ as an action integral.

To start with, we work with the principal symbol $p_2(x,\xi)= \xi^2+\tau V_2(x)$ of $P_2$ near $(0,0)$. Since $V_2$ is analytic, we can consider the Hamilton field $H_{p_2}=(2\xi ,-\tau \nabla V_2(x))$ in a complex neighborhood of $(0,0)$ in $\C^{2n}$, and, due to the assumptions on $V_2$, $(0,0)$ is a fix point for $H_{p_2}$.
Moreover,  following \cite{HeSj1}, near $(0,0)$ we can define the two following (outgoing, respectively  incoming) complex Lagragian submanifolds of $p_2^{-1}(0)$,
\be
\label{varsortentr}
\Sigma_\pm =\{ (x, \pm i\nabla\varphi_2(x))\; ;\; x\in {\mathcal V}_0\}
\ee
where $\varphi_2 (x) = d_2(x,0)$, $d_2$ is the analytic extension of the Agmon distance associated to the degenerate metric $\tau V_2(x)dx^2$ on $\R^n$, and ${\mathcal V}_0$ is a complex neighborhood of $0$ in $\C^n$, sufficiently small so that $\varphi_2 (x)$ remains holomorpic on it.

Next, we set $p_1(x,\xi)= \xi^2+x_n -\mu$ and,
$$
\Gamma_\pm =\Sigma_\pm \cap p_1^{-1}(0),
$$
and we say that a pair of points $(\rho_+,\rho_-)\in 
\Gamma_+\times\Gamma_-$ are in $p_1${\it -correspondence}, 
if there exits $t\in\C$  such that $\rho_-=\exp tH_{p_1}(\rho_+)$ 
(that is, if $\rho_+$ and $\rho_-$ are on the same complex 
integral curve of $H_{p_1}$).

We have,
\begin{lemma}\sl
\label{corresp}
Under Assumption (\ref{assV2}), there exists a complex neighborhood of $(0,0)$ that,   
for each $\mu$ sufficiently  small and for $\tau\in (0,1]$, contains
a unique pair  $(\rho_+,\rho_-)\in \Gamma_+\times\Gamma_-$ in 
$p_1$-correspondence.
\end{lemma}
\begin{proof} By a local change of symplectic coordinates centered at $(0,0)$ (of the form $(x,\xi)\mapsto (x,\xi)+{\mathcal O}(\xi^2)$), we can reduce to $p_1(x,\xi) =x_n -\mu$ and $p_2(x,\xi)= \xi^2+\tau V_2(x)+{\mathcal O}(|x|\xi^2 + |\xi|^4)$, both even in $\xi$. In particular, the image of $\Sigma_\pm$ keeps the same form (\ref{varsortentr}) with ${\rm Hess}\hskip 1pt \varphi_2 (x) >0$ near 0 on the real. In these coordinates, $H_{p_1}$ just becomes $-\partial_{\xi_n}$, and thus the points $\rho_\pm$ are necessarily given by $\rho_\pm = ((x'_\mu, \mu) ;  \pm i\nabla \varphi_2(x'_\mu, \mu))$, where $x'_\mu \in \R^{n-1}$ is the unique solution (near 0) of $\partial_{x'}\varphi_2 (x'_\mu, \mu) =0$.
\end{proof}

Now, let us denote by $\Gamma_2^+$ the complex integral curve of 
$H_{p_2}$ containing $\rho_+$,
$\Gamma_2^-$ the complex integral curve of $H_{p_2}$ containing 
$\rho_-$, and
$\Gamma_1$ the complex integral curve of $H_{p_1}$ containing 
$\rho_-$ and $\rho_+$. In particular, $(0,0) \in \overline {\Gamma_2^\pm}$, 
and we can choose a simple oriented loop 
$\gamma =\gamma_2^+\cup\gamma_1\cup\gamma_2^-$, such that,
\begin{eqnarray*}
&&  \gamma_2^+\mbox{ is a path from } 
(0,0)\; {\rm to\;}
\rho_+ \mbox{ in } \Gamma_2^+;\\
&& \gamma_1\; \mbox{ is a path from }  \rho_+\; 
{\rm to\;}
\rho_-  \mbox{ in } \Gamma_1;\\
&& \gamma_2^-\; \mbox{ is a path from } 
\rho_-\; {\rm  to\;}
(0,0)  \mbox{ in } \Gamma_2^-.
\end{eqnarray*}
Therefore, $\gamma$ is included in the union of integral curves of 
$H_{p_1}$ and $H_{p_2}$, and for this reason we call it a {\it broken 
instanton}. Moreover, since the two Hamilton flows define canonical 
transformations, the action defined by
$$
I:=\int_\gamma \xi dx
$$
does not depend on the particular choice of $\gamma$, and we set
\be
\label{actionS}
 S:=\Im \int_\gamma \xi dx.
\ee

\section{Main Result}

Now, we can state our main result. In addition to the previous assumptions,
we also assume that $V_2(x)$ and the coefficients $c(x)$ of $R$ extend to bounded holomorphic function near a complex domain of the form,
$$
\widetilde{\mathcal S}_\delta := \{ x\in\C^n\; ;\; 
\vert \Im x\vert \leq\delta\la\Re x_n\ra^{1/2}\}
$$
for some $\delta >0$. In particular, by Cauchy formula, for any multi-index $\alpha$, they satisfy,
\be
\label{ass2}
\vert\partial^\alpha V_2(x)\vert +\vert\partial^\alpha c(x)\vert
={\mathcal O}\left(\la \Re x_n\ra^{-\vert\alpha\vert /2}\right)
\ee
uniformly on $\widetilde{\mathcal S}_\delta$. Moreover, we assume that the coefficient of interaction $r_{1,2}$
satisfies,
\be
\label{ellr}
r_{1,2}(\rho_+)\not= 0
\ee
where $\rho_+$ is defined in Lemma \ref{corresp}. Then, our main result is,
\begin{theorem}\sl
\label{mainth}\sl
Assume (\ref{assV2}), (\ref{ass2}) and
(\ref{ellr}). Then, for any $\mu >0$ small enough, there exists $\tau(\mu)>0$ such that, for $\tau\in (0,\tau (\mu)]$, the  lowest resonance
$\rho_1=e_1h+{\mathcal O}(h^2)$ of $P$ is such that,
$$
\Im \rho_1 = -h^{\frac32}f(h, \ln\frac1{h})e^{-2S/h}
$$
where  $f(h, \ln\frac1{h})$ admits an asymptotic expansion of the form,
$$
f(h, \ln\frac1{h})\sim \sum_{0\leq m\leq\ell}f_{\ell, m}h^\ell (\ln\frac1{h})^m,\quad  (h\rightarrow 0),
$$
with $f_{0,0} >0$, and  $S=\Im \int_\gamma \xi dx >0$ is defined in (\ref{actionS}).
\end{theorem}
\begin{remark}\sl \label{autreSj}
As in \cite{GrMa}, the same techniques also give upper bounds on the widths of the other resonances $\rho_j$, $j\geq 2$, of the type: $|\Im \rho_j|={\mathcal O}(h^{-\beta_j}e^{-2S/h})$, for some $\beta_j>0$. In particular, one has ${\mathbf S}_j\geq S$ for all $j$.
\end{remark}
\begin{remark}\sl The assumption that the matrix $R(x,hD_x)$ is off-diagonal has been made for the sake of simplicity only. Our proof can be adapted without any problem to the case of a general (symmetric) matrix, with the difference that the lowest resonance is now given by: $\rho_1=(e_1+r_{2,2}(0,0))h+{\mathcal O}(h^2)$.
\end{remark}

\section{Reduction and WKB Constructions}
\label{sectionWKB}

This section is devoted to the construction of
asymptotics solutions of the equation
$Pu=\rho u$. In the next sections, these
approximate solutions will be shown to be
close enough to the actual ones, so that they
can be used to find the asymptotics of
Im$\rho$.

First of all, we conjugate our operator $P$ by
the unitary Fourier-integral
operator,
\be
\label{conjug}
{\mathcal U}:={\mathcal F}^{-1}{\rm exp}\left[
-\frac{2i}{h}\left(
\frac13\xi_n^3+\sum_{k=1}^{n-1}
\xi_k^2\xi_n\right) \right]{\mathcal F},
\ee
where ${\mathcal F}$ stands for the usual
$h$-dependent Fourier transform. The
canonical transformation $\kappa:\; (x,\xi
)\mapsto (y,\eta )$ associated with
${\mathcal U}$ is given by,
\begin{eqnarray}
\label{TC} 
&&y_k=x_k+4\xi_k\xi_n\quad (k\leq n-1);
\nonumber\\ 
&&
y_n = x_n +2\xi^2 ;\\ 
&& \eta =\xi,
\nonumber
\end{eqnarray}
and thus, under the action of $\mathcal U$, the operator $P(x,hD_x )$ is transformed into
$\widetilde P(y,hD_y) ={\rm diag}\hskip 1pt (\widetilde P_1,\widetilde P_2) +h\widetilde R(y,hD_y)$, with diagonal principal
symbol $\widetilde p (y,\eta )={\rm
diag}\left(\widetilde p_1(y,\eta ),\widetilde
p_2(y,\eta )\right)$, given by,
$$
\widetilde p_j (y,\eta ):=p_j(y'-4\eta_n\eta', y_n -2\eta^2 ; \eta).
$$
In particular,
\begin{eqnarray}
\label{symb} 
&&\widetilde p_1(y,\eta )=-(\eta^2 -y_n+\mu );
\nonumber\\ 
&&
\widetilde p_2(y,\eta )=p_2(y,\eta )+{\mathcal
O}\left(\vert (y,\eta )\vert ^3\right);\\  
&& \widetilde
p_2(y,-\eta )=\widetilde p_2(y,\eta ).
\nonumber
\end{eqnarray}
Moreover, thanks to (\ref{ass2}), we observe that,
for any compact set $K\subset\subset\R^n$, the full symbol of 
$\widetilde P$ extends, for $x\in K$, to an holomorphic function with
respect to $\eta$ in a complex strip around $\R^n$, and that its
derivatives of any order remain ${\mathcal O}(\la\eta\ra^2)$ on this strip.

Then, except from the minus sign in front of
$\widetilde p_1$, we see on (\ref{symb}) that we are in a situation very
similar to that studied by in \cite{GrMa}
(see also \cite{Pe}), namely, the transversal crossing 
between
two potentials such that one of them admits a
non degenerate point-well, and the two corresponding classically allowed regions are disjoints.

As in \cite{GrMa}, the starting point of the
construction consists in the WKB asymptotics
given near the well $y=0$ by the method of
Helffer and Sj\"ostrand \cite{HeSj1}. More precisely,
because of the matricial nature of the
operator and the fact that $\widetilde p_1$ is
elliptic above $y=0$, one finds a formal
solution $w_1$ of $\widetilde Pw_1=\rho w_1$ of the form,
\be
\label{BKWinit}
w_1(y;h)=\left(\begin{array}{clcr} 
ha_1(y,h) \\ 
a_2(y,h)\end{array}\right) e^{-\widetilde\varphi_2
(y)/h}
\ee
where $\widetilde\varphi_2 (y)$ is the (real and non
negative on the real domain) solution of $ \widetilde
p_2(y,i\nabla\widetilde \varphi_2 (y))=0\; ,\;
\widetilde\varphi_2 (0)=0$, defined by,
$$
\kappa 
(\Gamma_+)=\{ (y,i\nabla\widetilde\varphi_2(y))\; ;\; y {\rm\;
close\; to\;} 0\}.
$$
(In particular, we also have $\nabla\widetilde\varphi_2(0)=0$.) Moreover, the two $a_j$'s ($j=1,2$) are classical symbols of order
$0$ in $h$, that is,  they are formal series
in $h$ of the form,
\be
\label {AE1}
a_j(y,h) = \sum_{k=0}^\infty h^ka_{j,k}(y)
\ee
with $a_{j,k}$ smooth near 0, and $a_2$ is
elliptic in the sense that $a_{2,0}$ never
vanishes. 

It is easy to see that the above
constructions can be continued along the
integral curves of the (real) vector field
$\partial_\eta\widetilde
p_2(y,i\nabla\varphi_2)D_y=
(2\nabla\varphi_2(y)+
{\mathcal O}(y^2)).\nabla_y$,
as long as
$\widetilde p_1 (y,i\nabla 
\widetilde\varphi_2  (y))$ does not vanish.
We set,
\begin{eqnarray}
\label{inter} 
&&\Sigma_2^+=\{ (y, i\nabla
\widetilde\varphi_2(y))\; ;
\;  y\in {\mathcal V}\}
\nonumber\\ 
&&
\Gamma = \Gamma_\mu = \Sigma_2^+\cap \widetilde
p_1^{-1}(0)
\end{eqnarray}
where ${\mathcal V}$ is a fixed
($\mu$-independent) sufficiently small
neighbourhood of 0 in $\C^n$. Then $\Gamma$
is a complex $\C$-isotropic submanifold
of $\C^{2n}$ (that is, with respect to the
complex symplectic form
$\sigma = \sum d\eta_k \wedge {}dy_k$) of
dimension $n-1$, and standard arguments of
symplectic geometry show that if ${\mathcal T}$ is
a sufficiently small $\mu$-independent
neighborhood of 0 in
$\C$, then the set:
\be
\label{Sigma1}
\Sigma_1=\{ \exp tH_{\widetilde
p_1}(y,\eta )\; ; (y,\eta )\in \Gamma\; ,\;
t\in {\mathcal T}\}
\ee
is a complex Lagrangian manifold. Since
the projection
$\Pi_y\Gamma$ of
$\Gamma$ on  the
base $\{\eta =0\}$ has an equation of the form
$y_n= f(y')=\mu -\delta\mu^2 -f_1(y')+{\mathcal O}(\vert
y'\vert^3 +\mu^3)$ with $f_1(0)=0$, $f_1(y')\geq \delta'|y'|^2$ (where
$y'=(y_1,...,y_{n-1})$ and $\delta, \delta' >0$ do
not depend on $\mu$), an easy
computation shows that, for
$\mu >0$ small enough,
$\Pi_yH_{\widetilde p_1}$ keeps transversal to 
$\Pi_y\Gamma$, as long as $|y'| << 1$. Moreover, the only points of $\Sigma_1$ where $H_{\widetilde
p_1}$ becomes vertical are above $\{ y_n =\mu\}$, and $\Pi_y\Gamma\cap \R^n
\subset
\{ y_n <\mu\}$. 
As a consequence, there exists
a complex neighborhood ${\mathcal W}_\mu$ of
$\Gamma$ such that $\Pi_y{\mathcal W}_\mu\cap\R^n$ is of the form,
$$
\Pi_y{\mathcal W}_\mu\cap\R^n=\{ -\delta_0 <y_n<g(y')\, , \, |y'|<\delta_0\},
$$
with $\delta_0 >0$ independent of $\mu$, and $g(y') =\mu +{\mathcal O} (|y'-y'_\mu|^2)$ (for some $y'_\mu ={\mathcal O}(\mu^2)\in\R^n$), and such that 
$\Sigma_1\cap {\mathcal W}_\mu$ projects
bijectively on the base.
This
implies the existence of an analytic function
$\widetilde
\varphi_1$ defined on $\Pi_y{\mathcal W}_\mu$, such that,
\be
\label{phi1}
\Sigma_1\cap {\mathcal W}_\mu=\{ (y, i\nabla
\widetilde\varphi_1(y))\; ;
\;  y\in \Pi_y{\mathcal W}_\mu \}.
\ee
Moreover, $\widetilde\varphi_1$ can be normalized
by requiring that $\widetilde\varphi_1=\widetilde\varphi_2$ on
$\Pi_y\Gamma$, so that we then have,
\be
\widetilde\varphi_1=\widetilde\varphi_2\quad {\rm
and}\quad
\nabla\widetilde\varphi_1 = \nabla\widetilde\varphi_2
\quad {\rm on}\quad \Pi_y\Gamma .
\ee
Using that $\widetilde\varphi_2$ is real on
the real, and $\widetilde p_1(y,\eta )$ is even with
respect to $\eta$,  we also see that,
\be
\widetilde \varphi_1 \quad {\rm is\; real\; on\;
the\; real}.
\ee
Finally, one can easily check that the quantity,
$$
S(\mu):= \widetilde\varphi_1 (y'_\mu,\mu),
$$
is ${\mathcal O}(\mu^2)$, and that, on $\Pi_y{\mathcal W}_\mu\cap\R^n$, one has $\widetilde\varphi_1 \leq \widetilde\varphi_2$. Moreover, since $\widetilde\varphi_2(y)\geq |y|^2/C$ with some constant $C>0$, for $\mu$ small enough, we have
$$
\widetilde\varphi_2(y)\sim 1 >>S(\mu) \mbox{ on } \partial\left( \Pi_y{\mathcal W}_\mu\cap\R^n\right)\cap\{y_n < g(y')\}.
$$
Then, one can deduce,
$$
\widetilde\varphi_1(y)>S(\mu) \mbox{ on } \partial\left( \Pi_y{\mathcal W}_\mu\cap\R^n\right)\cap \{ f(y')\leq y_n <g(y')\}.
$$
On the other hand,
one can see as in \cite{HeSj2}, Lemma 10.2, that, for $y'\not= y'_\mu$, one has,
$$
\widetilde\varphi_1(y', g(y')) >S(\mu).
$$
In addition, a local study of $H_{p_1}$ shows that $(\Pi_y\Sigma_1)\cap\R^n$ has a contact of order exactly 2 with $\{y_n=\mu\}\cap \R^n$ at $(y'_\mu,\mu)$, and one can deduce the existence of a constant $C=C(\mu)>0$ such that $C^{-1}|y'-y'_\mu|^2\leq \mu -g(y')\leq C|y'-y'_\mu|^2$.
\begin{remark}\sl Let us observe that, by standard symplectic arguments, the quantity $S(\mu)$ is nothing but the constant $S$ defined in (\ref{actionS}).
\end{remark}

\vskip 0.2cm
As in \cite{GrMa}, the extension of the WKB constructions
(\ref{BKWinit}) accross $\Gamma$ are obtained by taking a formal ansatz of the form,
\be
\label{BKWcross}
w_2(y;h)=\sum_{k\geq 0}h^k\left( \alpha_k(y,h)
Y_{k,0}\left(\frac{z(y)}{\sqrt h}\right)
+{\sqrt h}\beta_k(y,h)Y_{k,1}
\left(\frac{z(y)}{\sqrt h}\right)\right)
e^{-\psi (y)/h},
\ee
where,
\be
\alpha_k(y,h)=\left(\begin{array}{clcr} 
h\alpha_{k,1}(y,h) \\ 
\alpha_{k,2}(y,h)\end{array}\right)
 \quad ;\quad
\beta_k(y,h)=\left(\begin{array}{clcr} 
\beta_{k,1}(y,h) \\ 
h\beta_{k,2}(y,h)\end{array}\right) ,
\ee
$\alpha_{k,j}$ and $\beta_{k,j}$ are formal
symbols of the form,
\be
\label{symbcross}
\sum_{l\geq 0}\sum_{m=0}^l h^l({\rm ln}h)^m
\gamma^{l,m}(y)
\ee
(with $\gamma^{l,m}$ analytic near $\Gamma$), the function $\psi$ is defined as,
\be
\psi (y):= \frac12\left(\widetilde\varphi_2(y)+
\widetilde\varphi_1(y)\right) ,
\ee
the function $z(y)$ is the only analytic function defined near $\Gamma$ such that,
\begin{eqnarray}
\label{z}
&& z(y)^2=2\left(\widetilde\varphi_2(y)-
\widetilde\varphi_1(y)\right) \nonumber\\
&& z(y)<0\quad {\rm on}\quad
\Gamma_-:=\{y_n<f(y')\; ; \; y\in \Pi_y
{\mathcal W}_\mu\cap\R^n\},
\end{eqnarray}
and for any $k\geq 0$ and $\varepsilon \in\C$, the function $Y_{k,\varepsilon}$ is the so-called Weber function, defined by,
\be
Y_{k,\varepsilon}(z)=\partial_\varepsilon^k
Y_{0,\varepsilon}(z)
\ee
where $Y_{0,\varepsilon}$ is the unique 
entire function with respect to $\varepsilon$ and $z$,
solution of the Weber equation,
\be
Y_{0,\varepsilon}''+(\frac12 -\varepsilon
-\frac{z^2}4)Y_{0,\varepsilon}=0
\ee
such that,  for $\varepsilon >0$, one has,
\be
Y_{0,\varepsilon}(z)\sim
\frac{\sqrt {2\pi}}{\Gamma (\varepsilon )}
e^{z^2/4}z^{\varepsilon -1}
\qquad (z\rightarrow +\infty ).
\ee
Then, a resummation of
(\ref{BKWcross}) is possible up to an error
of order ${\mathcal O}(h^\infty e^{-  \phi
/h})$, with,
\begin{eqnarray}
\label{phi}
&& \phi = \widetilde\varphi_2\quad {\rm in} \quad
\Gamma_- \, ;\nonumber\\
&& \phi = \widetilde\varphi_1= \widetilde\varphi_2\quad {\rm on}
\quad\Gamma \, ;\\
&& \phi = \widetilde\varphi_1\quad {\rm in} \quad
\Gamma_+:=\{ y_n> f(y')\; ; \; y\in \Pi_y
{\mathcal W}_\mu\cap\R^n\} .\nonumber
\end{eqnarray}
Actually, looking a little bit more carefully at the results of \cite{Pe} (in particular Proposition A.2 and the proof of Lemma 3.3), we see that the resummation can actually be done in a small enough complex neighborhood $\Omega$ of $\Gamma$, under the condition that $\Omega\subset \{  y\in\C^n\,;\, |\Im z(y)| \leq C\sqrt h\}$ (where $C>0$ is an arbitrary constant). In that case, the resummation is valid up to ${\mathcal O}(h^\infty e^{- \Re \phi
/h})$ error-terms. In what follows, we will use the notations,
$$
\psi_\delta (y):= \psi (y+i\delta)\quad ;\quad \phi_\delta (y):= \phi (y+i\delta),
$$
where $\delta =\delta (h) ={\mathcal O}(\sqrt h)\in \R^n$.
\vskip 0.2cm
For $\nu_0>0$ and $g\in C^\infty (\R^{2n}\; ;
\R_+)$,
we denote by
$S_{\nu_0}(g(x,\xi ))$ the set of (possibly $h$-dependent)
functions $p\in C^\infty (\R^{2n})$ that
extend to holomorphic functions with respect to
$\xi$  in the strip,
$$
{\mathcal A}_{\nu_0}:=\{(x,\xi )\in  \R^n\times\C^n\; ;\; 
\vert\Im\xi\vert < \nu_0\},
$$
and such that,
 for all $\alpha\in\N^{2n}$, one has,
\be
\partial^\alpha p(x,\xi )
={\mathcal O}(g(x,{\rm Re}\xi )),
\ee
uniformly with respect to $(x,\xi )\in
{\mathcal A}_{\nu_0}$ and $h>0$ small enough. We also denote by $S_0(g)$ the analogous  space of smooth symbols obtained by substituting $\R^{2n}$ to ${\mathcal A}_{\nu_0}$.
\vskip 0.2cm
If $a$ and $b$ are (scalar) formal
symbols of the type (\ref{symbcross}), and $k\in\N$, we
set,
\be
I_k(a,b)(y\; ;\; h)= a(y\; ;\;  h)Y_{k,0}\left(
\frac{z(y)}{\sqrt h}\right) +b(y\; ;\;  h)Y_{k,1}\left(
\frac{z(y)}{\sqrt h}\right).
\ee 
For $M\in\Z$ and $\Omega\subset\R^n$ open,
we also consider the space of sequences of formal
symbols,
\begin{eqnarray*}
S^M(\Omega ) := \{ a = (a_k)_{k\in\N}\; ;
a_k (y,h)=\sum_{l=-M}^\infty\sum_{m=0}^l
h^{l}({\rm ln}h)^m\gamma_k^{l,m}(y)\; ;\\
\gamma_k^{l,m}\in C^\infty (\Omega )\} .
\end{eqnarray*}
and, for $a,b\in S^M(\Pi_y
{\mathcal W}_\mu )$, we set,
\be
I(a,b)e^{-\psi_\delta /h}:=
\sum_{k\geq 0}h^k I_k(a_k,{\sqrt{h}}b_k)
e^{-\psi_\delta /h} .
\ee

Then, the arguments of \cite{GrMa}, Section4, give a natural way to make act ${\rm Op}_h^W(p)$ on expressions of the type,
\be
\label{typew}
w=\left(\begin{array}{clcr} 
I(h\alpha_1,\beta_1) \\ 
I(\alpha_2,h\beta_2)\end{array}\right)
e^{-\psi_\delta /h},
\ee
where $\alpha_j=(\alpha_{j,k})_{k\geq 0}$
and $\beta_j=(\beta_{j,k})_{k\geq 0}$ are in
$S^0(\Pi_y{\mathcal
W}_\mu )\;$ ($j=1,2$). Moreover, this action is well defined up to error terms that are exponentially smaller than $e^{\Re \phi_\delta /h}$.
\vskip 0.2cm

Now, denote by $p$
the full symbol of $ P$, and recall the definition $\widetilde P := {\mathcal U}P{\mathcal U}^{-1}$, where ${\mathcal U}$ is given in (\ref{conjug}). In order to be able to make act $\widetilde P$, too,  on such expressions, we first show,
\begin{lemma}\sl 
\label{lemmfond}
Let $c=c(x)$ satisfying (\ref{ass2}). Then,
for any $\chi =\chi (y_n)\in C_0^\infty (\R)$, there exists $\nu_0>0$ such that the operator $\chi \, {\mathcal U}\, c\, {\mathcal U}^{-1}$ can be written as,
$$
\chi \, {\mathcal U}\, c\, {\mathcal U}^{-1} ={\rm Op}_h^W (\check c + r),
$$
with $\check c\in S_{\nu_0}(1)$, $r\in S_0(1)$,  and, 
$$
r ={\mathcal O}(e^{-\nu_0 /h}) \mbox{ in }  S_0(1).
$$
Moreover, the symbol $\check c$ satisfies,
$$
\check c (y,\eta ) = \chi (y_n) c(y'-4\eta_n\eta', y_n-2\eta^2 ) +{\mathcal O}(h) \mbox{ in } S_{\nu_0}(1).
$$
\end{lemma}
\begin{proof} 
A straightforward computation leads to,
\be
\label{straightf}
\chi \, {\mathcal U}\, c\, {\mathcal U}^{-1} ={\rm Op}^L_h(a),
\ee
with,
\begin{eqnarray*}
a(y,\eta )&=&\frac1{(2\pi h)^n}\int \chi (y_n)e^{i(y-x)(\xi-\eta)/h}\times \\
&&\times c(x' -2\xi_n(\xi'+\eta'),x_n-\frac23(\eta_n^2+\eta_n\xi_n+\xi_n^2) -2|\eta'|^2)dxd\xi.
\end{eqnarray*}

In this integral, we can perform the change of contour of integration,
$$
\R^n \ni x\mapsto x-i\delta \frac{\xi -\eta}{\la \xi -\eta\ra}\in\C^n,
$$
where $\delta >0$ is a small enough constant. Setting $f=(f',f_n):= \delta (\xi -\eta)/\la \xi -\eta\ra$ and,
$$
F(x,\eta,\xi):=(x' -if'-2\xi_n(\xi'+\eta'),x_n-if_n-\frac23(\eta_n^2+\eta_n\xi_n+\xi_n^2) -2|\eta'|^2),
$$
we obtain,
\begin{eqnarray*}
a(y,\eta )&=&\frac1{(2\pi h)^n}\int \chi (y_n)e^{i(y-x)(\xi-\eta)/h - \frac{\delta}{h} |\xi -\eta|^2/\la\xi -\eta\ra} c(F(x,\eta,\xi))dxd\xi\\
&=&\frac1{(2\pi h)^n}\int \chi (y_n)e^{i(y-x)\xi/h - \frac{\delta}{h} |\xi |^2/\la\xi \ra} c(F(x,\eta,\xi+\eta))dxd\xi,
\end{eqnarray*}
and thus, for any fix $\chi_0\in C_0^\infty (\R^n;[0,1])$ such that $\chi_0=1$ near 0, 
\begin{eqnarray*}
a(y,\eta )=\frac1{(2\pi h)^n}\int\chi (y_n)\chi_0(\xi)e^{i(y-x)\xi/h - \frac{\delta}{h} |\xi |^2/\la\xi \ra} c(F(x,\eta,\xi+\eta))dxd\xi \\
+ r_1(y,\eta),
\end{eqnarray*}
where, by means of integrations by parts in $\xi$,  the symbol $r_1$ can be seen to satisfy,
$$
\partial^\alpha r_1 ={\mathcal O}(e^{-\delta' /h})  \mbox{ in } S_0(1),
$$
for some $\delta' >0$ independent of $\mu$ and $h$. 

Then, we split the previous integral into,
$$
I_1+I_2+I_3:=\int_{|x_n|\leq 2C} +\int_{2C\leq |x_n|\leq 2C +\la\eta\ra^2} + \int_{|x_n|\geq 2C +\la\eta\ra^2},
$$
where $C>0$ is chosen sufficiently large in order to have $|y_n-x_n|\geq |x_n|/2$ when $|x_n|\geq 2C$ and $y_n\in \mbox{Supp }\chi$.

On $\{|x_n|\leq 2C +\la\eta\ra^2\}\cap {\rm Supp} \chi_0$, we observe that $|\Re F_n| \sim  \la\eta\ra^2$. Therefore, if $\chi_1\in C_0^\infty(\R^n;[0,1])$ is such that $\chi_1=1$ near 0 and Supp$\chi_1 \subset \{ \chi_0=1\}$, in this region we can perform the change of contour of integration,
$$
\R^n \ni\xi \mapsto \widetilde \xi:=\xi + i\delta' \chi_1(\xi)\frac {y-x}{\la y-x\ra},
$$
that gives us,
\begin{eqnarray*}
&& I_1=\chi (y_n)\int_{|x_n|\leq 2C} {\mathcal O}(\chi_0(\xi)h^{-n}e^{-\frac{\delta'\chi_1(\xi)|y-x|^2}{2h\la y-x\ra} -\frac{\delta'|\xi|^2}{h\la\xi\ra}})c(F(x,\eta,\widetilde\xi +\eta))dxd\xi ;\\
&& I_2=\chi (y_n)\int_{2C\leq |x_n|\leq 2C +\la\eta\ra^2} {\mathcal O}(\chi_0(\xi)h^{-n}e^{-\frac{\delta'\chi_1(\xi)|y-x|^2}{2h\la y-x\ra} -\frac{\delta'|\xi|^2}{h\la\xi\ra}})\\
&& \hskip 8cm\times c(F(x,\eta,\widetilde\xi +\eta))dxd\xi,
\end{eqnarray*}
and thus, using Assumption (\ref{ass2}), we see that, for $\nu_0>0$ sufficiently small, we have $I_1\in S_{\nu_0}(1)$. Moreover, a stationary-phase expansion shows that,
$$
I_1 =\chi(y_n)c(F(y,\eta,\eta))+{\mathcal O}(h) =\chi(y_n)c(y'-4\eta_n\eta', y_n-2\eta^2)+{\mathcal O}(h),
$$
where the estimates ${\mathcal O}(h)$ hold in the space $S_{\nu_0}(1)$.

On the other hand, since $|x-y|\geq |x_n-y_n|\geq |x_n|/2$ for $|x_n|\geq 2C$ and $y_n\in \mbox{Supp }\chi$, we see that,
$$
I_2 ={\mathcal O}(e^{-\nu_1/h}) \mbox{ in } S_0(1),
$$
for some $\nu_1 >0$.

Finally, in the expression of $I_3$, we perform the change of contour of integration,
$$
\R^n \ni\xi \mapsto \hat\xi:= \xi + i\delta \chi_1(\xi)\frac {y-x}{\la y-x\ra\la \eta\ra},
$$
and since,  in the relevant region, $|y-x|\geq |x_n-y_n| \geq |x_n|/2\geq C+\la\eta\ra^2/2$, this leads to,
\begin{eqnarray*}
I_3 =\chi (y_n)\int_{|x_n|\geq 2C+\la\eta\ra^2} {\mathcal O}(\chi_0(\xi)h^{-n}e^{-\frac{\delta'\chi_1|x|}{h\la\eta\ra}-\frac{\delta'\chi_1\la\eta\ra}{h} -\frac{\delta'|\xi|^2}{h\la\xi\ra}})\\
\times V_2(F(x,\eta,\hat\xi +\eta ))dxd\xi,
\end{eqnarray*}
and thus,
$$
I_3 ={\mathcal O}(e^{-\delta' /h}) \mbox{ in } S_{0}(1).
$$
Therefore, we have proved the result of the lemma, except from the fact that we have worked with the left-quantization instead of the Weyl-one. But the passage from one quantization to the other one can be done in a standard way (see, e.g., \cite{Ma6}), and
this completes the proof of the lemma.
\end{proof}
As a corollary, and since ${\mathcal U}\,hD_x\,{\mathcal U}^{-1}=hD_y$, we immediately obtain (denoting by $ {\mathcal M}_2(E)$ the space of $2\times 2$ matrices with coefficients in $E$),
\begin{proposition}\sl Assume (\ref{ass2}), and denote by $T_\delta$ the complex translation given by $T_\delta u(y):=u(y+i\delta)$. Then,
for any $\chi =\chi (y_n)\in C_0^\infty (\R)$, there exists $\nu_0>0$ (independent of $\mu$) such that the operator $\chi \widetilde P^\delta:= \chi T_\delta \widetilde PT_\delta^{-1}$ can be written as,
$$
\chi \widetilde P^\delta ={\rm Op}_h^W (\check p_\delta + r_0),
$$
with $\check p_\delta\in {\mathcal M}_2\left( S_{\nu_0}(\la\eta\ra^2)\right)$, $r_0\in {\mathcal M}_2\left( S_0(\la\eta\ra^2)\right)$,  and, 
$$
r_0 ={\mathcal O}(e^{-\nu_0 /h}) \mbox{ in } {\mathcal M}_2\left( S_0(\la\eta\ra^2)\right).
$$
Moreover, the symbol $\check p_\delta$ satisfies,
$$
\check p_\delta (y,\eta ) = \chi (y_n) p(y'+i\delta'-4\eta_n\eta', y_n+i\delta_n-2\eta^2 ; \eta ) +{\mathcal O}(h) \mbox{ in } {\mathcal M}_2\left( S_{\nu_0}(\la\eta\ra^2)\right).
$$
\end{proposition}

Consequently, for $\mu$ sufficienly small, and up to error terms that are exponentially smaller that $e^{-\Re\phi_\delta /h}$, the operator $\widetilde P^\delta$  naturally acts 
on expressions of the type (\ref{typew}).
\vskip 0.2cm
Then, proceeding as in \cite{GrMa} (but in a quite simpler geometric situation, here), we finally obtain the existence of a 
$\mu$-independent neighborhood ${\mathcal V}$ of 0 in
$\R^n$, such that, if one sets,
\begin{eqnarray}
&& {\mathcal V}_\mu^-={\mathcal V}\cap\{y_n < f(y')\};
\nonumber\\
&&
{\mathcal V}_\mu^+={\mathcal V}\cap\{f(y') <y_n <
g(y')\},
\end{eqnarray}
then, we have,
\begin{proposition}\sl
\label{BKWtotal}
There exists a neighborhood $\Omega_\mu$ of
$\Pi_y\Gamma$ and symbols,
\begin{eqnarray}
&& a_1,a_2\in S^0({\mathcal
V}_\mu^-)\; ,\quad a_2\; 
elliptic\, ;\nonumber\\ &&
\alpha_1,\alpha_2,\beta_1,\beta_2
\in S^0(\Omega_\mu ) \; , \quad\alpha_2\; and\;
\beta_1\; elliptic\, ;\\ 
&& b_1,b_2\in S^0({\mathcal
V}_\mu^+)\; ,\quad  b_1\; elliptic,\nonumber
\end{eqnarray}
such that, if one sets,
\begin{eqnarray*}
 &&w_1=\left(\begin{array}{clcr} 
ha_1 \\ 
a_2\end{array}\right) e^{-\widetilde\varphi_2
/h}\quad ;\quad
w_2=\left(\begin{array}{clcr} 
I(h\alpha_1,\beta_1) \\ 
I(\alpha_2,h\beta_2)\end{array}\right)
e^{-\psi /h}\quad ;\\
&&
w_3={\sqrt {2\pi
h}}\left(\begin{array}{clcr}  b_1 \\ 
hb_2\end{array}\right) 
e^{-\widetilde\varphi_1
/h},
\end{eqnarray*}
then, one  formally has,
$$
\widetilde Pw_1=\rho w_1\quad ; \quad
\widetilde Pw_2=\rho w_2\quad ; \quad
\widetilde Pw_3=\rho w_3 .
$$
Moreover, if $\widetilde w_1$, $\widetilde w_2$ and
$\widetilde w_3$ denote resummations of $w_1$, $w_2$
and $w_3$ respectively, one has,
\begin{eqnarray}
&&
\widetilde w_1 - \widetilde w_2 ={\mathcal O}(h^\infty
e^{-\widetilde\varphi_2
/h})\quad in\; 
{\mathcal
V}_\mu^-\cap \Omega_\mu \\
&& \widetilde w_2 - \widetilde w_3 ={\mathcal O}(h^\infty
e^{-\widetilde\varphi_1
/h})\quad in\; 
{\mathcal
V}_\mu^+\cap \Omega_\mu .
\end{eqnarray}
\end{proposition}
Now, let $(\chi_1,\chi_2,\chi_3)$ be a
partition of unity adapted to $({\mathcal
V}_\mu^-,\Omega_\mu ,\widetilde {\mathcal
V}_\mu^+)$ with $\widetilde {\mathcal
V}_\mu^+\subset\subset {\mathcal
V}_\mu^+$ and $\chi_3\in C_0^\infty (
{\mathcal
V}_\mu^+)$, and
 set,
\be
\widetilde U=\sum_{j=1}^3\chi_j \widetilde w_j \in C_0^\infty ({\mathcal V}
\cap \{y_n<g(y')\}).
\ee
Then, the arguments in \cite{GrMa} show that,
\be
\widetilde P\widetilde U = \rho \widetilde U +{\mathcal
O}(h^\infty e^{-\phi /h})
\quad {\rm uniformly\; in}\;
\{ \chi_1 +\chi_2 +\chi_3 =1\} .
\ee

Finally, the caustic set ${\mathcal C}:=\{ y_n =g(y')\}$ can be crossed over, too, by using, as in 
\cite{HeSj2} Section 10, an Airy representation near $\mathcal C$, 
and by using the 
stationary phase method in ${\mathcal V}\cap\{y_n> g(y')\}$ in order to recover
 there an expression of the type ${\sqrt {2\pi
h}}\left(\begin{array}{clcr}  c_1 \\ 
hc_2\end{array}\right) 
e^{\widetilde\varphi
/h}$
where the symbol $c_1$ is elliptic and $\widetilde\varphi$ is a complex-valued
 analytic function that satisfies $(\nabla\widetilde\varphi)^2=\mu - y_n$,
 $\Re\widetilde \varphi
=Cte=S(\mu)$ on $\gamma_\mu:=\{ \exp tH_{\widetilde p_1}(y'_\mu, \mu)\, ;\, t\in \R \mbox{ small enough}\}$, $\Re\widetilde \varphi (y)\geq S(\mu) + 
\delta_1{\rm dist\,}(y,\gamma_\mu)^2$ for some $\delta_1>0$. In particular, one also has,
\be
\vert\partial_{y_n}\Im\widetilde\varphi (y'_\mu,y_n)\vert =\sqrt{y_n -\mu}.
\ee
By using an additional cut-off function $\chi_4$ near $(0,\mu)$, and 
by shrinking
a little bit (but still in a $\mu$-independent way)  $\mathcal V$ around 0,
we finally obtain a function, 
$$
{\mathbf w}=\sum_{j
=1}^4\chi_j\widetilde w_j \in
C_0^\infty ( {\mathcal V}\cap \{y_n< g(y') + 2\varepsilon\}
$$
 ($\varepsilon = 
\varepsilon (\mu) >0$),
such that, setting $\widetilde{\mathcal V}_1:= \{ \chi_1+\chi_2+\chi_3+\chi_4=1\}$ and $\widetilde{\mathcal V}_2:= \bigcup_{1\leq j\leq 4}{\rm Supp}( \chi_j)$, one has,
\begin{eqnarray}
\label{eqW1}
\widetilde P {\mathbf w} &=& \rho {\mathbf w} + {\mathcal O}(h^\infty e^{-\varphi /h})\quad{\rm on\quad}
\widetilde{\mathcal V}_1\cap \{y_n< g(y') + \varepsilon\}\, ;\\
\label{eqW2}
\widetilde P {\mathbf w} &=& \rho {\mathbf w} + {\mathcal O}(h^{-N_0} e^{-S(\mu ) /h})\quad{\rm on\quad}
\widetilde{\mathcal V}_2\cap \{y_n\geq g(y')\},
;\\
\label{eqW3}
\widetilde P {\mathbf w} &=& \rho {\mathbf w} + {\mathcal O}(h^{-N_0} e^{-\varphi /h})\quad{\rm on\quad}
\widetilde{\mathcal V}_2\backslash\widetilde{\mathcal V}_1,
\end{eqnarray}
with $\varphi = \widetilde\varphi_2$ in $ \widetilde{\mathcal V}_2\cap \{y_n \leq f(y')\}$,
$\varphi = \widetilde\varphi_1$ in $ \widetilde{\mathcal V}_2\cap \{ f(y')\leq y_n \leq g(y')\}$,
$\varphi =\Re \widetilde\varphi$ in $ \widetilde{\mathcal V}_2\cap \{ g(y')\leq y_n 
\leq g(y')+\varepsilon\}$, and $N_0\geq 0$.
\vskip 0.2cm
Observe that, by construction,  $\varphi > S(\mu)$ on $( \widetilde{\mathcal V}_2\backslash  \widetilde{\mathcal V}_1)\cap \{ y_n\leq g(y')\}$.

\section{Agmon estimates}

In order to estimate the width of the resonance, we need
to prove that the asymptotic solution ${\mathbf w}$ constructed in
 the previous section gives
a sufficiently good approximation of the true resonant state ${\mathbf u}$, specially
in the region $\{y_n > \mu\}$. 
\vskip 0.2cm
For this, we would like to adopt the usual strategy of Agmon estimates (see \cite{Ag}), 
but, since we work with operators that are not differential (even at the level of principal symbols), we have to be 
more careful. In particular, our weight functions need to be smooth (and not 
only Lipschitz like in the case of a Schr\"odinger operator). Moreover, 
these weight functions need to be very close (at least up to $O(h)$) to
the $C^{1,1}$ weight $\phi$ introduced in the previous section. This lead 
to estimates of the type $\partial^\alpha \widetilde\phi =O(h^{-(\vert\alpha\vert
 -2)_+})$ for such functions, and also estimates of the type 
$\partial_x^\alpha \partial_\xi^\beta a(x,\xi ) 
=O(h^{-(\vert\alpha\vert -1)_+})$ for the related symbols. 

Thus, we first have to establish some preliminary results for such objects.

\subsection{Preliminaries}
In the sequel, we denote by $\lambda_h$ a positive function of $h$ such that,
$$
\lambda_h\geq 1\quad ;\quad  h\lambda_h={\mathcal O}(1) \, \mbox{ as } h\rightarrow 0_+.
$$
For $\Omega\subset\R^n$, we also denote by ${\bf 1}_\Omega$ the characteristic function of $\Omega$.

\begin{proposition}\sl
\label{poidsvar}
Let  $\nu_0 >0$, $m\geq 0$, $a=a(y,\eta )\in S_{\nu_0}(\langle\eta\rangle^{2m})$, and $\Omega\subset\R^n$.
For $h>0$ small enough, let also
$\Phi_h\in C^\infty (\R^n)$ real valued, such that,
\be
\label{estreg1}
 \sup\vert\nabla\Phi_h
\vert < \nu_0,
\ee
and, for 
 any multi-index $\alpha\in\N^n$ with $|\alpha|\geq 2$,
\be
\label{estreg}
\partial^\alpha\Phi_h (y)= {\mathcal O}\left((1+\lambda_h{\bf 1}_\Omega (y) ) h^{2-\vert
\alpha\vert} \right),
\ee
uniformly for $y\in\R^n$ and $h>0$ small enough. 
Then,  for any $\widetilde\Omega\subset \R^n$ with dist$(\Omega ,\R^n\backslash \widetilde\Omega )>0$,  the operator $e^{\Phi_h /h}Ae^{-\Phi_h /h}:=e^{\Phi_h /h} {\rm Op}_h^W(a)e^{-\Phi_h /h}$ can be split into ,
\be
e^{\Phi_h /h}Ae^{-\Phi_h /h}={\rm Op}_h^W\left( a_0(y,\eta +i\nabla\Phi_h(y)
\right)
+hB+R
\ee
where  $a_0$ is the principal
symbol of  $A={\rm Op}_h^W(a)$, and the two operators $B$ and $R$ satisfy,
\begin{eqnarray}
\label{estA}
&&\vert\langle Bu,u\rangle\vert \leq C_1\left(  \Vert\langle hD_y\rangle^m u\Vert^2+\lambda_h \Vert\langle hD_y\rangle^m u\Vert^2_{L^2(\widetilde\Omega)}\right)\, ;\\
&& \vert\langle Ru,u\rangle\vert \leq C_2e^{(2\sup|\Phi| -\varepsilon_0)/h}\Vert \la hD_y\ra^m u\Vert^2,
\end{eqnarray}
where the two positive constants $C_2$ and $\varepsilon_0$ do not depend on $\Phi$, and where the estimate holds
for all $h>0$ small enough and $u\in H^m(\R^n )$.
\end{proposition}
\begin{proof} For $u\in C_0^\infty (\R^n)$, we write,
$$
e^{\Phi /h}Ae^{-\Phi /h}u(y) =\frac1{(2\pi h)^n}\int e^{i(y-y')\eta /h + (\Phi (y)-\Phi (y'))/h}
a(\frac{y+y'}2, \eta )u(y')dy'd\eta,
$$
and we first make the change of contour of integration given by,
$$
\R^n\ni\eta \mapsto \eta +i\delta\frac{y-y'}{\la y-y'\ra} \in \C^n,
$$
where $\delta >0$ is fixed small enough (independently of $\Phi$).

Then, we fix a cut-off $\chi_1\in C_0^\infty (\R^n;[0,1])$ such that  $\chi_0=1$ near 0. Using this cut-off to split the integral, and setting $f:= \delta\frac{y-y'}{\la y-y'\ra}$ and $g:=\delta\frac{|y-y'|^2}{\la y-y'\ra}$, we obtain,
$$
e^{\Phi /h}Ae^{-\Phi /h}u(y) = I_1 +I_2,
$$
with,
\begin{eqnarray*}
I_1 := \frac1{(2\pi h)^n}\int \chi_0 (y-y')e^{i(y-y')\eta /h + (\Phi (y)-\Phi ('y))/h - g/h}\\
\times a(\frac{y+y'}2, \eta +i f )u(y')dy'd\eta ;
\end{eqnarray*}
\begin{eqnarray*}
I_2 := \frac1{(2\pi h)^n}\int (1-\chi_0 (y-y'))e^{i(y-y')\eta /h + (\Phi (y)-\Phi (y'))/h - g/h}\\
\times a(\frac{y+y'}2, \eta +i f )u(y')dy'd\eta.
\end{eqnarray*}
Now, on the support of $1-\chi_0 (y-y')$, we have $g\geq 2\varepsilon_0$ for some $\varepsilon_0>0$ that does not depend on $\Phi$, and thus, by the Calder\'on-Vaillancourt theorem (see, e.g., \cite{Ma6}), we immediately obtain,
\be
|\la I_2, u\ra | \leq Ce^{(2\sup|\Phi| -\varepsilon_0)/h}\Vert \la hD_y\ra^m u\Vert^2,
\ee
where the constant $C>0$ does not depend on $\Phi$.

On the other hand, if $\delta$ has been chosen sufficiently small, the property (\ref{estreg1}) shows that, in $I_1$, we can make the change of contour of integration given by,
$$
\R^n\ni \eta \mapsto \eta +i\Psi (y,y'),
$$
where the smooth vectorial function $\Psi (y,y')$ is defined by the identity: $\Phi (y)-\Phi (y') = (y-y')\Psi (y,y')$. Then, denoting by ${\rm Op}_h$  the semiclassical quantization of symbols depending on
$3n$ variables (see e.g. \cite{Ma6} Section 2.5), we obtain,
$$
I_1 = {\rm Op}_h\left(e^{-g/h}\chi_0(y-y')a(\frac{y+y'}2, \eta +i\Psi (y,y') +if)\right)u (y).
$$
Next, passing to 
symbols depending
on $2n$ variables (see e.g. \cite{Ma6} Theorem 2.7.1), this gives,
$$
I_1 = {\rm Op}_h^W( b)u (y)
$$
with $b$  given by the oscillatory integral,
$$
b(y,\eta )= \frac1{(2\pi h)^n}\int e^{i(\eta'-\eta )\theta /h-g_1/h}\chi_0(\theta )a(y, \eta'+i\Psi (y+\frac{\theta}2, y-\frac{\theta}2)+if_1)d\eta'd\theta,
$$
where we have used the notations: $g_1:= \delta |\theta|^2/\la \theta\ra$; $f_1:= \delta \theta/\la \theta\ra$.
Then, we need,
\begin{lemma} \sl
\label{symbOh} For any $\delta>0$, the symbol $b$ can be split into,
$$
b(y,\eta) = a(y, \eta +i\nabla\Phi (y))+ \widetilde b(y,\eta),
$$
where, for any $\alpha,\beta\in\N^n$, $\widetilde b$ satisfies,
$$
\partial_y^\alpha\partial_\eta^\beta \widetilde b (y,\eta)={\mathcal O}\left( (1+\lambda_h{\bf 1}_{\Omega_\delta}(y)) h^{1-|\alpha|}\la \eta\ra^{2m}\right),
$$
uniformly for $h>0$ small enough and $(y,\eta)\in\R^{n}\times \C^n$, $|\Im \eta|$ small enough. Here, we have set,
$$
\Omega_\delta := \{ x\in\R^n\, ;\, {\rm dist}(x,\Omega )\leq \delta\}.
$$
\end{lemma}
\begin{proof} see Appendix \ref{appB}.
\end{proof}

This means that $h^{-1}\la\xi\ra^{-2m}\widetilde b$ belongs to the space of symbols $\widetilde S_{0}(\lambda_h, \Omega_\delta)$ introduced in Appendix \ref{appA}. Therefore, applying Proposition \ref{propA1} (iv), we deduce that,
for any $\delta>0$, one has,
\be
\label{3}
\vert \langle {\rm Op}_h^W( \widetilde b)u, u\rangle\vert =
{\mathcal O}(h\lambda_h)\Vert \langle hD_y\rangle^m u\Vert_{L^2(\Omega_\delta)}^2+ {\mathcal O}(h)\Vert \langle hD_y\rangle^m u\Vert^2,
\ee 
uniformly as $h\rightarrow 0$
 (see e.g. \cite{Ma6} Exercise 2.10.15). Moreover, one also has,
$$
a(y,\eta +i\nabla\Phi_h(y))= a_0(y,\eta +i\nabla\Phi_h(y))+h
a_1(y,\eta +i\nabla\Phi_h(y)),
$$
and,
\begin{eqnarray*}
\partial_y^\alpha\partial_\eta^\beta [a_1 (y,\eta +i\nabla\Phi_h(y))]
 &=& {\mathcal O}
\left( (1+\sup_{1\leq \vert\gamma\vert
\leq \vert\alpha\vert +1}\vert\partial^\gamma\Phi_h\vert )
\langle\eta\rangle^{2m}\right)\\
&=&
{\mathcal O}((1+\lambda_h h^{1-\vert
\alpha\vert}) \langle\eta\rangle^{2m})\\
&=&
{\mathcal O}( h^{-\vert
\alpha\vert} \langle\eta\rangle^{2m}).
\end{eqnarray*}
Thus, $\vert \langle {\rm Op}_h^W( a_1)u, u\rangle\vert =
{\mathcal O}(1)\Vert \langle hD_y\rangle^m u\Vert^2$,
and the result follows.
\end{proof}
\vskip 0.2cm

Thanks to the previous lemma, only the principal symbols of our operators 
will be involved in the estimates.  However, the study of $\widetilde P_2$ remains 
delicate since its principal symbol is not a polynome in $\eta$.
However, due to the parity of $\widetilde p_2(y,\eta)$ with respect to $\eta$, we have for
any real-valued Lipschitz function $\Phi (y)$ with 
$\vert\nabla\Phi\vert$ small enough,
\be
\Re \widetilde p_2(y,\eta + i\nabla\Phi(y)) = \widetilde p_2 (y,i\nabla\Phi(y))
+ \frac12\langle {\rm Hess}_\eta \widetilde p_2 (y,i\nabla\Phi(y))\eta ,\eta\rangle
+ {\mathcal O}(\vert\eta\vert^3)
\ee
and therefore, since ${\rm Hess}_\eta \widetilde p_2(0,0)$ is positive definite,
\be
\label{agm}
\Re \widetilde p_2(y,\eta + i\nabla\Phi(y))\geq \widetilde p_2 (y,i\nabla\Phi(y))+
\delta_0\vert\eta\vert^2
\ee
for some constant $\delta_0 >0$ and $(y,\eta)$ small enough. 
Actually, we have a more complete result on 
$\widetilde P_2:={\mathcal U}P_2{\mathcal U}^{-1}$, as stated in the following proposition:
\begin{proposition}\sl
\label{opp2}
 Let  $K\subset \R$ be a fix compact set, and let $\Phi =\Phi_{h,\mu}\in C^\infty (\R^n ;\R)$ be
 such that,
 \be
{\rm Supp }\nabla\Phi \subset \{ y_n\in K\},
 \ee
 and, 
\be
\label{hypoPhi1}
|\Phi| + | \nabla\Phi | ={\mathcal O}(\mu)
 \ee
 uniformly with respect to $h,\mu>0$ both  small enough. Moreover, assume that
 for some $\Omega\subset \R^n$, for
 any multi-index $\alpha\in\N^n$ with $|\alpha|\geq 2$, and for any fix $\mu>0$ small enough, one has,
\be
\label{hypoPhi2}
\partial^\alpha\Phi = {\mathcal O}\left((1+\lambda_h{\bf 1}_\Omega) h^{2-\vert
\alpha\vert } \right),
\ee
uniformly with respect to $h>0$ small enough. 
\vskip 0.1cm
 Then, for $\mu$ sufficiently small, there exist three positive
constants $\tau(\mu)$, $\delta =\delta(\mu)$ and $C_\mu $, such that, for all $\tau\in (0,\tau(\mu)]$ and for all $h>0$ small enough, one has,
\begin{eqnarray*}
\Re e^{\Phi /h} \widetilde P_2e^{-\Phi /h}\geq 
\frac1{C_\mu}\left( -\frac{h^2}2\Delta+\widetilde p_2 (y,i\nabla\Phi(y))\right)
 - C_\mu h (1+\lambda_h{\bf 1}_{\Omega_{\delta}}).
\end{eqnarray*}
\end{proposition}
\begin{proof} Let  $\chi  \in C_0^\infty (\R)$ be
 such that
$  \chi =1$ in a n neighborhood of   $K\cup \{0\}$.
Since $P_2 =-h^2\Delta +\tau V_2$ with $V_2$ satisfying (\ref{ass2}), by Lemma \ref{lemmfond} there exists $\nu_1>0$, such that,
$$
\chi(y_n)\widetilde P_2=\chi(y_n) (-h^2\Delta_y )+\tau {\rm Op}^W_h(\check V_2) +{\mathcal O}(e^{-\nu_1/h}),
$$
with $\check V_2\in S_{\nu_1}(1)$ satisfying,
$$\check V_2(y,\eta )=\chi (y_n)
V_2(y'-4\eta_n\eta', y_n-2\eta^2) +{\mathcal O}(h) \mbox{ in } S_{\nu_1}(1).
$$
Moreover, by assumption  $|\Phi| +|\nabla\Phi(y)| \leq C\mu$ for some constant $C>0$, and thus, for $\mu$ sufficiently small, we can apply Proposition \ref{poidsvar} with  $m=0$, and we obtain,
$$
e^{\Phi /h}{\rm Op}^W_h(\check V_2)e^{-\Phi /h}={\rm Op}_h^W\left( \chi(y_n) \widetilde V_2(y,\eta +i\nabla\Phi)
\right)
+hB+R,
$$
where we have set,
\be
\label{defVtilde}
\widetilde V_2(y,\eta) :=V_2(y'-4\eta_n\eta', y_n-2\eta^2),
\ee
and with,
$$
\vert\langle Bu,u\rangle\vert \leq C_1(\Vert u\Vert^2+\lambda_h\Vert u\Vert_{L^2(\Omega_\delta)}^2)\, ;\, \vert\langle Ru,u\rangle\vert \leq C_2e^{(2C\mu -\varepsilon_0)/h}\Vert u\Vert^2.
$$
Here, $\varepsilon _0 >0$ does not depend on $\mu$, and thus, for $\mu$ small enough, this gives,
$$
\Re e^{\Phi /h}{\rm Op}^W_h(\check V_2)e^{-\Phi /h}\geq \Re {\rm Op}_h^W\left( \chi(y_n) \widetilde V_2(y,\eta +i\nabla\Phi)
\right) -Ch(1+\lambda_h{\bf 1}_{\Omega_{\delta}}).
$$
Since we also have,
$$
\Re e^{\Phi /h}\chi(y_n)(-h^2\Delta)e^{-\Phi /h}=\Re \chi(y_n)(-h^2\Delta -|\nabla\Phi|^2)+{\mathcal O}(h),
$$
we obtain,
\begin{eqnarray}
\label{premest}
&& \Re e^{\Phi /h} \chi \widetilde P_2e^{-\Phi /h}\\
&&\geq \Re {\rm Op}_h^W\left( \chi(y_n)(\eta^2  -|\nabla\Phi|^2+\tau \widetilde V_2(y,\eta +i\nabla\Phi))
\right) \nonumber\\
&& \hskip 7cm -Ch(1+\lambda_h{\bf 1}_{\Omega_{\delta}}),\nonumber
\end{eqnarray}
with some new constant $C>0$. 

On the other hand, using (\ref{defVtilde}) and making a Taylor expansion, we see,
\begin{eqnarray*}
\Re \widetilde V_2(y,\eta +i\nabla\Phi) = V_2(y'-4\eta_n\eta' +4(\partial_{y_n}\Phi) \nabla_{y'}\Phi, y_n-2\eta^2+2(\nabla\Phi)^2)\\ -\la Y, AY \ra,
\end{eqnarray*}
where we have set,
$$
Y=(Y',Y_n):=(4\eta_n\nabla_{y'}\Phi+\eta'(\partial_{y_n}\Phi), 4\eta\cdot\nabla\Phi)\in \R^n,
$$
and,
$$
A:=\int_0^1 \Re V_2'' (y'-4\eta_n\eta' +4(\partial_{y_n}\Phi) \nabla_{y'}\Phi -itY', y_n-2\eta^2+2(\nabla\Phi)^2-itY_n)dt.
$$
In particular, by (\ref{hypoPhi1})-(\ref{hypoPhi2}),  we see that, for $\mu$ sufficienly small, 
$$
\chi (y_n)\la Y, AY \ra \in {\mathcal S}_{0}(\lambda_h, \Omega),
$$
where ${\mathcal S}_{0}(\lambda_h, \Omega)$ is the exotic class of symbols introduced in Appendix \ref{appA}. Actually, both $Y$ and $\chi A$ are in similar (vectorial) classes of symbols (with weight $\la\eta\ra$ for $Y$), and, for any $u\in C_0^\infty (\R^n)$, the results of Appendix \ref{appA} permit us to write,
\begin{eqnarray*}
\la {\rm Op}^W_h(\chi(y_n)\la Y, AY \ra)u,u\ra &=& \la {\rm Op}^W_h(\chi A){\rm Op}^W_h(Y)u, {\rm Op}^W_h(Y)u\ra \\
&& +{\mathcal O}_{\mu}(h\lambda_h)\Vert  \la hD_x\ra u\Vert^2_{\Omega_\delta} +{\mathcal O}_{\mu}(h)\Vert  \la hD_x\ra u\Vert^2\\
&=& {\mathcal O}_{\mu}(\Vert {\rm Op}^W_h(Y)u\Vert^2  \\
&& +{\mathcal O}_{\mu}(h\lambda_h)\Vert  \la hD_x\ra u\Vert^2_{\Omega_\delta} +{\mathcal O}_{\mu}(h)\Vert  \la hD_x\ra u\Vert^2\\
&=&  {\mathcal O}_{\mu}(\Vert  h\nabla u\Vert^2 +h\lambda_h\Vert  u\Vert_{\Omega_\delta}^2+h\Vert  u\Vert^2).
\end{eqnarray*}
Here,  ${\mathcal O}_{\mu}$ means that the estimate is not necessarily uniform with respect to $\mu$ (and this is precisely the reason why we need  the small parameter $\tau$ in our problem).

In the same way, we also have,
\begin{eqnarray*}
{\rm Op}^W_h(\chi (y_n)V_2(y'-4\eta_n\eta' +4(\partial_{x_n}\Phi) \nabla_{x'}\Phi, y_n-2\eta^2+2(\nabla\Phi)^2))\\
={\rm Op}^W_h(\chi (y_n)V_2(y'+4(\partial_{x_n}\Phi) \nabla_{x'}\Phi, y_n+2(\nabla\Phi)^2))+ B_1,
\end{eqnarray*}
with,
$$
\la B_1u,u\ra ={\mathcal O}_{\mu}(\Vert  h\nabla u\Vert^2 +h\lambda_h\Vert  u\Vert_{\Omega_\delta}^2+h\Vert  u\Vert^2).
$$
Therefore, coming back to (\ref{premest}), we deduce,
\begin{eqnarray}
 e^{\Phi /h} \chi (y_n) \widetilde P_2e^{-\Phi /h}\geq   {\rm Op}_h^W\left( (\chi(y_n)-C_\mu \tau)\eta^2  +\chi (y_n)\widetilde p_2(y,i\nabla\Phi(y))\right)\nonumber\\
\label{estP2-1}
 -C_\mu h\lambda_h{\bf 1}_{\Omega_\delta}-C_\mu h,
 \end{eqnarray}
where the constant $C_\mu>0$ may depend on $\mu$. Since $e^{\Phi /h}(1-\chi)  \widetilde P_2\chi e^{-\Phi /h}=\left(e^{-\Phi /h}  \chi \widetilde P_2(1-\chi)e^{\Phi /h}\right)^*$, we also have,
\begin{eqnarray}
\Re  e^{\Phi /h} (1-\chi)\widetilde P_2\chi e^{-\Phi /h}&\geq&  {\rm Op}_h^W\left((1-\chi(y_n)  \chi(y_n) (\eta^2  +\widetilde p_2(y,i\nabla\Phi(y))\right)\nonumber\\
\label{estP2-2}
&&- C_\mu\tau (hD_y)^2-C_\mu h\lambda_h{\bf 1}_{\Omega_\delta}-C_\mu h.
\end{eqnarray}
Finally, since $\Phi$ is constant on the support of $(1-\chi)$, we can write,
\begin{eqnarray}
&& \Re \la e^{\Phi /h} (1-\chi)\widetilde P_2(1-\chi )e^{-\Phi /h}u,u\ra\nonumber\\
\label{unmoinschiP}
&& =\la  (1-\chi)\widetilde P_2(1-\chi )u,u\ra\\
&& =\la {\rm Op}_h^W\left((1-\chi)^2\eta^2\right)u,u\ra + \tau \la (1-\chi){\mathcal U}V_2{\mathcal U}^{-1}(1-\chi)u,u\ra. \nonumber
\end{eqnarray}
Here, one must be aware of the fact that the pseudodifferential operator $(1-\chi){\mathcal U}V_2{\mathcal U}^{-1}(1-\chi)$, though being uniformly bounded on $L^2$, does not have its symbol in the class $S_0(1; \R^n)$ (there are directions where the repeated derivatives with respect to $\eta$ increase more and more as $|\eta|\rightarrow \infty$). However, we have,
\begin{lemma} The operator  $(1-\chi){\mathcal U}V_2{\mathcal U}^{-1}(1-\chi)$ satisfies,
$$
 \la (1-\chi){\mathcal U}V_2{\mathcal U}^{-1}(1-\chi)\geq \frac{(1-\chi)^2}{C} -C(hD_y)^2 -Ch^2,
$$
where $C>0$ is a constant.
\end{lemma}
\begin{proof} Let $\chi_0\in C_0^\infty (\R^n ;\R)$ be such that Supp $\chi_0 \subset \{ \chi (y_n) =1\}$ and $\chi_0(0)\not=0$.
Then, the assumptions on $V_2$ imply the existence of a positive constant $C$ such that,
$$
V_2 +\chi_0^2 \geq \frac1{C}.
$$
As a consequence,
\be
\label{unmoinschi}
(1-\chi){\mathcal U}V_2{\mathcal U}^{-1}(1-\chi)\geq \frac{(1-\chi)^2}{C}-(1-\chi){\mathcal U}\chi_0^2{\mathcal U}^{-1}(1-\chi),
\ee
and thus, it is enough to estimate $(1-\chi){\mathcal U}\chi_0^2{\mathcal U}^{-1}(1-\chi)$. We write it as,
$$
(1-\chi){\mathcal U}\chi_0^2{\mathcal U}^{-1}(1-\chi)={\mathcal U}A\chi_0^2 A{\mathcal U}^{-1},
$$
with,
$$
A:={\mathcal U}^{-1}(1-\chi){\mathcal U} =1-{\mathcal U}^{-1}\chi{\mathcal U}.
$$
By a straightforward computation (similar to that of (\ref{straightf})), we have,
$$
A\chi_0 =\chi_0 -{\rm Op}_h^R (a),
$$
with,
$$
a(y,\eta ):= \frac1{2\pi h}\int e^{i(x_n-y_n)(\xi_n -\eta_n)/h}\chi (x_n+\theta (\xi_n,\eta))\chi_0(y)dx_nd\xi_n,
$$
where we have set: $\theta (\xi, \eta )=2|\eta'|^2 +2(\eta_n^2+\eta_n\xi_n+\xi_n^2)/3$. Therefore, writing $\chi (x_n+\theta ) =\chi (x_n) +\theta \int_0^1\chi' (x_n+t\theta )dt$ and using the fact that $\chi\chi_0 =\chi_0$, we find,
$$
\chi_0(y)-a(y,\eta ) ={\rm Op}_h^R (b),
$$
with,
$$
b(y,\eta ):= \frac1{2\pi h}\int_0^1\int_{\R^{2}} e^{i(x_n-y_n)(\xi_n -\eta_n)/h}\theta(\xi_n,\eta) \chi' (x_n+t\theta (\xi,\eta))\chi_0(y)dxd\xi dt.
$$
Then, writing $\theta (\xi_n, \eta )=2\eta^2+ 2\eta_n(\xi_n-\eta_n) + \frac23 (\xi_n -\eta_n)^2$, and making integrations by parts with respect to $x_n$ and $\xi_n$, for any $M,N\geq 1$, we find,
$$
b(y,\eta )=2\eta^2A_1 + 2h\eta_nA_2 + h^2A_3,
$$
where, for $j=1,2,3$,
$$
A_j= \int_0^1\int_{{\mathcal E}_t }{\mathcal O}\left(h^{-1}\frac{(1+t|\xi_n|+t|\eta|)^M}{\la h^{-1}(\xi_n -\eta_n)\ra^N\la x_n-y_n\ra^M}\right) dx_nd\xi_n dt.
$$
Here, ${\mathcal E}_t$ stands for the  support of $\chi (x_n+t\theta (\xi_n,\eta))\chi_0(y)$. In particular, on this set we have,
$$
t|\xi_n| + t|\eta|= {\mathcal O}(\sqrt{t|\xi_n|^2 + t|\eta|^2}) ={\mathcal O}(\la x_n\ra^{1/2} +\la\xi_n -\eta_n\ra) \quad ;\quad |y| ={\mathcal O}(1),
$$
and thus, for $N\geq 2M$, this gives,
$$
h^{-1}\frac{(1+t|\xi_n|+t|\eta|)^M}{\la h^{-1}(\xi_n -\eta_n)\ra^N\la x_n-y_n\ra^M}={\mathcal O}\left(\frac{h^{-1}}{\la h^{-1}(\xi_n -\eta_n)\ra^M\la x_n\ra^{M/2}}\right).
$$
Therefore, $A_j={\mathcal O}(1)$ uniformly, and the same is true for all its derivatives. The same procedure also shows that $a\in S_0(1)$, and we can write,
\begin{eqnarray*}
&& \la A\chi_0^2 Au,u\ra=  \la {\rm Op}_h^R(b)\chi_0Au,u\ra={\mathcal O}( \Vert hD_y\chi_0Au\Vert\cdot \Vert hD_yu\Vert \\
&&\hskip 5cm + h\Vert \chi_0A u\Vert \cdot \Vert hD_yu\Vert +h^2\Vert \chi_0A u\Vert \cdot \Vert u\Vert),
\end{eqnarray*}
that is, since $\chi_0A ={\rm Op}_h^L(a)$ and $[D_y, \chi_0A]$ are uniformly bounded,
$$
\la A\chi_0^2 Au,u\ra ={\mathcal O}(\Vert hD_yu\Vert^2 +h^2\Vert u\Vert^2).
$$
Then, by (\ref{unmoinschi}) (and the fact that $D_y$ and ${\mathcal U}$ commute), the result follows.
\end{proof}

Going back to (\ref{unmoinschiP}), we deduce,
\begin{eqnarray*}
&& \Re e^{\Phi /h} (1-\chi)\widetilde P_2(1-\chi )e^{-\Phi /h}\\
&& \geq 
 {\rm Op}_h^W\left(((1-\chi)^2-C\tau)\eta^2+\frac{\tau}{ C}(1-\chi)^2-Ch^2\right).
 \end{eqnarray*}
Summing up with (\ref{estP2-1})-(\ref{estP2-2}), we have proved,
\begin{eqnarray}
&&\Re e^{\Phi /h} \widetilde P_2e^{-\Phi /h}\geq  {\rm Op}_h^W\left((1-C\tau )\eta^2+(1-(1-\chi)^2)\widetilde p_2(y,i\nabla\Phi(y))\right)\nonumber\\
\label{estP2-3}
 && \hskip 4cm+ \frac{\tau}{C} {\rm Op}_h^W\left((1-\chi)^2\right) - Ch\lambda_h{\bf 1}_{\Omega_\delta} -C h.
 \end{eqnarray}
On the other hand, on the support of $(1-\chi)$, we have $\widetilde p_2(y,i\nabla\Phi(y)) =\widetilde p_2(y,0)= \tau V_2(y)\geq 0$, and thus $(1-(1-\chi)^2)\widetilde p_2(y,i\nabla\Phi(y)) + \frac{\tau}{C}\geq \frac{\tau}{C}\geq C_1^{-1}\tau V_2(y)=C_1^{-1}\widetilde p_2(y,i\nabla\Phi(y))$, for some positive constant $C_1$. As a consequence, for $\tau >0$ sufficiently small, we deduce from (\ref{estP2-3}),
$$
\Re e^{\Phi /h} \widetilde P_2e^{-\Phi /h}\geq \frac1{C} {\rm Op}_h^W\left(\eta^2+\widetilde p_2(y,i\nabla\Phi(y))\right)-Ch\lambda_h{\bf 1}_{\Omega_\delta} -C h,
$$
where the (new) constant $C>0$ may depend on $\mu$. Thus, Proposition \ref{opp2} is proved.
\end{proof}

\subsection{Agmon Estimates}
\label{sectionAgmon}

Now, we can specify what we call Agmon estimates in our case. In the rest of the paper, we set,
$$k:=h\ln h^{-1},$$
and, for $t\in \R$, we consider the  translation $e^{itD_{x_n}}$ given by,
$$
e^{itD_{x_n}} f(y) = 
f(y', y_n+t).
$$
We also set,
$$
\widetilde P^t:=  e^{itD_{y_n}} \widetilde P e^{-itD_{y_n}}.
$$
Since $e^{itD_{y_n}}$ and  ${\mathcal U}$ commute, we have,
$$
\widetilde P^t ={\mathcal U}P^t{\mathcal U}^{-1},
$$
where,
$$
P^t:=e^{itD_{x_n}}Pe^{-itD_{x_n}}=\left(\begin{array}{cc}
P_1^t & 0\\
0 & P_2^t
\end{array}\right) + hR(x', x_n+t;hD_x),
$$
with,
$$
P_1^t:= -h^2\Delta + x_n+t-\mu\quad ;\quad P_2^t:=-h^2\Delta +\tau V_2(x',x_n+t).
$$
In particular, by Assumption (\ref{ass2}), we see that $\widetilde P^t $ extends analytically to complex values of $t$, $|t|$ small enough, and we set,
\be
Q := \widetilde P^{-ik}.
\ee
Then, the essential spectrum of $Q$ is $\R-ik$, and the resonances  of $P$ (or $\widetilde P$) close enough to 0 are discrete eigenvalues of $Q$. In particular, the resonant state $\widetilde{\mathbf u}={\mathcal U}{\mathbf u}$ of $\widetilde P$ associated to the resonance $\rho$ (where ${\mathbf u}$ is the corresponding resonant state of $P$), is such that ${\mathbf v}:=e^{kD_{y_n}}\widetilde{\mathbf u}$ is well defined, belongs to $L^2(\R^n)\oplus L^2(\R^n)$, and satisfies,
$$
(Q-\rho){\mathbf v} =0.
$$
Moreover, ${\mathbf v}$ can be normalized by requiring: $\Vert{\mathbf v}\Vert =1$.
\vskip 0.2cm
In the sequel, we set,
$$
Q_j= \widetilde P_j^{-ik}:={\mathcal U}P_j^{-ik}{\mathcal U}^{-1}.
$$
Actually,  since ${\mathcal F}x_n{\mathcal F}^{-1} = ih\partial_{\xi_n}$, the operator $Q_1$ can be explicitly computed, and one  finds,
$$
Q_1= h^2\Delta + y_n  -ik-\mu.
$$
Moreover, setting $\widetilde R_k:= {\mathcal U}R(x', x_n-ik; hD_x){\mathcal U}^{-1}$, we have,
\be
\label{Pjtildek}
Q=\left(\begin{array}{cc}
Q_1 & 0\\
0 & Q_2
\end{array}\right) + h\widetilde R_k,
\ee
Now, we fix an $h$-dependent cut-off function $\chi_h\in C^\infty (\R ; [0,1])$, such that,
$$
\chi_h =1 \mbox{ on } (-\infty, -2k^{2/3}]\, ; \, \chi_h =0 \mbox{ on } [-k^{2/3}, +\infty)\, ;\, \partial^\alpha\chi_h={\mathcal O}(k^{-2|\alpha|/3}),
$$
and we set,
\begin{eqnarray}
&&\widetilde V_1(y_n):= (\mu -y_n)\chi_h(y_n-\mu)+k^{2/3}\la y_n\ra(1-\chi_h(y_n-\mu))\, ;\nonumber\\
&&\widetilde Q_1:= h^2\Delta   -ik -\widetilde V_1(y_n);\nonumber\\
\label{defQtilde}
&& \widetilde Q=\left(\begin{array}{cc}
\widetilde Q_1 & 0\\
0 & Q_2
\end{array}\right) + h\widetilde R_k.
\end{eqnarray}
We also denote by ${\mathcal D}(\widetilde Q_1)$ the (natural) domain of $\widetilde Q_1$,
and we prove,

\begin{proposition}\sl
\label{agmon}
Let $\Phi$ be as in Proposition \ref{opp2}. Then, for any $\delta >0$, there exists  a positive
constant $C_2=C_2(\mu)$ such that, for any $v=(v_1,v_2)\in 
{\mathcal D}(\widetilde Q_1)\times H^2(\R^n) $, and setting $\check v :=(-v_1 , v_2)$, one
 has,
\begin{eqnarray*}
\Re\langle e^{\Phi /h}\widetilde Q v,e^{\Phi /h}\check v\rangle &\geq & \frac1{C_2}
\Vert h\nabla (e^{\Phi/h}v_1)\Vert^2
+\Vert h\nabla (e^{\Phi/h}v_2)\Vert^2\\
&&+\frac1{C_2}\langle \widetilde p_2(y,i\nabla\Phi(y))e^{\Phi/h}v_2,
e^{\Phi/h}v_2\rangle  \\
 && 
+\langle \left(\widetilde V_1(y_n)-(\nabla\Phi)^2)\right)e^{\Phi/h}v_1,e^{\Phi/h}v_1\rangle\\
&&  
-C_2h\Vert  e^{\Phi/h}v\Vert^2-C_2h\lambda_h\Vert  e^{\Phi/h}v_2\Vert_{\Omega_\delta}^2,
\end{eqnarray*}
uniformly for $h>0$ small enough. 
\end{proposition}
\begin{proof}   We first prove,
\begin{lemma}\sl
Let $c=c(x)$ satisfying (\ref{ass2}). Then, with $\Phi$ as in Proposition \ref{opp2}, the operator $e^{\Phi/h}{\mathcal U}c{\mathcal U}^{-1}e^{-\Phi/h}$ is uniformly bounded on $L^2$ as $h\rightarrow 0_+$.
\end{lemma}
\begin{proof}
Using Lemma \ref{lemmfond}, and Propositions \ref{poidsvar} and \ref{propA1}, we see that is is enough to study the operator $(1-\chi)e^{\Phi/h}{\mathcal U}c{\mathcal U}^{-1}e^{-\Phi/h}(1-\chi)$, where $\chi\in C_0^\infty (\R)$ is such that $\chi (y_n) =1$ on Supp$\nabla \Phi$. But then, this operator coincides with $(1-\chi){\mathcal U}c{\mathcal U}^{-1}(1-\chi)$, and the result is obvious.
\end{proof}
Now, if $c$ satisfies (\ref{ass2}), so does $c(x', x_n-ik)$, and thus, we deduce from the previous lemma that we have,
\be
\Vert e^{\Phi /h}h\widetilde R_k v\Vert \leq Ch\Vert \la hD_y\ra e^{\Phi /h}v\Vert.
\ee
As a consequence, by (\ref{defQtilde}), we deduce,
\begin{eqnarray}
\langle e^{\Phi /h}\widetilde Qv,e^{\Phi /h}\check v\rangle = -\la e^{\Phi /h}\widetilde Q_1 v_1,e^{\Phi /h}v_1\ra + \la e^{\Phi /h}Q_2 v_2,e^{\Phi /h}v_2\ra \nonumber \\
\label{estQtilde1}
+{\mathcal O}(h\Vert \la hD_y\ra e^{\Phi /h}v\Vert^2).
\end{eqnarray}
We also compute,
$$
-e^{\Phi/h}\widetilde Q_1e^{-\Phi/h} =-h^2\Delta + \widetilde V_1(y_n)   -|\nabla\Phi|^2+ik-ih(D_y\cdot\nabla\Phi +\nabla\Phi\cdot D_y),
$$
and thus,
$$
-\Re \la e^{\Phi /h}\widetilde Q_1 v_1,e^{\Phi /h}v_1\ra = \Vert h\nabla u\Vert^2 + \la (\widetilde V_1(y_n)   -|\nabla\Phi|^2)e^{\Phi /h}v, e^{\Phi /h}v\ra.
$$
Then, the result follows from (\ref{estQtilde1}) by applying Proposition \ref{opp2} and by taking $h$ small enough.
\end{proof}

\begin{remark}
\label{unifwrphi} \sl
It results from the proof that the dependence with respect to $\Phi$ of the constant $C_2$ is performed only through the estimates (\ref{hypoPhi1})-(\ref{hypoPhi2}). As a consequence, if $\Phi$ depends on some extra-parameter and satisfies (\ref{hypoPhi1})-(\ref{hypoPhi2}) uniformly with respect to this parameter, then the constant $C_2$ can be taken independent of it.
\end{remark}

\section{Comparison Between Formal and True Solution}
\subsection{On $\{\phi < S(\mu)\}$}\label{intboule}

Now, we start by studying the region $\{\phi < S(\mu)\}$. The idea is to use the previous Agmon estimates in order to have informations on the decay of the first  eigenfunction $\widetilde{\mathbf v}$ of $\widetilde Q$, and then to transfer these informations on ${\mathbf v}$ by comparing $\widetilde{\mathbf v}$ and  ${\mathbf v}$. 
\vskip 0.2cm
Note that, by construction, we have $\widetilde V_1\geq \delta k^{2/3}\la y_n\ra$ everywhere (with some $\delta >0$ constant). As a consequence, we see that the operator $\widetilde Q_1 - z$ is invertible for any complex number $z$ with  $z={\mathcal O}(h)$, and the norm of its inverse is ${\mathcal O}(k^{-2/3})$. Therefore, the $2\times 2$ system of equations: $Qu=zu$ (with $u=(u_1,u_2)$), can be reduced to a scalar equation on $u_2$, of the form,
$$
(Q_2 + h^2k^{-2/3}B(z))u_2=zu_2,
$$
with $\Vert B(z)\Vert ={\mathcal O}(1)$ uniformly. Since $h^2k^{-2/3} =h^{4/3}|\ln h|^{-2/3} << h$, one can perform the standard regular-perturbation procedure, and conclude (also using Proposition \ref{propC1} in Appendix \ref{appC}) that there exists a unique discrete eigenvalue $\widetilde\rho$ of $\widetilde Q$ near $\rho$, that coincides with $\rho$ up to ${\mathcal O}(h^{4/3}|\ln h|^{-2/3})$. Finally, using, e.g., the WKB constructions of Section \ref{sectionWKB}, we easily conclude that 
$\widetilde \rho$ coincides with $\rho$ up to ${\mathcal O}(h^\infty)$.
\vskip 0.2cm
We have the following result:
\begin{proposition}
\label{estile}
For any multi-index 
$\alpha$, the (conveniently normalized) eigenfunction $\widetilde{\mathbf v}$ of $\widetilde Q$, associated with the eigenvalue $\widetilde\rho$, satisfies,
\be
\partial^\alpha ({\widetilde{\mathbf v}} - {\mathbf w}) = {\mathcal O}(h^\infty e^{-\phi /h})
\ee
locally uniformly in $\{\phi < S(\mu)\}$. Moreover, for any $\varepsilon >0$
\be
\partial^\alpha {\widetilde{\mathbf v}} = {\mathcal O}(e^{-(S(\mu)-\varepsilon) /h}),
\ee
locally uniformly on $\R^n\backslash \{\phi < S(\mu)\}$.
\end{proposition}
\begin{proof} At first, 
we observe that, by arguments similar to those of  \cite{HeSj1} (that is, using the WKB solution ${\mathbf w}$ constructed in Section \ref{sectionWKB}, and comparing it with $\widetilde\Pi {\mathbf w}$, where $\widetilde\Pi$ is the spectral projector of $\widetilde Q$ onto $Ker(\widetilde Q -\widetilde \rho)$),  for any multi-index $\alpha$, one has,
\be
\label{prespuits}
\partial^\alpha (\widetilde{\mathbf v} -{\mathbf w})={\mathcal O}(h^\infty),
\ee
uniformly in some ($\mu$-dependent) neighborhood of 0.
\vskip 0.2cm

Then, we  consider the regularization $\widetilde \phi$ of $\phi$ obtained by 
modifying it near $\Gamma$ in the following way: 
$$
\widetilde\phi = \widetilde\varphi_2 + \chi_0(z/h)(\widetilde\varphi_1-\widetilde\varphi_2)
$$
where  $z=z(y)$ is defined in (\ref{z}), and $\chi_0
\in C^\infty 
(\R; [0,1])$, $\chi_0(t)=0$ for $t\leq -1$, $\chi_0 (t) = 1$ for $t\geq 1$, and $0\leq \chi'\leq 3/2$. 
In particular, $\widetilde \phi$ is well defined on $\{ \phi < S(\mu)\}$, and, by (\ref{z}), we see that $\widetilde \phi$ satisfies
(\ref{hypoPhi1})-(\ref{hypoPhi2})  locally uniformly, with $\lambda_h=1$. 
Moreover, one also has,
$$
\widetilde\phi - \phi ={\mathcal O}( h^2)\quad ; \quad \nabla\widetilde\phi - \nabla\phi
 ={\mathcal O}( h) {\rm \quad a.e.},
$$
and, for some $\delta =\delta(\mu)>0$ small, $\partial^\alpha\widetilde\phi ={\mathcal O}(1)$ on $\{\widetilde\phi <\delta\}$.

Now, for $C>0$ fixed large enough, we consider the function
\be
\label{Phi0}
\Phi_0 = \widetilde\phi -Ch\ln \left( 1 + (\frac{\widetilde\phi}{Ch} -1)\chi 
( \widetilde\phi /(Ch) ) \right)
\ee
where $\chi \in C^\infty (\R)$, $\chi (t) =0$ for $t\leq 1$, $\chi (t) =1$ 
for $t\geq 2$, $0\leq \chi'\leq 2$.
\vskip 0.2cm 
On the support of $\chi'(\widetilde\phi /(Ch))$, we have $\widetilde\Phi =\widetilde\varphi_2$ e $\nabla\widetilde\phi ={\mathcal O}(\sqrt {Ch})$. Therefore, on this set, we see that $\Phi_0$ satisfies,
$$
\partial^\alpha \Phi_0 = {\mathcal O}(1 + h^{1-\vert\alpha\vert /2}),
$$
uniformly with respect to $C\geq 1$ and $h$ small enough. On the other hand, away from Supp$\chi'(\widetilde\phi /(Ch))$, we have  either $\Phi_0 =\widetilde\phi$, or $\Phi_0 =\widetilde\phi -Ch\ln (\widetilde\phi /Ch)\\ = \widetilde\phi -Ch\ln \widetilde\phi +Ch\ln (Ch)$. In any case,  we conclude that $\Phi_0$ satisfies (\ref{hypoPhi1})-(\ref{hypoPhi2}) with $\lambda_h=1$, locally uniformly on $\{ \phi < S(\mu)\}$, and uniformly with respect to the pair $(C,h)$ such that $Ch\leq 1$.

\vskip 0.2cm
We also have,
$$
\widetilde p_2 (y, i\nabla\Phi_0(y))=\widetilde p_2 (y, i(1-F(\widetilde\phi /(Ch)))
\nabla\widetilde\phi(y))
$$
with $F(t)= ((t-1)\chi'(t) +\chi(t))/(1+(t-1)\chi (t))$, and one can easily check that $\vert F(t)\vert\leq 4/t$. Therefore,  $F(
\widetilde\phi /(Ch))\nabla\widetilde\phi(y) \leq 
4h\vert\nabla\widetilde\phi\vert /\widetilde
\phi$, and thus, 
by a Taylor expansion, we find,
\begin{eqnarray*}
\widetilde p_2 (y, i\nabla\Phi_0(y))=\widetilde p_2 (y, i\nabla\widetilde\phi(y))-
iF(\widetilde\phi /(Ch))\nabla\widetilde\phi(y)\nabla_\eta
\widetilde p_2 (y, i\nabla\widetilde\phi(y)) \\+{\mathcal O}((Ch\vert\nabla\widetilde\phi\vert 
/\widetilde\phi)^2).
\end{eqnarray*}
Since $\widetilde p_2 (y, i\nabla\widetilde\varphi_2(y))=0$ and
$\widetilde p_2(y,i\nabla\widetilde\varphi_1(y))>0$ on $\Gamma_+$ (see Appendix \ref{appD}),
 and 
$-i\nabla\widetilde\phi(y)\nabla_\eta
\widetilde p_2 (y, i\nabla\widetilde\phi(y))\geq \delta_0
\vert\nabla\widetilde\phi(y)\vert^2$ 
in some fixed neighborhood of 0 and for some uniform constant $\delta_0>0$, we conclude,
$$
\widetilde p_2 (y, i\nabla\Phi_0(y))\geq \delta_0F(\widetilde\phi /(Ch))
\vert\nabla\widetilde
\phi(y)\vert^2+{\mathcal O}((Ch\vert\nabla\widetilde\phi\vert /\widetilde\phi)^2)+
{\mathcal O}(h^2) +{\mathcal O}_\Gamma(h)
$$
where the last ${\mathcal O}_\Gamma(h)$ is uniform with respect to $C$ and 
supported near $\Gamma$.
In particular, on $\{ \widetilde \phi\geq C^2h\}$, and if $C$ is 
large enough,
$$
\widetilde p_2 (y, i\nabla\Phi_0(y))\geq C\frac{h}{2\widetilde\phi}\vert\nabla
\widetilde\phi(y)\vert^2 +{\mathcal O}(h^2) +{\mathcal O}_\Gamma(h).
$$
Since, near $\Gamma$, $\vert\nabla\widetilde\phi\vert$
 stays uniformly away from 0, we 
finally obtain, by taking $C$ large enough,
\be
\label{ellp2}
\widetilde p_2 (y, i\nabla\Phi_0(y))\geq C\frac{h}{4\widetilde\phi}\vert\nabla
\widetilde\phi(y)\vert^2\quad {\rm on}\quad \{ \widetilde \phi\geq C^2h\},
\ee
uniformly for $h$ small enough.

Now, for a given $\varepsilon >0$ arbitrarily small, we set,
$$
\Phi(y)= \chi_\varepsilon ( \Phi_0(y) ),
$$
where   
$\chi_\varepsilon\in C^\infty (\R_+ )$ is such that 
$\chi_\varepsilon (t) 
= t$ if
$0\leq  t\leq S(\mu ) -3\varepsilon$, $\chi_\varepsilon (t) = 
S(\mu ) -2\varepsilon$ if
$t \geq S(\mu ) -\varepsilon$, $0\leq \chi_\varepsilon'(t)\leq 1$ everywhere.
Then, the function $\Phi$ is smooth on $\R^n$ and it  satisfies,
\begin{eqnarray*}
&&\Phi = \Phi_0(y)\quad {\rm on}\quad \{ \phi \leq S(\mu) - 
4\varepsilon\};\\
&&\Phi= S(\mu) -2\varepsilon\quad {\rm on}\quad \{ \phi \geq S(\mu) - 
\frac{\varepsilon}{2}\};\\
&&\vert\nabla\Phi\vert \leq  \vert \nabla\Phi_0\vert \quad {\rm everywhere};\\
&& \partial^\alpha\Phi ={\mathcal O}(h^{-(2-|\alpha|)_+}) \mbox{ uniformly on } \R^n.
\end{eqnarray*}
Moreover, the last estimate is also uniform with respect to the pair $(C,h)$ as long as $Ch\leq 1$.
\vskip 0.2cm
Then, thanks to  (\ref{ellp2}), and the fact that $\vert\nabla
\phi\vert^2 \geq \delta \phi$ for some positive constant  $\delta$ on 
$\{\phi\leq S(\mu )- \varepsilon \}$, we see
 that, (denoting by ${\bf 1}_A$ the characteristic function of a set $A$),
$$
\widetilde p_2(y,i\nabla\Phi(y))\geq C'h+\delta_3{\bf 1}_{\{ \phi (y)\geq S(\mu )-\varepsilon /2\}}
$$
on $\{\widetilde\phi \geq
C^2h\}$, with $\delta_3 >0$ and $C'=C'(C)>o$ arbitrarily large, depending on the 
choice of $C$ (here, we also use the fact that
$\widetilde p_2(y,is\nabla\widetilde\varphi_1(y)) >0$ for $y\in \Gamma_+$ and 
$0\leq s\leq 1$: See Appendix \ref{appD}).
On the other hand, on the support of $\nabla\Phi
$, we have $\vert\nabla\Phi\vert^2\leq \vert\nabla\Phi_0\vert^2\leq
(1-Ch/\widetilde\phi )\vert\nabla\widetilde\phi\vert^2$, and we also obtain
on this set,
$$
\widetilde V_1 (y)-(\nabla\Phi)^2= \mu - y_n - (\nabla\Phi)^2\geq C'h\la y_n\ra,
$$
with $C'>0$ arbitrarily large,
while, away from Supp$\nabla\Phi$, we have $\widetilde V_1 (y)-(\nabla\Phi)^2=\widetilde V_1 (y)\geq \delta k^{2/3}\la y_n\ra >> h\la y_n\ra $. 

Finally, on $\{\phi\leq C^2h\}$, one has $\Phi=\phi$, and thus $e^{\Phi /h} = 
{\mathcal O}(1)$.

Now, we apply the Agmon 
estimates of Proposition \ref{agmon} with $v:=\widetilde{\mathbf v}-{\mathbf w}$ (and $\lambda_h=1$). Then,
thanks also to (\ref{eqW1})-(\ref{eqW3}) and Remark \ref{unifwrphi},
we  obtain,
\begin{eqnarray*}
\frac1{C_2}\Vert h\nabla (e^{\Phi /h}v)\Vert^2 + (C'-C_2)h\Vert e^{\Phi /h}v_2\Vert^2_{\phi\geq C^2h}+(C'-C_2)h\Vert \sqrt{\la y_n}\ra e^{\Phi /h}v_1\Vert^2 \\
={\mathcal O}(\Vert v\Vert_{\phi\leq C^2h} )+{\mathcal O}(h^\infty),
\end{eqnarray*}
where the constant $C_2$ does not depend on the choice of $C$ in (\ref{Phi0}), and where $C'=C'(C)$ tends to $\infty$ as $C\rightarrow +\infty$. Therefore, choosing $C$ sufficiently large, and using (\ref{prespuits}), we conclude,
$$
\Vert h\nabla (e^{\Phi /h}v)\Vert + \Vert \sqrt{\la y_n}\ra e^{\Phi /h}v_1\Vert +\Vert e^{\Phi /h}v_2\Vert={\mathcal O}(h^\infty).
$$
Finally, using repeatedly the partial differential equations satisfied by $\widetilde{\mathbf v}$ and $\mathbf w$, and taking advantage of the classical ellipticity of $\widetilde Q_1$ and $Q_2$, together with the positivity of $\widetilde V_1$ for $|y_n|$ large, we deduce in a standard way that, for all $s\in\R$,
\be
\label{estv1}
\Vert \la y_n\ra^s e^{\Phi /h}v_1\Vert_{H^s(\R^n)} +\Vert e^{\Phi /h}v_2\Vert_{H^s(\R^n)}={\mathcal O}(h^\infty),
\ee
and thus, for any $\alpha\in\N^n$, by standard Sobolev estimates,
$$
\partial^{\alpha}v ={\mathcal O}(h^\infty e^{-\Phi /h}),
$$
uniformly. Then, the results follows from the fact that, on $\{\phi \leq S(\mu) -4\varepsilon\}$, one has  $\Phi =\Phi_0\geq \phi -Ch\ln (S(\mu)/Ch)$, while, on $\R^n\backslash \{\phi \leq S(\mu) -4\varepsilon\}$, one has $\Phi \geq S(\mu) - 5\varepsilon$.
\end{proof}

Now, we are able to compare the (complex-translated) resonant state ${\mathbf v}$ with $\mathbf w$.
\begin{proposition} 
\label{estv-w1}
For any multi-index 
$\alpha$, the first (conveniently normalized) eigenfunction ${\mathbf v}$ of $Q$ satisfies,
\be
\partial^\alpha ({\mathbf v} - {\mathbf w}) = {\mathcal O}(h^\infty e^{-\phi /h})
\ee
locally uniformly in $\{\phi < S(\mu)\}$. Moreover, for any $\varepsilon >0$
\be
\partial^\alpha {{\mathbf v}} = {\mathcal O}(e^{-(S(\mu)-\varepsilon) /h}),
\ee
locally uniformly on $\R^n\backslash \{\phi < S(\mu)\}$.
\end{proposition}
\begin{proof} Let $\gamma$ be the oriented circle centered at $\rho$ with radius $\delta h$, where $\delta >0$ is a small enough constant. Then, $\gamma$ does not meet the spectrum of $Q$, and we can write the spectral projector $\Pi$ of $Q$ associated with $\rho$, as,
$$
\Pi = \frac1{2i\pi}\oint_\gamma (z-Q)^{-1}dz.
$$
We set,
$$
{\mathbf v}_0 := \Pi \widetilde{\mathbf v}.
$$
Then, ${\mathbf v}_0$ is solution of $(Q-\rho){\mathbf v}_0=0$ and is colinear with ${\mathbf v}$.
Therefore, the result will be an easy consequence of Proposition \ref{estile} and the following lemma:
\begin{lemma}\sl
\label{estv0vtilde}
For any multi-index 
$\alpha$ and for any $\varepsilon >0$, one has,
\be
\partial^\alpha ({\mathbf v}_0 -  \widetilde{\mathbf v}) = {\mathcal O}(e^{-(S(\mu) -\varepsilon)/h})
\ee
locally uniformly in $\R^n$.
\end{lemma}
\begin{proof} Since $\widetilde Q  \widetilde{\mathbf v} =\widetilde E  \widetilde{\mathbf v}$, for any $z\in \gamma$, we have,
\be
\label{eqapprvtilde}
(Q-z) \widetilde{\mathbf v} = (\widetilde E -z) \widetilde{\mathbf v} + r,
\ee
with $r:= (Q-\widetilde Q) \widetilde{\mathbf v}=((\la y_n\ra + y_n -\mu)(1-\chi_h(y_n-\mu)) \widetilde{\mathbf v}_1, 0) $ (where we have set $ \widetilde{\mathbf v}=:( \widetilde{\mathbf v}_1,  \widetilde{\mathbf v}_2)$). In particular, by (\ref{estv1}), for any $s\geq0$ and any $\varepsilon >0$, we see,
\be
\label{estr}
\Vert r\Vert_{H^s(\R^n)} ={\mathcal O}(e^{-(S(\mu)-\varepsilon )/h}).
\ee
Moreover, we deduce from (\ref{eqapprvtilde}),
$$
(z-Q)^{-1}\widetilde{\mathbf v}= (z-\widetilde E)^{-1}\widetilde{\mathbf v} +(z-\widetilde E)^{-1}(z-Q)^{-1}r,
$$
and thus,
\be
\label{compv0v}
{\mathbf v}_0 = \widetilde{\mathbf v} +\frac1{2i\pi}\int_\gamma (z-\widetilde E)^{-1}(z-Q)^{-1}r dz.
\ee
Now, along $\gamma$, we have $(z-\widetilde E)^{-1}={\mathcal O}(h^{-1})$, $(z-Q_1)^{-1} ={\mathcal O}(k^{-1})$, and $(z-Q_2)^{-1} ={\mathcal O}(h^{-1})$ (this last estimate comes from  Appendix \ref{appC} and the fact that $Q_2={\mathcal U}(-h^2\Delta +V_2(x',x_n-ik)){\mathcal U}^{-1}$). Therefore, 
$$
\Vert hR_{2,1}(Q_1-z)^{-1}hR_{1,2}(Q_2-z)^{-1}\Vert ={\mathcal O}(hk^{-1}) ={\mathcal O}\left(\frac1{|\ln h|}\right).
$$
In particular, for $h$ small enough, the operator $I+hR_{2,1}(Q_1-z)^{-1}hR_{1,2}(Q_2-z)^{-1}$ is invertible with uniformly bounded inverse, and one easily deduces that $\Vert (z-Q)^{-1}\Vert ={\mathcal O}(h^{-1})$. Since one also has $(z-\widetilde E)^{-1}={\mathcal O}(h^{-1})$, we learn from (\ref{estr})-(\ref{compv0v}),
$$
\Vert {\mathbf v}_0 - \widetilde{\mathbf v}\Vert_{L^2}={\mathcal O}(e^{-(S(\mu)-\varepsilon )/h}).
$$
The same estimate for the $H^s$-norms is proved in the same way, by applying any arbitrary power of $Q$ to (\ref{compv0v}), and the result follows by Sobolev estimates.
\end{proof}
Now, to complete the proof of the proposition, we just observe that, by construction, ${\mathbf v}$ and ${\mathbf v}_0$ are co-linear, and $\Vert {\mathbf v}_0\Vert =1+ {\mathcal O}(e^{-(S(\mu) -\varepsilon)/h})$. therefore, we can take ${\mathbf v}=\lambda {\mathbf v}_0$ with $\lambda  =1+ {\mathcal O}(e^{-(S(\mu) -\varepsilon)/h})$, and the result follows from Lemma \ref{estv0vtilde} and Proposition \ref{estile}.
\end{proof}

\subsection{Near  $\{ y_n =\mu\}$}

\subsubsection{A priori estimates}

We
first prove the following a priori  estimates:
\begin{proposition} \sl
\label{apriorivtilde}
For any $\alpha\in \N^n$,
there exists $N_\alpha\geq 0$ such that, for any $N\geq 1$,
the eigenfunction $\widetilde{\mathbf v}$ satisfies,
\be
\label{vtilde-W}
\partial^{\alpha}(\widetilde{\mathbf v} - {\mathbf w}) ={\mathcal O}( h^{-N_\alpha}e^{-S(\mu )/h})
\ee
locally uniformly in $ \{y_n\geq \mu  - (Nk)^{2/3}\}$ (where $k:=h\ln h^{-1}$).
\end{proposition}
\begin{remark}\sl  An important feature of the last estimate is that the number 
$N_\alpha$ does not depend on the choice of $N$.
\end{remark}

\begin{proof} Again, the proof is based on Agmon estimates,
but we have to be more precise on  the construction of the weight function 
near
$\{y_n= \mu\}$.  

Fix $\varepsilon >0$ small enough, and let $\chi \in C^\infty (\R; [0,1])$ such that $\chi =1$ on $[\mu -\varepsilon, +\infty)$, and $\chi =0$ on $(-\infty, \mu -2\varepsilon]$. We set,
\be
\label{defphi0tilde}
\widetilde\phi_0 (y_n)=\chi (y_n)\left(S(\mu) - \frac23(\mu -y_n)^{3/2}\right)
\ee
and, for $N\geq 1$ large enough and $y_n\in[0,\mu)$,  we define,
\be
\label{defPhiN}
\Phi_N(y_n) =  \widetilde\chi \left( \widetilde\phi_0(y_n)  +3k(S(\mu )-\widetilde\phi_0(y_n) )^{1/3} +3Nk\right)
\ee
where $\widetilde\chi\in C^\infty (\R_+)$ 
depends on $Nk$ and is such that $\widetilde\chi (s)= s$ for $s\leq S(\mu )$,
$\widetilde\chi (s) = S(\mu ) +Nk$ for $s\geq S(\mu ) +2Nk$, $0\leq \widetilde
\chi'\leq 1$ 
everywhere, and, for all $\ell \geq 0$,  
$\widetilde\chi^{(\ell )}={\mathcal O}((Nk)^{-(\ell
-1)_+})$ uniformly.

Then, $\Phi_N$ is well defined and smooth on $(-\infty,\mu)$, and we have $\Phi_N = S(\mu )+Nk$ when $\widetilde\phi_0  +3k(S(\mu )-\widetilde\phi_0 )^{1/3}\geq S(\mu )-Nk$. The latter condition is implied by
$S(\mu ) -\widetilde\phi_0 \leq Nk$, and thus, if we
extend $\widetilde \Phi_N$ by the constant value $S(\mu )+Nk$ on 
$\{\widetilde\phi_0  \geq S(\mu )-Nk\}$, we obtain a smooth function on $\R$, that satisfies,
\begin{eqnarray}
\Phi_N' &=&\widetilde\chi'\left( \widetilde\phi_0  +3k(S(\mu )-\widetilde\phi_0)^{1/3} +3Nk\right)\left( 1-\frac{k}{(S(\mu)-
\widetilde\phi_0 )^{2/3}}
\right)\widetilde\phi_0'\nonumber\\
\label{nablaPhiN}
  &=&{\mathcal O}(1).
\end{eqnarray}
Moreover, one can check that,
for $\ell\geq 1$, one has,
\be
\label{estderiv2Phi}
 \widetilde\phi_0^{(\ell)}={\mathcal O}\left(
 (S(\mu )-\widetilde\phi_0  )^{1-2\ell/3}\right) \mbox{ in } (\mu-2\varepsilon, \mu).
\ee
In particular, $\widetilde\phi_0' ={\mathcal O}((S(\mu) - \widetilde\phi_0)^{1/3})$, and,
since $S(\mu )-\widetilde\phi_0  \in [Nk , 3Nk]$ on 
the support of $\widetilde\chi''\left( \widetilde\phi_0  +3k(S(\mu )-\widetilde\phi_0 )^{1/3} +3Nk\right)$, and $\widetilde\chi^{(\ell )}={\mathcal O}((Nk)^{-(\ell
-1)_+})$,
we deduce that, for all $\ell\geq 2$, 
$$
\Phi_N^{(\ell)}={\mathcal O}(
(Nk)^{1-2\ell /3})
$$
uniformly in $(\mu-2\varepsilon, \mu)$. Moreover, by construction, $\Phi_N$ is constant on both $(-\infty, \mu-2\varepsilon]$ and $[\mu, +\infty)$, and since $S(\mu) -\widetilde \phi_0$ vanishes at $y_n=\mu$ only, we also have $\partial^\alpha\Phi_N ={\mathcal O}(1)$ uniformly on $[\mu -2\varepsilon, \mu -\varepsilon ]$. Therefore, $\Phi_N$ satisfies (\ref{hypoPhi2}) with $\lambda(h)=k^{-1/3}$ and $\Omega =\widetilde\Omega:=\{ y\in \R^n\,;\, \mu -\varepsilon \leq y_n\leq \mu\}$, and we can apply Proposition \ref{agmon} with $\Phi =\Phi_N$, $\lambda_h =k^{-1/3}$,  and $v  =\widetilde{\mathbf v}-{\mathbf w}$. Then, taking into account  (\ref{eqW1})-(\ref{eqW3}) and the fact that $\widetilde \rho -\rho ={\mathcal O}(h^\infty)$, for any $\delta >0$, we obtain,
\begin{eqnarray*}
&& \frac1{C_2}
\Vert h\nabla (e^{\Phi_N/h}v)\Vert^2
+\frac1{C_2}\langle \widetilde p_2(y,i\nabla\Phi_N(y))e^{\Phi_N/h}v_2,
e^{\Phi_N/h}v_2\rangle \\
&& \hskip 4cm+\langle \left(\widetilde V_1(y_n)-(\nabla\Phi_N)^2)\right)e^{\Phi_N/h}v_1,e^{\Phi_N/h}v_1\rangle\\
&&  \leq C_2h\Vert  e^{\Phi_N/h}v_1\Vert^2+C_2h k^{-1/3}\Vert  e^{\Phi_N/h}v_2\Vert_{\widetilde\Omega_\delta}^2 +C_2h \Vert  e^{\Phi_N/h}v_2\Vert^2
\\
&& \hskip 2cm+ \varepsilon (h)e^{ 2c_1/h}+C_2h^{-2N_0}e^{ 2c_2/h}+C_2h^{-2N_0}e^{ 2c_3/h},
\end{eqnarray*}
where $C_2>0$ is a constant, the quantities  $N_0, \varphi$ are those appearing in (\ref{eqW1})-(\ref{eqW3}), $ \varepsilon (h)={\mathcal O}(h^\infty)$, and where we have set,
\begin{eqnarray*}
&& c_1:= \sup_{\widetilde{\mathcal V}_2} (\Phi_N -\varphi);\\
&& c_2:= \sup_{\widetilde{\mathcal V}_2\cap \{y_n\geq g(y')\}} (\Phi_N -S(\mu));\\
&& c_3:= \sup_{\widetilde{\mathcal V}_2\backslash\widetilde{\mathcal V}_1} (\Phi_N -\varphi).
\end{eqnarray*}
Now, using Lemma \ref{lemmephi0},  we see,
$$
\Phi_N -\varphi \leq 3k(S(\mu )-\widetilde\phi_0 )^{1/3} + 3Nk\leq 3(N+S(\mu)^{1/3})k,
$$
and thus, for $h$ small enough, 
$$
e^{2c_1/h}\leq h^{-6(N+S(\mu)^{1/3})}.
$$
Moreover,  since $\Phi_N \leq S(\mu)+Nk$ everywhere, we have,
$$
e^{2c_2/h}=h^{-2N}.
$$
Finally, since $\varphi > S(\mu)$ on $\widetilde{\mathcal V}_2\backslash\widetilde{\mathcal V}_1$, we have $c_3<0$ and thus $e^{2c_3/h}={\mathcal O}(h^\infty)$. Summing up, for $N$ large enough, we have proved,
\begin{eqnarray}
&& \frac1{C_2}
\Vert h\nabla (e^{\Phi_N/h}v)\Vert^2
+\frac1{C_2}\langle \widetilde p_2(y,i\nabla\Phi_N(y))e^{\Phi_N/h}v_2,
e^{\Phi_N/h}v_2\rangle \nonumber\\
&& \hskip 5cm+\langle \left(\widetilde V_1(y_n)-(\nabla\Phi_N)^2)\right)e^{\Phi_N/h}v_1,e^{\Phi_N/h}v_1\rangle \nonumber\\
\label{presyn=mu}
&&  \leq C_2h\Vert  e^{\Phi_N/h}v_1\Vert^2+C_2h k^{-1/3}\Vert  e^{\Phi_N/h}v_2\Vert_{\widetilde\Omega_{\delta}}^2\\
&&\hskip 5cm +C_2h \Vert  e^{\Phi_N/h}v_2\Vert^2
+C_2h^{-2(N+N_0)}+C_N, \nonumber
\end{eqnarray}
where $C_N>0$ is an $h$-dependent constant.
Now, for $N$ sufficiently large, on Supp\hskip 1pt$\nabla\Phi_N$, we have $\widetilde V_1(y_n)=\mu -y_n$. Therefore, by (\ref{nablaPhiN}),  on $\widetilde\Omega\cap {\rm Supp}(\nabla\Phi_N)$, we have,
\begin{eqnarray*}
\widetilde V_1(y_n)-(\nabla\Phi_N)^2 &\geq& \mu -y_n-(1-k/(S-\widetilde\phi_0)^{2/3})^2(\nabla\widetilde\phi_0)^2\\
&\geq& (\mu -y_n)\left( \frac{2k}{(S-\widetilde\phi_0)^{2/3}}-\frac{k^2}{(S-\widetilde\phi_0)^{4/3}}\right).
\end{eqnarray*}
Since, on this set, we also have $k(S-\widetilde\phi_0)^{-2/3}\leq k^{1/3}$ and $S-\widetilde\phi_0 ={\mathcal O}((\mu -y_n)^{3/2})$, we conclude that, for $h$ small enough,
$$
\widetilde V_1(y_n)-(\nabla\Phi_N)^2\geq \frac{k}{C_3} \quad \mbox{on } \widetilde\Omega\cap {\rm Supp}(\nabla\Phi_N),
$$
where $C_3>0$ is a large enough constant. On the other hand, away from ${\rm Supp}(\nabla\Phi_N)$, we have $\widetilde V_1(y_n)-(\nabla\Phi_N)^2 =\widetilde V_1(y_n)\geq  k^{2/3}\la y_n\ra /C_3$, and thus, in any case,
\be
\label{ellipp1}
\widetilde V_1(y_n)-(\nabla\Phi_N)^2\geq \frac{k\la y_n\ra}{C_3} \quad \mbox{on }  \{ y_n\geq \mu-\varepsilon\}.
\ee
Moreover, since $\nabla\widetilde\phi_0 =0$ on $\{y_n=\mu\}$, while $V_2> 0$ there, one easily checks that, if $\varepsilon$ and $\delta$ are chosen sufficiently small,  we have,
\be
\label{ellipp2}
\widetilde p_2(y,i\nabla\Phi_N(y)) \geq c_\mu \quad {\rm on} \quad \{y_n\geq\mu-\varepsilon -\delta\},
\ee
where $c_\mu >0$ is a constant.
\vskip 0.2cm
On the other hand, on $\{ y_n\leq \mu -\varepsilon\}$, using Lemma \ref{lemmephi0}, we see that,
$$
\Phi_N\leq {\rm Min}\{ S(\mu)-\delta_\mu , \widetilde\phi +4Nk\},
$$
 for some $\delta_\mu >0$ independent of $h$. Therefore, by using Proposition \ref{estile}, we obtain,
\be
\label{ellipp3}
\Vert h\nabla (e^{\Phi_N/h}v)\Vert_{\{ y_n\leq \mu -\varepsilon\}} + \Vert e^{\Phi_N/h}v\Vert_{\{ y_n\leq \mu -\varepsilon\}} ={\mathcal O}(1).
\ee
Going back to (\ref{presyn=mu}), and using (\ref{ellipp1})-(\ref{ellipp3}) (plus the fact that $\widetilde V_1 = |y_n| +\mu$ for $y_n\leq 0$), we finally obtain,
\begin{eqnarray*}
&& 
\Vert h\nabla (e^{\Phi_N/h}v)\Vert^2
+\Vert e^{\Phi_N/h}v_2\Vert^2 +k\Vert \sqrt{\la y_n}\ra e^{\Phi_N/h}v_1\Vert^2 \\
&&  \leq Ch\Vert  e^{\Phi_N/h}v_1\Vert^2+Ch k^{-1/3}\Vert  e^{\Phi_N/h}v_2\Vert_{\widetilde\Omega_{\delta}}^2\\
&&\hskip 4cm +Ch \Vert  e^{\Phi_N/h}v_2\Vert^2
+Ch^{-2(N+N_0)}+C, 
\end{eqnarray*}
where $C>0$ is a constant (that may depend on $N$ and $\mu$). Since $k>> h$ and $h k^{-1/3}\rightarrow 0$ as $h\rightarrow 0_+$, we conclude,
\be
\Vert h\nabla (e^{\Phi_N/h}v)\Vert^2
+\Vert e^{\Phi_N/h}v_2\Vert^2 +k\Vert \sqrt{\la y_n}\ra e^{\Phi_N/h}v_1\Vert^2 ={\mathcal O}(h^{-2(N+N_0)}),
\ee
uniformly for $h>0$ small enough. 
As before, using iteratively the equation,  this permits us to obtain,
\be
\label{finalestvtilde}
\Vert e^{\Phi_N/h}v_2\Vert_{H^s} +\Vert \la y_n\ra^s e^{\Phi_N/h}v_1\Vert_{H^s} ={\mathcal O}(h^{-(N+N_s)}),
\ee
 for any $s\in\R$, where $N_s\geq 0$ does not depend on N.

Now, by Lemma \ref{lemmephi0},  
for any $M\geq 1$,  on $[ \mu -(Mk)^{2/3}, \mu)$ we have $
\widetilde\phi_0 \geq S(\mu) -\frac23Mk$ .
Therefore, taking $M=N':= 3N/2$, we obtain $\widetilde\phi_0 \geq S(\mu) -Nk$, and thus (still by construction),
$$
\Phi_N = S(\mu) +Nk \,\,\mbox{ for }\,\, y_n\geq \mu -(N'k)^{2/3}.
$$
In particular, $e^{\Phi_N/h}\geq h^{-N}e^{S(\mu) /h}$ there, and thus, we deduce from (\ref{finalestvtilde}),
\begin{eqnarray}
\Vert v_2\Vert_{H^s(y_n\geq \mu -(N'k)^{2/3})} +\Vert \la y_n\ra^s v_1\Vert_{H^s(y_n\geq \mu -(N'k)^{2/3})}\nonumber\\ 
\label{estynv1}
={\mathcal O}(h^{-N_s}e^{-S(\mu )/h}),
\end{eqnarray}
Since $N'\rightarrow +\infty$ as $N\rightarrow +\infty$, the result of Proposition \ref{apriorivtilde} follows by standard Sobolev estimates.
\end{proof}

\begin{proposition} \sl
\label{apriori}
For any multi-index 
$\alpha$, there exists $N_\alpha\geq 0$, such that, for all $N\geq 1$, the first eigenfunction ${\mathbf v}$ of $Q$ satisfies,
\be
\label{u-W2}
\partial^\alpha ({\mathbf v} - {\mathbf w}) ={\mathcal O}( h^{-N_\alpha}e^{-S(\mu )/h})
\ee
locally uniformly in $ \{y_n\geq \mu  - (Nk)^{2/3}\}$.
\end{proposition}
\begin{proof} The proof is similar to that of Proposition \ref{estv-w1}. 

With the same notations, we see that $(Q-z)\widetilde{\mathbf v}=(\widetilde E -z)\widetilde{\mathbf v}+r$, with $r$ supported in $\{ y_n\geq \mu -k^{2/3}$\}. 
Then, using (\ref{estynv1}), we find that, for any $s\geq0$, $\Vert r\Vert_{H^s} ={\mathcal O}(h^{-N_s}e^{-S(\mu)/h})$ for some $N_s\geq 0$, and we deduce that $\Vert {\mathbf v}_0 -\widetilde{\mathbf v}\Vert ={\mathcal O}(h^{-N'_0}e^{-S(\mu)/h})$ for some $N'_0\geq 0$. In particular, $\Vert {\mathbf v}_0 \Vert =1+{\mathcal O}(h^{-N'_0}e^{-S(\mu)/h})$,  thus ${\mathbf v} =\lambda {\mathbf v}_0$ with $\lambda =1+{\mathcal O}(h^{-N'_0}e^{-S(\mu)/h})$, and the result follows by applying Proposition \ref{apriorivtilde}.
\end{proof}

\subsubsection{Propagation}
Now, we are able  to prove,
\begin{proposition}\sl
\label{asymptdev}
For any multi-index 
$\alpha$, there exists $\delta_\alpha > 0$, such that, for all $N\geq 1$, the first eigenfunction ${\mathbf v}$ of $Q$ satisfies,
\be
\label{u-W3}
\partial^\alpha ({\mathbf v} - {\mathbf w}) ={\mathcal O}( h^{\delta_\alpha N}e^{-S(\mu )/h})
\ee
uniformly with respect to $h>0$ small enough, and locally uniformly with respect to $y$ in $\{ \mu  - (Nk)^{2/3}\leq y_n\leq  \mu  + (Nk)^{2/3}\}$.
\end{proposition}
\begin{proof} We proceed in a way very similar to that of \cite{FuLaMa} (but in a simpler context, here). At first, for $\nu >0$, we introduce the Bargmann transform $T_\nu$, defined by the formula,
$$
T_\nu u(y,\eta;h) := \int_{\R^n} e^{i(y-z)\eta /h - (y'-z')^2/2h - \nu (y_n-z_n)^2/2h}u(z) dz.
$$
We also set $q_1(y,\eta ):= \eta^2 + \mu -y_n$ and, for $y'\in \R^{n-1}$ and $t\in\R$,  
$$
\gamma (t, y') 
:=\exp tH_{q_1}(y',\mu ; 0)=(y',\mu +t^2 ; 0,-t).
$$
By the same arguments as at the beginning of Section \ref{intboule}, we see that the estimate (\ref{prespuits}) remains  valid in any complex neighborhood of 0 of the form: $\{ y\in \C^n\, ;\, |\Re y| \leq 1/C; |\Im y|\leq C\sqrt h \}$, with a large enough constant $C>0$. Moreover,  because of the quadratic behaviour of $\widetilde\varphi_2(y)$ near 0, we also have $\Vert {\mathbf w}(y+i\delta)\Vert_{L^2_{loc}} ={\mathcal O}(1)$ for any $\delta ={\mathcal O}(\sqrt h )$. As a consequence, the same is true for $\widetilde{\mathbf v}$, and we see that all the previous arguments can be extended to $\Im y_n = -\delta (h)$, if $\delta (h)$ satisfies,
$$
0\leq \delta (h)\leq C\sqrt h.
$$
In particular,  we obtain, 
\be
\label{aprioripourv}
\partial^\alpha {\mathbf v} (y) ={\mathcal O}( h^{-N_\alpha}e^{-S(\mu )/h}),
\ee
locally uniformly in $ \{\Re y_n\geq \mu \, ,\, -C\sqrt h\leq \Im y_n \leq 0\}$. (Observe that, since $\widetilde\varphi_1$ and $\widetilde\varphi_2$ are real on the real, one has $\Re\widetilde\varphi_j(y', y_n-i\delta(h))= \widetilde\varphi_j(y) +{\mathcal O}(\delta(h)^2)$.) Therefore, taking $\nu =1$ and making a complex change of contour in the expression of $T_1 \mathbf v$, it is easy to see that, if $(y,\eta)\in \R^{2n}$ is fixed (independently of $\mu$), with $\eta_n >0$ and $y_n >\mu$, then,
$$
T_1 \mathbf v ={\mathcal O}(e^{ -1/\sqrt h}e^{-S(\mu) /h})={\mathcal O}(h^\infty)e^{-S(\mu) /h},
$$
uniformly near $(y,\eta)$. In particular, for any $t\leq-1$ and $y'\in\R^{n-1}$, we have,
$$
\gamma (t, y') \notin FS (e^{S(\mu)/h}\mathbf v),
$$
where $FS(e^{S(\mu)/h}\mathbf v)$ stands for the frequency set of $e^{S(\mu)/h}\mathbf v$ (see, e.g., \cite{GuSt, Ma6}). By the standard result of propagation of the frequency set, we deduce that $\gamma (t, y')$ stays away from $FS (e^{S(\mu)/h}\mathbf v)$ as long as the   {\it a priori} estimate (\ref{aprioripourv}) remains valid in some fix neighborhood of $\pi_y\gamma (t, y'):= (y', \mu +t^2)$, that is, as long as $t$ remains (strictly) negative. In other words, for any fixed $t<0$ (arbitrarily close to 0), one has $\gamma (t, y') \notin FS (e^{S(\mu)/h}\mathbf v)$. 
\vskip 0.2cm
On the other hand, if one has chosen the correct critical point in the definition of ${\mathbf w}$ on $\{y_n\geq\mu\}$, one also has $\gamma (t, y') \notin FS (e^{S(\mu)/h}\mathbf w)$ for any $t<0$ close enough to 0. Therefore, $\gamma (t, y') \notin FS (e^{S(\mu)/h}({\mathbf v}-{\mathbf w})$, and thanks to Proposition  \ref{apriori}, we can apply the same argument of propagation used in \cite{FuLaMa} Section 6 (see also \cite{GrMa}), and conclude to the existence of a constant $\delta >0$, such that, for all $N\geq 1$,
\be
\label{estmicrolN}
T_1\widetilde v_N (\widetilde y, \widetilde\eta ;\widetilde h_N) ={\mathcal O}(e^{-\delta /\widetilde h_N})\,\, \mbox{ uniformly in } {\mathcal V}_N(\delta),
\ee
where we have set,
\begin{eqnarray*}
&& \widetilde v_N (\widetilde y):= e^{S(\mu)/h}({\mathbf v}-{\mathbf w})(y'_0+(Nk)^{\frac12}\widetilde y' , \mu +(Nk)^{\frac23}\widetilde y_n)\\
&& \widetilde h_N := (Nk)^{-1} h;\\
&&{\mathcal V}_N(\delta):=\{ (\widetilde y, \widetilde \eta)\in\R^{2n}\,;\, |\widetilde y|\leq \delta\, ,\, (Nk)^{\frac16}|\widetilde\eta'|+|\widetilde\eta_n|\leq \delta\}.
\end{eqnarray*}
Here, $y'_0\in\R^{n-1}$ is fixed arbitrarily, and the estimate is uniform with respect to $h>0$ small enough and locally uniform with respect to $y'_0$ (see \cite{FuLaMa}, Proposition 6.8).

Note that, respect to \cite{FuLaMa}, here the situation is slightly simpler, because the asymptotic solution $\mathbf w$ exists in a whole $h$-independent neighborhood of $(y',\mu)$. This is the reason why the estimate (\ref{estmicrolN}) is valid for all $N$ large enough, while that of  \cite{FuLaMa}, Proposition 6.8 was valid for special values of $N$ only.

As in \cite{FuLaMa}, the main point  of the proof consists in observing that the function $ \widetilde v_N$ is solution of an equation of the form,
\be
\label{eqvN}
\widetilde{\mathbf Q}_N (\widetilde y, \widetilde h_ND_{\widetilde y}; \widetilde h_N) \widetilde v_N=0,
\ee
where the symbol $ \widetilde{\mathbf q}_N$ of $\widetilde{\mathbf Q}_N$ satisfies,
\begin{eqnarray*}
\widetilde{\mathbf q}_N (\widetilde y, \widetilde \eta ; \widetilde h_N) = \left((Nk)^{1/3}|\widetilde \eta '|^2 + \widetilde \eta_n^2\right)\left(
\begin{array}{cc}
-1 & 0\\
0 & 1
\end{array}
\right)
 + \left(
\begin{array}{cc}
\widetilde y_n & 0\\
0 & \check V_N(\widetilde y, \widetilde \eta)
\end{array}
\right)\\
+{\mathcal O}(\widetilde h_N (Nk)^{1/3}),
\end{eqnarray*}
with $\Re \check V_N>0$ near ${\mathcal V}_N(\delta)$.

Using this equation, and in particular the fact that $\det \widetilde{\mathbf q}_N (\widetilde y, \widetilde \eta ; \widetilde h_N) \not=0$ in $\{ |\widetilde y|\leq \delta'\, ;\, (Nk)^{1/3}|\widetilde \eta '|^2 + \widetilde \eta_n^2\geq\delta\}$ if $\delta' >0$ is small enough ($\delta >0$ fixed), we deduce as in \cite{FuLaMa} that, for all $m\geq 0$, one has,
\be
\label{finalestv}
\Vert \widetilde v_N\Vert_{ H^m (|\widetilde y|\leq \frac12 \delta')} ={\mathcal O}(e^{-\delta_1/\widetilde h_N})={\mathcal O}(h^{\delta_1N}),
\ee
for some  constant $\delta_1 >0$. Then, the result follows from Sobolev estimates and from the fact  that the previous estimate is locally uniform with respect to $y'_0\in\R^{n-1}$.
\end{proof}
\vskip 0.3cm

Now, recall from Section \ref{sectionAgmon} that the resonant state $\widetilde{\mathbf u}$, solution of $\widetilde P \widetilde{\mathbf u} =\rho \widetilde{\mathbf u}$, is related to ${\mathbf v}$ by,
$$
\widetilde{\mathbf u}(y) = {\mathbf v}(y', y_n +ik).
$$
We have,
\begin{proposition}\sl  
\label{globalestu}
For any $\varepsilon >0$ small enough, set,
$$
\phi_\varepsilon := {\rm Max}\left( \widetilde\phi_0, \min(\phi, S(\mu) -\varepsilon)\right),
$$
where $\widetilde \phi_0(y):= \chi_0(y_n)\left(S(\mu) -\frac23 (\mu -y_n)_+^{3/2}\right)$, with $\chi_0\in C^\infty (\R;[0,1])$, $\chi_0 =1$ on $[\mu -\varepsilon_0 +\infty)$, $\chi_0=0$ on $(-\infty, \mu -2\varepsilon_0]$, $\varepsilon_0>0$ fixed sufficiently small.
Then, for any $\alpha\in\N^n$, there exists $N_\alpha\geq 0$, such that, for any $\varepsilon >0$,
$$
\partial^\alpha \widetilde{\mathbf u}(y) ={\mathcal O}(h^{-N_\alpha} e^{-\phi_\varepsilon (y)/h}),
$$
locally uniformly on $\R^n$.\end{proposition}
\begin{proof} We first observe that, by (\ref{estv1}) and (\ref{finalestvtilde}), the result is true for the function $ \widetilde{\mathbf v}- {\mathbf w}$ (instead of $\widetilde{\mathbf u}$). Moreover, by construction and the proof of Lemma \ref{lemmephi0}, we see that it is also true for ${\mathbf w}$.  
Therefore, it is true for $ \widetilde{\mathbf v}$, and the proofs of Propositions \ref{estv-w1} and \ref{apriori} permit us to deduce that it is also true for ${\mathbf v}$. 

Then, we use the elementary fact that,
\be
\label{linkTuTv}
(T_1\widetilde{\mathbf u})(y,\eta)= (T_1{\mathbf v}) (y', y_n+ik; \eta) = e^{k^2/2h +k\eta_n/h}(T_1 {\mathbf v})(y; \eta', \eta_n+k),
\ee
and, since $\Vert {\mathbf v}\Vert_{L^2}=1$, we also know that $\Vert T_1{\mathbf v}\Vert_{L^2}$ is uniformly bounded  (actually, it is ${\mathcal O}(h^{\frac{3n}4})$). Therefore, $e^{-k\eta_n/h}T_1 \widetilde{\mathbf u}\in L^2(\R^{2n})$, and,
\be
\label{estaprioriTu}
\Vert e^{-k\eta_n/h}T_1 \widetilde{\mathbf u}\Vert_{L^2}={\mathcal O}(1),
\ee
uniformly as $h\rightarrow 0_+$. In particular, for any $\varepsilon >0$, $e^{-\varepsilon k^{2/3}\la\eta_n\ra/h}T_1 \widetilde{\mathbf u}\in L^2(\R^{2n})$.

Then, using the classical ellipticity (that is, as $|\eta|\rightarrow +\infty$) of the symbol of $\widetilde P$, together with the analyticity of $V_2$, one can show,
\begin{lemma}
\label{estuetalarge}
For any $\chi\in C_0^\infty (\R^n)$ and for any $M\geq 1$, there exists a constant $C>0$ independent of $\mu$ such that,
$$
\sup_{|\eta|\geq C} \left| \la \eta\ra^M \chi (y)T_1\widetilde{\mathbf u}\right| ={\mathcal O}(e^{-1/Ch}),
$$
uniformly for $h>0$ small enough.
\end{lemma}
\begin{proof}
See Appendix \ref{appF}.
\end{proof}
Now, for $y$ in some fixed arbitrary compact set, we write,
\begin{eqnarray*}
\widetilde{\mathbf u}(y) &=& \frac1{(2\pi h)^n}\int_{\R^{2n}}e^{i(y-z)\eta /h - (y-z)^2/2h}\widetilde{\mathbf u}(z) dzd\eta\\
&=& \frac1{(2\pi h)^n}\int_{\R^n}T_1\widetilde{\mathbf u}(y,\eta)d\eta,
\end{eqnarray*}
and thus, by Lemma \ref{estuetalarge} and (\ref{linkTuTv}),
\begin{eqnarray}
&& \widetilde{\mathbf u}(y)\nonumber\\
&& =\frac1{(2\pi h)^n}\int_{|\eta|\leq C}e^{k^2/2h +k\eta_n/h}T_1 {\mathbf v}(y; \eta', \eta_n+k)d\eta +{\mathcal O}(e^{-1/Ch})\nonumber\\
&& =\frac1{(2\pi h)^n}\int_{|\eta|\leq C}e^{k^2/2h +k\eta_n/h+i(y-z)\eta/h +ik(y_n-z_n)/h-(y-z)^2/2h} {\mathbf v}(z)dzd\eta\nonumber\\
\label{repruparv}
 &&\hskip 7cm+{\mathcal O}(e^{-1/Ch})
\end{eqnarray}

In this integral, we perform the (singular) complex change of contour of integration given by,
\be
\label{chcontoureta1}
\R^n\ni\eta\mapsto \eta +i\delta \chi (|\eta|)\frac{y-z}{|y-z|} \in \C^n,
\ee
with $\delta >0$ small enough independent of $\mu$, and where $\chi\in C_0^\infty ([0,C); [0.1])$, $\chi=1$ on $[0,C/2]$. We obtain,
$$
\widetilde {\mathbf u}(y)= \int_{|\eta|\leq C} {\mathcal O}(h^{-n}e^{k\eta_n/h-\delta \chi (|\eta|) (|y-z|/h - (y-z)^2/2h }){\mathbf v}(z) dzd\eta +{\mathcal O}(e^{-1/Ch}).
$$
Then, taking advantage that we are in dimension $n\geq 2$, we use the polar coordinates in $\eta$: $\eta =r\omega$ ($r>0$, $\omega\in S^{n-1}$), and since the integrated function is analytic with respect to $\omega$, we can perform the (singular) complex change of contour of integration given by,
\be
\label{chcontoureta2}
S^{n-1}\ni\omega\mapsto \omega +i\delta \frac{y-z}{|y-z|} \in \C^n.
\ee
 This gives,
$$
\widetilde {\mathbf u}(y)= \int_{|\eta|\leq C} {\mathcal O}(h^{-n}e^{k\eta_n/h-\delta(\chi_1 (|\eta|))|y-z|/h - (y-z)^2/2h }){\mathbf v}(z) dzd\eta +{\mathcal O}(e^{-1/Ch}),
$$
with $\chi_1(|\eta|):=|\eta| +\chi (|\eta|)\geq 1$,
and thus,
\be
\label{lienuv}
\widetilde {\mathbf u}(y)= \int  {\mathcal O}(h^{-n-C}e^{-\delta|y-z|/h  }){\mathbf v}(z) dz +{\mathcal O}(e^{-1/Ch}).
\ee
Finally, we observe that, if $\varepsilon $ is sufficiently enough, the Lipshitz function $\phi_\varepsilon$ is such that 
$\Vert\nabla \phi_\varepsilon\Vert_{L^\infty} ={\mathcal O}(\sqrt \mu)$. Therefore, there exists a $\mu$-independent constant $C>0$ such that,
$$
-\delta'|y-z| +\phi_\varepsilon (y) -\phi_\varepsilon (z) \leq -(\delta' -C\sqrt \mu)|y-z|,
$$
and thus, since $\delta'>0$ can be taken independently of $\mu$, we deduce from (\ref{lienuv}) and from the estimate on ${\mathbf v}$ that, for $\mu$ sufficiently small, we have,
$$
\widetilde {\mathbf u}={\mathcal O}(h^{-N_0}e^{-\phi_\varepsilon /h})
$$
uniformly for $h>0$ small enough, and with some constant $N_0\geq 0$.

 The estimates on the derivatives of $\widetilde {\mathbf u}$ can be done in the same way.
\end{proof}

\vskip 0.2cm
 Now, we denote by $\widetilde{\mathbf w}$ the asymptotic solution constructed in Section \ref{sectionWKB} on the real domain. Then, we claim,
\begin{proposition}\sl
\label{estaprioriu-w} For any $\alpha\in \N^n$, there exists $\delta_\alpha > 0$, such that, for all $N\geq 1$, the resonant state $\widetilde{\mathbf u}$ of $\widetilde P$ satisfies,
\be
\label{u-W1}
\partial^\alpha (\widetilde{\mathbf u}-\widetilde{\mathbf w}) ={\mathcal O}( h^{\delta_\alpha N}e^{-S(\mu )/h})
\ee
uniformly with respect to $h>0$ small enough, and locally uniformly with respect to $y$ in $\{ | y_n-\mu|\leq  (Nk)^{2/3}\}$.
Moreover, for any $\alpha\in \N^n$, there exists $N_\alpha\geq 0$ such that, for all $N\geq 1$, one has,
$$
\partial^\alpha\widetilde{\mathbf u} ={\mathcal O}(h^{-N_\alpha}e^{-S(\mu )/h}),
$$
locally uniformly in $\{ y_n \geq \mu \}$.
\end{proposition}
\begin{proof}
We use Proposition \ref{apriori} and the representation of $\widetilde{\mathbf u}$ given in (\ref{repruparv}). After the changes of contour (\ref{chcontoureta1})-(\ref{chcontoureta2}), we obtain,
\begin{eqnarray*}
\widetilde{\mathbf u}(y)=\frac{e^{k^2/2h}}{(2\pi h)^n}\int_{|\eta|\leq C} e^{ i\theta/h +k\eta_n/h-\delta \chi_1(|\eta|)|y-z|- (y-z)^2/2h }{\mathbf v}(z) J(y,z,\eta)dz d\eta\\
+{\mathcal O}(e^{-1/h}),
\end{eqnarray*}
with $\theta := \delta k \chi_1(|\eta|)\frac{y_n-z_n}{|y-z|}+(y-z)\eta +k(y_n-z_n)$, and $J(y,z,\eta)={\mathcal O}(1)$.

On $\{ z_n\leq\mu -\varepsilon\}$ (with $\varepsilon >0$ fixed small enough), the estimates on ${\mathbf v}$ show that, for any $\varepsilon'>0$, one has,
$$
e^{-\delta |y-z|/h}{\mathbf v}(z)={\mathcal O}(h^{-N_0}e^{-\delta |y-z|/h -\min(\phi, S-\varepsilon')/h},
$$
and thus, if in addition $y_n\geq \mu-(Nk)^{2/3}$ (and using the fact that $\Vert\nabla\phi \Vert_{L^\infty}={\mathcal O}(\mu)$), for $\varepsilon '$ small enough, on this set we obtain,
$$
e^{-\delta |y-z|/h}{\mathbf v}(z)={\mathcal O}(h^{-N_0}e^{-(S/h +\delta \varepsilon/2h -\varepsilon'/h)})={\mathcal O}(e^{-(S+\delta \varepsilon/4)/h)}).
$$
Moreover, on $\{ \mu -\varepsilon\leq z_n\leq \mu -2(Nk)^{2/3} \}$, we have,
$$
e^{-\delta |y-z|/h}{\mathbf v}(z)={\mathcal O}(h^{-N_0}e^{-\delta |y-z|/h -S/h+2(\mu -z_n)^{3/2}/h}),
$$
and thus, for $y_n\geq \mu-(Nk)^{2/3}$,
$$
e^{-\delta |y-z|/h}{\mathbf v}(z)={\mathcal O}(h^{-N_0}e^{-\delta (Nk)^{2/3}/2h -S/h})={\mathcal O}(h^\infty e^{-S/h}).
$$
As a consequence, for $\widetilde{\mathbf u}(y)$ with $y_n\geq \mu-(Nk)^{2/3}$, we obtain,
\vskip 0.2cm

$\widetilde{\mathbf u}(y)$
\begin{eqnarray*}
=\frac{e^{k^2/2h}}{(2\pi h)^n}\int_{{|\eta|\leq C}\atop{z_n\geq \mu -2(Nk)^{2/3}}} e^{ i\theta/h +k\eta_n/h-\delta \chi_1(|\eta|)|y-z|- (y-z)^2/2h }{\mathbf v}(z) J(y,z,\eta) dz d\eta\\
+{\mathcal O}(h^{\infty}e^{-S/h}).
\end{eqnarray*}
In particular, for $y_n\geq\mu$, by Proposition \ref{apriori} we deduce,
$$
\widetilde{\mathbf u}(y)={\mathcal O}(h^{-N_0}e^{-S(\mu)/h}.
$$
Moreover, if $|y_n-\mu|\leq (Nk)^{2/3}$, the previous  arguments lead to,
\vskip 0.2cm

$\widetilde{\mathbf u}(y)$
\begin{eqnarray*}
=\frac{e^{k^2/2h}}{(2\pi h)^n}\int_{{|\eta|\leq C}\atop{|z_n-\mu|\leq 2(Nk)^{2/3}}} e^{ i\theta/h +k\eta_n/h-\delta \chi_1(|\eta|)|y-z|- (y-z)^2/2h }{\mathbf v}(z) J(y,z,\eta) dz d\eta\\
+{\mathcal O}(h^{\infty}e^{-S/h}).
\end{eqnarray*}

Then, using Proposition \ref{asymptdev}, we conclude,
\vskip 0.2cm
$\widetilde{\mathbf u}(y)$
\begin{eqnarray*}
=\frac{e^{k^2/2h}}{(2\pi h)^n}\int_{{|\eta|\leq C}\atop{|z_n-\mu|\leq 2(Nk)^{2/3}}} e^{ i\theta/h +k\eta_n/h-\delta \chi_1(|\eta|)|y-z|- (y-z)^2/2h }{\mathbf w}(z) J(y,z,\eta) dz d\eta\\
+{\mathcal O}(h^{\delta N}e^{-S/h}),
\end{eqnarray*}
for some $\delta >0$ independent of $N$. Finally, taking $y'$ close to some $y_0'\in\R^{n-1}$, either ${\mathbf w}$ is analytic around $(y_0', y_n)$, or it is exponentially smaller than $e^{-S(\mu)/h}$. In the second case, we obtain $\widetilde{\mathbf u}(y)={\mathcal O}(h^{\delta N}e^{-S/h})$, while, in the first case, a change of contour of the type $z_n\mapsto z_n+ik\chi (|z_n -\mu|(Nk)^{-2/3})$ gives, by the same arguments as before, $\widetilde{\mathbf u}(y)=\widetilde{\mathbf w}(y)+{\mathcal O}(h^{\delta N}e^{-S/h})$. The same estimates holds for the derivatives, and this proves the proposition.
\end{proof}

\subsection{All around $\{\phi =S(\mu)\}$}
We fix $\nu_0>0$ 
 sufficiently small, but independent of $\mu$, and, for $\nu \in (0, \nu_0/4]$, we consider the following spherical domain of $\R^n$,
$$
{\mathcal A}_\nu:= \{ \nu_0-\nu <|y| <\nu_0+\nu\}.
$$
In this section, we plan to estimate $\widetilde{\mathbf u}$ on ${\mathcal A}_\nu$. By construction, $\widetilde{\mathbf w}=0$ in ${\mathcal A}_{2\nu}$ if $\mu<< \nu_0$, and then we also have  ${\mathcal A}_{2\nu}\subset \R^n\backslash \{\phi < S(\mu)\}$. Therefore, Propositions \ref{globalestu} and \ref{estaprioriu-w} already tell us that,
\begin{itemize}
\item $\partial^{\alpha}\widetilde u ={\mathcal O}(h^{-N_\alpha} e^{-S(\mu)/h})$ on ${\mathcal A}_{2\nu} \cap \{ y_n\geq \mu \}$;
\item $\partial^{\alpha}\widetilde u ={\mathcal O}(h^{\delta N} e^{-S(\mu)/h})$ on ${\mathcal A}_{2\nu} \cap \{ |y_n-\mu|\leq 4(Nk)^{2/3}\}$;
\item For any $\varepsilon >0$, $\partial^{\alpha}\widetilde u ={\mathcal O}(h^{-N_\alpha} e^{-\max (\widetilde\phi_0, S(\mu) -\varepsilon)/h})$ on ${\mathcal A}_{2\nu} \cap \{ y_n\leq \mu - (Nk)^{2/3}\}$,
\end{itemize}
where $N\geq 1$ is arbitrary, and the positive constants $\delta, N_\alpha$ do not depend on $N$ (note that the same estimates also hold for ${\mathbf v}$). Then, recording that $S(\mu)={\mathcal O}(\mu^2)$, we consider a weight function $\psi \in C^\infty (\R^n; \R_+)$ such that,
\begin{itemize}
\item $\psi$ is supported in ${\mathcal A}_{2\nu}$;
\item $\psi \leq S(\mu)$ on $\{ y_n\geq \mu +4(Nk)^{2/3}\}$;
\item $\psi \leq S(\mu) + \delta N k$ on $\{ y_n\leq \mu +4(Nk)^{2/3}\}$;
\item $\psi = S(\mu)$ on $\{ y_n\geq \mu +4(Nk)^{2/3}\}\cap {\mathcal A}_\nu$;
\item $\psi = S(\mu) + \delta N k$ on $\{ y_n\leq \mu +2(Nk)^{2/3}\}\cap {\mathcal A}_\nu$;
\item $|\nabla\psi| \leq \mu$ on $\{y_n\leq \mu+2(Nk)^{2/3}\}\cup \{y_n\geq \mu+4(Nk)^{2/3}\}$;
\item $\partial^\alpha \psi ={\mathcal O}(1)$ on $\{y_n\leq \mu+2(Nk)^{2/3}\}\cup \{y_n\geq \mu+4(Nk)^{2/3}\}$, $ (|\alpha|\geq 2$);
\item $\partial^\alpha \psi ={\mathcal O}((Nk)^{1-2|\alpha|/3})$ on $\{ \mu+2(Nk)^{2/3}\leq y_n\leq  \mu+4(Nk)^{2/3}\}$, $ (|\alpha|\geq 1$).
\end{itemize}
Then, as before, the Agmon estimates applied with this weight function $\psi$ lead to $
\partial^\alpha\widetilde u ={\mathcal O}(h^{-N_\alpha}e^{-\psi /h})$,
and thus, in particular, for any $N\geq 1$, one has,
\be
\partial^\alpha\widetilde u ={\mathcal O}(h^{\delta N}e^{-S(\mu) /h})
\ee
uniformly on $\{ y_n\leq \mu +2(Nk)^{2/3}\}\cap {\mathcal A}_\nu$.

\section{Completion of the Proof}
We fix $\nu \in (0, \nu_0/4]$ small enough, independent of $\mu$, and $\chi_\nu\in C^\infty (\R; [0,1])$, such that $\chi_\nu =1$ on $(-\infty,\nu_0-\nu]$, $\chi_\nu =0$ on $[\nu_0+\nu, +\infty)$. We also fix $\chi_0\in C^\infty ( \R:[0,1])$ such that $\chi_0=1$ on $(-\infty, 1]$, $\chi_0 =0$ on $[2,+\infty)$. Then, for $y=(y',y_n)\in\R^n$ and $N\geq 1$, we set,
$$
\chi_N(y):= \chi_\nu (|y|)\chi_0 \left( \frac{y_n-\mu}{(Nk)^{2/3}}\right).
$$
In particular, the support of $\nabla\chi_N$ is included in the set.
\begin{eqnarray*}
{\mathcal B}_N&:=&\left({\mathcal A}_\nu\cap \{ y_n-\mu \leq 2(Nk)^{2/3}\}\right)\\
&& \cup \left(\{(Nk)^{2/3}\leq y_n-\mu\leq 2(Nk)^{2/3}\}\cap \{ |y|\leq \nu_0+\nu\}\right).
\end{eqnarray*}
Then, we write,
$$
(\Im\rho )\Vert\chi_N\widetilde{\mathbf u}\Vert^2 = \Im \la \chi_N \widetilde P \widetilde{\mathbf u}, \chi_N \widetilde{\mathbf u}\ra =-\Im\la \chi_N[\widetilde P, \chi_N] \widetilde{\mathbf u}, \widetilde{\mathbf u}\ra,
$$
and, at this point, we can readily follow the arguments of \cite{GrMa}, Section 8 (with, in our case, $n_0=1$ and $n_\Gamma =0$), and Theorem \ref{mainth} follows.

\bigskip

\appendix
\vskip 1cm
\centerline{\bf APPENDIX}

\bigskip

\section{Proof of Lemma \ref{symbOh}}\label{appB}
Setting $F(\theta, y ):= \Psi (y+\frac{\theta}2, y-\frac{\theta}2)$, we write,
\be
\label{decompa}
a(y, \eta'+iF(\theta, y )+if_1)=a(y, \eta'+i\nabla\Phi (y)+if_1) + \widetilde a(y,\eta',\theta),
\ee
where,
\be
\label{atilde}
\widetilde a(y,\eta',\theta)=i\theta  \int_0^1(\nabla_\theta F)(t\theta, y )\cdot (\nabla_{\eta}a)(y, \eta'+iF(t\theta, y )+if_1)\, dt.
\ee
Now, by assumption (\ref{estreg}),  $F$ satisfies,
\begin{eqnarray*}
&& |F| \leq \nu_0;\\
&&  \partial^\alpha F ={\mathcal O}\left((1+\lambda_h{\bf 1}_{{y\in \Omega_\delta}\atop{|\theta| <\delta}}) h^{1-|\alpha|}+ \lambda_hh^{2-|\alpha|}\right)\quad (|\alpha|\geq 1),
\end{eqnarray*}
where $\delta >0$ is  arbitrarily small.

In the following, for any set ${\mathcal E}$, we set,
$$
\Lambda_h({\mathcal E}):= 1+\lambda_h{\bf 1}_{\mathcal E},
$$
where, as before, ${\bf 1}_{\mathcal E}$ stands for the characteristic function of ${\mathcal E}$.

Then, we see on (\ref{atilde}) that, for any  $\alpha, \beta, \gamma\in\N^{n}$, we have,
\begin{eqnarray}
\partial_{\theta}^\alpha \partial_y^\beta\partial_{\eta'}^\gamma \widetilde a(y,\eta',\theta) &=&{\mathcal O}\left( \la \eta'\ra^{2m}((|\theta|+h) \Lambda_h(|\theta| <\delta\, ;\, y\in\Omega_\delta) h^{-|\alpha|-|\beta|} \right)\nonumber\\
\label{estatilde} && + {\mathcal O}\left( \la \eta'\ra^{2m}((|\theta|+h)\lambda_h h^{1-|\alpha|-|\beta|} \right).
\end{eqnarray}
(Note that in the factor $(|\theta|+h)$, the term in $h$ appears if $\alpha\not= 0$ only, and that the validity of the estimate extends to $(y,\eta')\in \R^n\times \C^n$, $|\Im \eta'|$ small enough.)

Now, using (\ref{decompa}), we can split $b(y,\eta)$ into,
\be
\label{decomb}
b(y,\eta)=b_0(y,\eta) + b_1(y,\eta) + b_2(y,\eta),
\ee
with,
\vskip 0.2cm
$
b_0(y,\eta) :=  \frac1{(2\pi h)^n}\int e^{i(\eta'-\eta )\theta /h-g_1/h}a(y, \eta'+i\nabla\Phi (y)+if_1)d\eta'd\theta ;$
\vskip 0.2cm
$b_1(y,\eta) :=  \frac1{(2\pi h)^n}\int e^{i(\eta'-\eta )\theta /h-g_1/h}(\chi_0(\theta )-1)a(y, \eta'+i\nabla\Phi (y)+if_1)d\eta'd\theta ;$
\vskip 0.2cm
$
b_2(y,\eta) :=  \frac1{(2\pi h)^n}\int e^{i(\eta'-\eta )\theta /h-g_1/h}\chi_0(\theta )\widetilde a(y,\eta',\theta)d\eta'd\theta.
$
\vskip 0.4cm
Making the change of contour of integration $\eta'\mapsto \eta'-if_1$ in the expression of $b_0$, and using the fact that $\frac1{(2\pi h)^n}\int e^{i(\eta'-\eta )\theta /h}d\theta = \delta_{\eta'=\eta}$, we immediately obtain,
\be
\label{expb0}
b_0(y,\eta) = a(y, \eta+i\nabla\Phi (y)).
\ee
On the other hand, if we set,
$$
L:= (1+|\eta' -\eta|^2 +|\theta|^2)^{-1}(1+(\eta'-\eta)\cdot hD_\theta + \theta\cdot hD_{\eta'}),
$$
then, by integrations by parts, for any $N\geq 1$, we have,
\begin{eqnarray}
&& b_1(y,\eta) =  \frac1{(2\pi h)^n}\int e^{i(\eta'-\eta )\theta /h}c_N(y, \eta ,\eta' ,\theta)d\eta'd\theta ;\\
\label{b2}
&& b_2(y,\eta) =  \frac1{(2\pi h)^n}\int e^{i(\eta'-\eta )\theta /h}\widetilde c_N(y, \eta, \eta' ,\theta)d\eta'd\theta,
\end{eqnarray}
with,
\begin{eqnarray*}
&& c_N(y, \eta, \eta' ,\theta):=({}^tL)^N\left(e^{-g_1/h}(\chi_0(\theta )-1)a(y, \eta'+i\nabla\Phi (y)+if_1)\right);\\
&& \widetilde c_N(y, \eta, \eta' ,\theta):=({}^tL)^N\left(e^{-g_1/h}\chi_0(\theta )\widetilde a(y,\eta',\theta)\right).
\end{eqnarray*}
Therefore, using (\ref{estreg}), (\ref{estatilde}), and the fact that  $g\geq 2\varepsilon_0$ on the support of $1-\chi_0 (y-y')$, we see that, for any $\alpha, \beta, \gamma \in\N^n$, we have,
\begin{eqnarray*}
\partial_y^\alpha \partial_\eta^\beta\partial_{\eta'}^\gamma c_N(y, \eta' ,\theta)={\mathcal O}(e^{-\varepsilon_0/h}(1+|\eta' -\eta| +|\theta|)^{-N})\la\eta'\ra^{2m};\\
\partial_y^\alpha \partial_\eta^\beta\partial_{\eta'}^\gamma \widetilde c_N(y, \eta' ,\theta)={\mathcal O}\left( (\Lambda_h(|\theta| <\delta; y\in\Omega_\delta)+h\lambda_h)(|\theta|+h)  h^{-|\alpha|}\right) \\
\times  (1+|\eta' -\eta| +|\theta|)^{-N}\la\eta'\ra^{2m}),
\end{eqnarray*}
uniformly for $h>0$ small enough, $(\theta,  \eta ) \in\R^{2n}$, and $(y,\eta')\in\R^{n}\times \C^n$, $|\Im \eta'|$ small enough.

In particular, by choosing $N$ sufficiently large, we derive,
\be
\label{estb1}
 \partial_y^\alpha \partial_\eta^\beta b_1(y,\eta )={\mathcal O}(e^{-\varepsilon_0 /h}\la\eta\ra^{2m}).
\ee
 Concerning $b_2$, in (\ref{b2}) we make the (singular) change of contour of integration,
 $$
 \R^n\ni\eta' \mapsto \eta' +i\delta\frac{\theta}{|\theta|},
 $$
 where, as before, $\delta >0$ is taken sufficiently small. We obtain,
 \begin{eqnarray*}
  \partial_y^\alpha \partial_\eta^\beta b_2(y,\eta )=\frac1{(2\pi h)^n}\int {\mathcal O}\left( e^{-\delta  |\theta|/h} (\Lambda_h(|\theta| <\delta; y\in \Omega_\delta)+h\lambda_h) (|\theta| +h)\right)\\
  \times h^{-|\alpha|}\la\theta\ra^{-N}\la\eta\ra^{2m} d\theta
\end{eqnarray*}
 and thus,
 \be
 \label{estb2}
  \partial_y^\alpha \partial_\eta^\beta b_2(y,\eta )= {\mathcal O}(\Lambda_h( y\in \Omega_\delta) h^{1-|\alpha|}\la\eta\ra^{2m}).
  \ee
Then, Lemma \ref{symbOh} follows from (\ref{decomb}), (\ref{expb0}), (\ref{estb1}), and (\ref{estb2}).

\section{An exotic pseudofifferential calculus}\label{appA}
In this appendix we have gathered some useful results concerning the pseudodifferential operators $A={\rm Op}_h^W (a)$ with $a =a(x,\xi;h)$ smooth, such that $a$ extends to a holomorphic function of $\xi$ in a strip  ${\mathcal A}_\nu := \{ (x,\xi)\in \R^n \times \C^n\, ;\, |\Im\xi| < \nu\}$ with $\nu >0$, and, for all $\alpha,\beta\in\N^n$, it satisfies,
\begin{eqnarray}
\label{exotic}
&& \partial_\xi^\beta a(x,\xi) ={\mathcal O}(\la \Re\xi\ra^{m});\\
&& \partial_x^\alpha\partial_\xi^\beta a(x,\xi) ={\mathcal O}((1+\lambda_h{\bf 1}_\Omega(x)) h^{1-|\alpha|}\la \Re\xi\ra^{m})\quad (|\alpha|\geq 1),
\end{eqnarray}
uniformly for $(x,\xi)\in {\mathcal A}_\nu$ and $h>0$ small enough. Here, $m\in\R$ is some fix real number, $\Omega$ is some fix subset of $\R^n$, and $\lambda_h \geq 1$ is some given function of $h$ such that $h\lambda_h={\mathcal O}(1)$ uniformly.
\vskip 0.2cm
We denote by ${\mathcal S}_m(\lambda_h, \Omega)$ the space symbols $a$ such that the previous properties hold for some $\nu >0$. 
For $t\in[0,1]$, we also denote by ${\rm Op}_h^t (a)$ the $t$-quantization of $a$ (see, e.g., \cite{Ma6}).
\vskip 0.3cm
In order to describe the remainders terms that will appear in the symbolic calculus, we need to introduce another class of symbols. We say that a smooth symbol $a =a(x,\xi;h)$ is in $\widetilde {\mathcal S}_m(\lambda_h, \Omega)$ if there exists some $\nu >0$ such that $a$ extends to a holomorphic function of $\xi$ in   ${\mathcal A}_\nu$, and, for all $\alpha,\beta\in\N^n$, it satisfies,
\be
\label{exotictilde}
 \partial_x^\alpha\partial_\xi^\beta a(x,\xi) ={\mathcal O}((1+\lambda_h{\bf 1}_\Omega(x)) h^{-|\alpha|}\la \Re\xi\ra^{2m}),
\ee
uniformly for $(x,\xi)\in {\mathcal A}_\nu$ and $h>0$ small enough.
\vskip 0.3cm
In particular, one has,
$$
{\mathcal S}_m(\lambda_h, \Omega)\subset \widetilde {\mathcal S}_m(\lambda_h, \Omega)\subset h^{-1} {\mathcal S}_m(\lambda_h, \Omega).
$$
\begin{proposition}\sl 
\label{propA1}
With the previous notations, one has,
\begin{itemize}
\item[(i)] If $a\in {\mathcal S}_0(\lambda_h, \R^n)$, then ${\rm Op}_h^W (a)$ is uniformly bounded on $L^2(\R^n)$.
\item[(ii)] If $a\in {\mathcal S}_m(\lambda_h, \Omega)$, then, for any $t\in[0,1]$ any $\delta>0$, and any $u\in L^2$,  one has,
$$
\Vert\la hD_x\ra^{-\frac{m}2}({\rm Op}_h^W (a)-{\rm Op}_h^t (a))\la hD_x\ra^{-\frac{m}2}u\Vert ={\mathcal O}(h\Vert u\Vert + h\lambda_h\Vert u\Vert_{L^2(\Omega_\delta)}),
$$
where, as before, $\Omega_\delta:= \{ x\in\R^n\,;\, {\rm dist} (x,\Omega )\leq\delta\}$, and where the estimate is uniform with respect to $h$ and $u$.
\item[(iii)] Let $m,m'\in\R$,   $a\in {\mathcal S}_m(\lambda_h, \Omega)$, and $b\in {\mathcal S}_{m'}(\lambda_h, \Omega)$. Then, for any $\delta >0$,  $a\#^W b \in {\mathcal S}_{m+m'}(\lambda_h, \Omega_\delta)$, and one has,
$$
a\#^W b = ab +h c,
$$
with $c\in \widetilde {\mathcal S}_{m+m'}(\lambda_h, \Omega_\delta)$.
\item[(iv)] If $a\in \widetilde{\mathcal S}_{0}(\lambda_h, \Omega)$, then, for any $\delta >0$, ${\rm Op}_h^W (a)$ satisfies,
$$
\Vert {\rm Op}_h^W (a)u\Vert ={\mathcal O}\left( \Vert u\Vert + \lambda_h\Vert u\Vert_{L^2(\Omega_\delta)}\right),
$$
uniformly with respect to $u\in L^2(\R^n)$ and $h>0$ small enough.
\item[(v)] Let $m,m'\in\R$,   $a\in \widetilde {\mathcal S}_m(\lambda_h, \Omega)$, and $b\in {\mathcal S}_{m'}(\lambda_h, \Omega)$. Then, for any $\delta >0$,  $a\#^W b \in \widetilde{\mathcal S}_{m+m'}(\lambda_h, \Omega_\delta)$.
\item[(vi)] If $a\in {\mathcal S}_0(\lambda_h, \R^n)$ takes it values in $\R_+$, then, there exists a constant $C>0$ such that,
$$
{\rm Op}_h^W (a) \geq -C\sqrt{h\lambda_h}.
$$
\end{itemize}
\end{proposition}
\begin{remark}\sl No complete symbolic calculus is available for this exotic class of symbols, but just up to ${\mathcal O}(h\lambda_h )$.
\end{remark}
\begin{remark}\sl It results from (iii) and (v) that, if $a_j\in  {\mathcal S}_{m_j}(\lambda_h, \Omega)$ ($j=1,2,3$), then, for any $\delta >0$, the symbol $a_1\#^W a_2\#^W a_3$ can be written as,
\be
a_1\#^W a_2\#^W a_3=a_1a_2a_3+hb,
\ee
with $c\in \widetilde {\mathcal S}_{m_1+m_2+m_3}(\lambda_h, \Omega_\delta)$. Indeed, writing $a_1\#^W a_2=a_1a_2+hb_1$ with $b_1\in  {\mathcal S}_{m_1+m_2}(\lambda_h, \Omega_{\delta/2})$, it comes out,
$$
a_1\#^W a_2\#^W a_3=(a_1a_2)\#^W a_3+hb_1\#^W a_3,
$$
and since $a_1a_2\in {\mathcal S}_{m_1+m_2}(\lambda_h, \Omega)$ (see (\ref{derivprod1})-(\ref{derivprod2})), the result follows from a new application of (iii), and from (v). In particular, by (i) and (iv), this implies,
$$
{\rm Op}_h^W (a_1){\rm Op}_h^W (a_2){\rm Op}_h^W (a_3)={\rm Op}_h^W (a_1a_2a_3) +hB,
$$
with,
$$
\Vert \la hD_x\ra^{-\frac{m_1+m_2+m_3}{2}}B\la hD_x\ra^{-\frac{m_1+m_2+m_3}{2}}u\Vert = {\mathcal O}\left( \Vert u\Vert + \lambda_h\Vert u\Vert_{L^2(\Omega_\delta)}\right).
$$
\end{remark}
\begin{proof}

(i) The assumptions imply that $\partial^\alpha_x\partial^\beta_\xi a ={\mathcal O}(h^{-|\alpha|})$. Then, this is just an easy consequence of the Calder\'on-Vaillancourt Theorem: see, e.g., \cite{Ma6}, exercise 2.10.15.
\vskip 0.2cm
(ii) It is well known that ${\rm Op}_h^W (a) ={\rm Op}_h^t (a_t)$, with (see, e.g., \cite{Ma6}, Theorem 2.7.1),
$$
a_t(x,\xi ):= \frac1{(2\pi h)^n}\int e^{i(\xi'-\xi)\theta /h} a(x +(t-\frac12)\theta, \xi')d\xi'd\theta,
$$
and thus,
$$
a_t(x,\xi )-a(x,\xi):= \frac1{(2\pi h)^n}\int e^{i(\xi'-\xi)\theta /h} (a(x +(t-\frac12)\theta, \xi')-a(x,\xi'))d\xi'd\theta.
$$
Then, first making $N$ integrations by parts by means of the operator $L:= \frac1{1+(\xi'-\xi)^2+\theta^2}(1+(\xi'-\xi)\cdot hD_\theta + \theta\cdot hD_{\xi'})$, and next performing the change of contour of integration given by,
$$
\R^n\ni \xi' \mapsto \xi' +i\delta\frac{\theta}{|\theta|},
$$
(with $\delta >0$ small enough), we obtain as in Appendix \ref{appB},
\begin{eqnarray*}
\partial^\alpha_x\partial_\xi^\beta (a_t(x,\xi )-a(x,\xi)) =\int {\mathcal O}\left(h^{-n}e^{-\delta|\theta|/h}\frac{(\Lambda_h(|\theta|\leq\delta \,; \, x\in\Omega_\delta) h^{-|\alpha|}}{(1+(\xi'-\xi)^2+\theta^2)^N}\right)\\
\times (|\theta|+h)\la\xi'\ra^{m}d\xi' d\theta,
\end{eqnarray*}
with the notation $\Lambda_h({\mathcal E}):= 1+\lambda_h{\bf 1}_{\mathcal E}$.
Thus, if $N$ has been taken sufficiently large, this implies,
$$
\partial^\alpha_x\partial_\xi^\beta (a_t(x,\xi )-a(x,\xi)) ={\mathcal O}(\Lambda_h( x\in\Omega_\delta) h^{1-|\alpha|}\la\xi\ra^{m}).
$$
Then, as in Appendix \ref{appB}, the result follows by taking a partition of unity adapted to the pair $(\Omega_\delta, \R^n\backslash \Omega_{\delta /2})$.
\vskip 0.2cm
(iii) For any $X=(x,\xi)\in\R^{2n}$, one has,
\begin{eqnarray*}
a\#^W b(X)=\frac1{(2\pi h)^{4n}}\int_{\R^{8n}} e^{i[(X-Y)Y^*+(X-Z)Z^* +\sigma(Y^*,Z^*)]/h}a(Y)b(Z)\\
\times d(Y,Z,Y^*,Z^*),
\end{eqnarray*}
where $\sigma$ stands for the usual canonical simplectic form on $\R^{2n}$. Therefore, integrating first with respect to $(y^*,z^*)$ (where we have set $Y^*=(y^*,\eta^*)$ and $Z^*=(z^*,\zeta^*)$), we obtain,
\begin{eqnarray*}
a\#^W b(X)=\frac1{(2\pi h)^{2n}}\int_{\R^{4n}} e^{i[(\xi-\eta)\eta^*+(\xi -\zeta)\zeta^* ]/h}a(x-\zeta^*,\eta)b(x+\eta^*,\zeta)\\
\times d(\eta,\eta^*,\zeta,\zeta^*).
\end{eqnarray*}
Here, we observe that, by Leibniz formula, for $x\in\Omega$, we have,
\be
\label{derivprod1}
\nabla_x\left( ab\right) = a\nabla_xb+b\nabla_xa={\mathcal O}(\lambda_h \la\xi\ra^{m+m'}),
\ee
and, for $|\alpha|\geq 2$,
\begin{eqnarray}
\partial_x^\alpha\left( ab\right)&=& a\partial_x^\alpha b + b\partial_x^\alpha a+\sum_{{\beta\leq\alpha}\atop {1<|\beta| <|\alpha|}}{\mathcal O}\left(|\partial^{\beta}a|\cdot|\partial^{\alpha-\beta}b|\right)\nonumber\\
&=&{\mathcal O}(\lambda_hh^{1-|\alpha|} \la\Re\xi\ra^{m+m'} +  \lambda_h^2h^{2-|\alpha|} \la\Re\xi\ra^{m+m'})\nonumber\\
\label{derivprod2}
&=& {\mathcal O}(\lambda_hh^{1-|\alpha|}\la\Re\xi\ra^{m+m'}),
\end{eqnarray}
where the last estimate comes from the fact that $h\lambda_h={\mathcal O}(1)$.
\vskip 0.2cm  
In particular $ab\in S_{m+m'}(\lambda_h, \Omega)$ and, making integrations by parts as before, and performing the change of contour,
$$
\R^{2n}\ni (\eta,\zeta) \mapsto (\eta, \zeta) -i\delta' (\frac{\eta^*}{|\eta^*|},\frac{\eta^*}{|\eta^*|}),
$$
with $\delta'>0$ small enough, we  see that, for any $\delta >0$, $a\#^W b$ belongs to $S_{m+m'}(\lambda_h, \Omega_\delta)$. Moreover, by an argument similar to that of the proof of (ii), this also gives,
\begin{eqnarray*}
a\#^W b(X)=\frac1{(2\pi h)^{2n}}\int_{\R^{4n}} e^{i[(\xi-\eta)\eta^*+(\xi -\zeta)\zeta^* ]/h}a(x,\eta)b(x,\zeta)
d(\eta,\eta^*,\zeta,\zeta^*) \\
+hc(X),
\end{eqnarray*}
with,
$$
\partial_x^\alpha\partial_\xi^\beta c(x,\xi) ={\mathcal O}(\Lambda_h(x\in \Omega_\delta)h^{-|\alpha|}\la \Re\xi\ra^{m+m'}).
$$
Thus, $c\in \widetilde{\mathcal S}_{m+m'}(\lambda_h, \Omega_\delta )$, and one finds,
$$
a\#^W b(X) =a(X)b(X) +h c(X).
$$
\vskip 0.2cm
(iv) For any fix $\delta>0$, by a partition of unity we can write $a=a_1+a_2$ with $a_1\in \widetilde{\mathcal S}_{0}(1, \R^n )$, and $a_2\in \widetilde{\mathcal S}_{0}(\lambda_h, \Omega )$ supported in $\{ x\in \Omega_{\delta/4}\}$. Then, the Calder\'on-Vaillancourt theorem tells us that ${\rm Op}_h^W (a_1)$ is uniformly bounded on $L^2(\R^n)$, and it just remains to study ${\rm Op}_h^W (a_2)$. Let $\chi\in C^\infty (\R^n; [0,1])$ be such that Supp$\chi\subset \Omega_\delta$ and $\chi =1$ near $\Omega_{\delta/2}$. Then, we write,
$$
{\rm Op}_h^W (a_2)u ={\rm Op}_h^W (a_2)(\chi u) +{\rm Op}_h^W (a_2)((1-\chi)u),
$$
and, still by the Calder\'on-Vaillancourt theorem, we have,
$$
\Vert {\rm Op}_h^W (a_2)(\chi u)\Vert ={\mathcal O}(\lambda_h \Vert \chi u\Vert ) ={\mathcal O}(\lambda_h \Vert u\Vert_{L^2(\Omega_\delta)} ).
$$
On the other hand,  since the function $a_2((x+y)/2, \xi)(1-\chi (y))$ vanishes  identically in $\{ |x-y|\leq \delta/4\}$, in the expression of ${\rm Op}_h^W (a_2)((1-\chi)u)(x)$ we can make integrations by means of the operator,
$$
L:= |x-y|^{-2} (x-y)\cdot hD_\xi.
$$
Then, for any $N\geq 1$, we obtain,
$$
{\rm Op}_h^W (a_2)((1-\chi)u)={\rm Op}_h^W (b_N)u
$$
where $b_N:= (-1)^NL^N\left( a_2((x+y)/2, \xi)(1-\chi (y))\right)$ satisfies,
$$
\partial_x^\alpha \partial_y^\beta\partial_\xi^\gamma b_N ={\mathcal O}(h^N\lambda_h h^{-|\alpha|-|\beta|}).
$$
Therefore, $\Vert {\rm Op}_h^W (b_N)\Vert ={\mathcal O}(\lambda_h h^N)={\mathcal O}(h^{N-1})$, and the result follows by taking $N\geq 2$.

\vskip 0.2cm
(v) For $a\in \widetilde {\mathcal S}_m(\lambda_h, \Omega)$ and $b\in {\mathcal S}_{m'}(\lambda_h, \Omega)$, on $\Omega$ we have,
$$
\nabla_x\left( ab\right) = a\nabla_xb+b\nabla_xa={\mathcal O}((\lambda_h^2 +\lambda_h h^{-1} )\la\Re\xi\ra^{m+m'})={\mathcal O}(\lambda_h h^{-1} \la\Re\xi\ra^{m+m'}),
$$
and, for $|\alpha|\geq 2$,
\begin{eqnarray*}
\partial_x^\alpha\left( ab\right)&=& a\partial_x^\alpha b + b\partial_x^\alpha a+\sum_{{\beta\leq\alpha}\atop {1<|\beta| <|\alpha|}}{\mathcal O}\left(|\partial^{\beta}a|\cdot|\partial^{\alpha-\beta}b|\right)\\
&=&{\mathcal O}(\lambda_h^2h^{1-|\alpha|}+\lambda_h h^{-|\alpha|} +  \lambda_h^2h^{1-|\alpha|}) \la\Re\xi\ra^{m+m'}\\
&=& {\mathcal O}(\lambda_hh^{-|\alpha|}\la\Re\xi\ra^{m+m'}).
\end{eqnarray*}
Therefore, $ab$ is in $\widetilde {\mathcal S}_{m+m'}(\lambda_h, \Omega)$, and the rest of the proof proceeds exactly as in  (iii).
\vskip 0.2cm
(vi) We introduce the modified anti-Wick quantization, given by,
$$
\widetilde{\rm Op}^{AW}_h (a) :={\rm Op}^W_h(\widetilde a),
$$
where we have set,
$$
\widetilde a (x,\xi ):= \frac1{(\pi h)^{n}}\int a(y,\eta)e^{-\lambda_h^{1/2}(x-y)^2/h^{3/2} - \lambda_h^{-1/2}(\xi -\eta)^2/h^{1/2}}dyd\eta.
$$
Then, by a second-order Taylor expansion of $a(y,\eta)$ at $(x,\xi)$, and using the fact that $\nabla_x^2a={\mathcal O}(\lambda_h h^{-1})$ and $\nabla_{\xi}^2a={\mathcal O}(1)$, we see that,
$$
\widetilde a - a={\mathcal O}(\sqrt{h\lambda_h}),
$$
uniformly. In the same way, using that $\partial_x^\alpha\partial_\xi^\beta \widetilde a =\widetilde {\partial_x^\alpha\partial_\xi^\beta a}$, we obtain,
$$
\partial_x^\alpha\partial_\xi^\beta (\widetilde a - a)={\mathcal O}(\sqrt{h\lambda_h}\, h^{-|\alpha|}).
$$
As a consequence, by the Calder\'on-Vaillancourt theorem,
\be
\label{compAW}
\Vert \widetilde{\rm Op}^{AW}_h (a) -{\rm Op}^W_h(a)\Vert ={\mathcal O}(\sqrt {h\lambda_h}).
\ee
On the other hand, for any $u\in L^2 (\R^n)$, a straightforward computation gives,
$$
\la \widetilde{\rm Op}^{AW}_h (a)u,u\ra =\pi^{-n/2}h^{-7n/4}\lambda_h^{n/2}\int a(y,\eta )|H_u(y,\eta )|^2 e^{-\lambda_h^{1/2}y^2/h^{3/2}/h}dyd\eta,
$$
where $H_u$ is the (modified) Husimi function, defined by,
$$
H_u(y,\eta) := \int e^{i[x\eta/h -\lambda_h^{1/2}x^2/(2h^{3/2}) -\lambda_h^{1/2}xy/h^{3/2}]}u(x)dx.
$$
In particular, $\la \widetilde{\rm Op}^{AW}_h (a)u,u\ra \geq 0$, and the result follows by using (\ref{compAW}).
\end{proof}

\section{An estimate on $P_2$}\label{appC}
Here we prove,
\begin{proposition}\sl
\label{propC1}
Let $E_1<E_2$ be the first two eigenvalue of $P_2=-h^2\Delta +V_2$, with $V_2$ satisfying (\ref{assV2}), and let $k=h\ln(h^{-1})$. Then, for any $z\in\C$ such that $|z-E_1|=\delta h$ with $0<\delta <\inf_{0<h<<1}(E_2-E_1)/h$, one has,
$$
\Vert (-h^2\Delta + V_2(x', x_n-ik) -z)^{-1}\Vert ={\mathcal O}(h^{-1}).
$$
\end{proposition}
\begin{proof} By a Taylor expansion, we have,
$$
V(x', x_n-ik) -z= V_2(x)  -ik\partial_{x_n}V_2(x) +k^2W,
$$
with $W={\mathcal O}(1)$, and thus,
\be
\label{V2k}
-h^2\Delta + V_2(x', x_n-ik) -z=\left(I+(-ik\partial_{x_n}V_2(x) +k^2W)(P_2-z)^{-1}\right)(P_2-z).
\ee
Moreover, the assumptions on $V_2$ imply that $(\partial_{x_n}V_2)^2 ={\mathcal O}(V_2)$. Therefore, for any $u\in L^2$, we have,
\begin{eqnarray*}
\Vert (\partial_{x_n}V_2)(P_2-z)^{-1}u\Vert^2 &=& \la (\partial_{x_n}V_2)^2(P_2-z)^{-1}u, (P_2-z)^{-1}u\ra\\
&\leq& C\la V_2(P_2-z)^{-1}u, (P_2-z)^{-1}u\ra\\
&\leq& C\la P_2(P_2-z)^{-1}u, (P_2-z)^{-1}u\ra\\
&\leq& C\la (1+z(P_2-z)^{-1})u, (P_2-z)^{-1}u\ra\\
&\leq& C\Vert u\Vert\cdot \Vert (P_2-z)^{-1}u\Vert + C|z|\cdot  \Vert (P_2-z)^{-1}u\Vert^2.
\end{eqnarray*}
Since $|z| ={\mathcal O}(h)$ and we also know that $\Vert (P_2-z)^{-1}\Vert ={\mathcal O}(h^{-1})$, we deduce,
$$
\Vert (\partial_{x_n}V_2)(P_2-z)^{-1}u\Vert^2={\mathcal O}(h^{-1})\Vert u\Vert^2,
$$
and thus,
$$
\Vert (\partial_{x_n}V_2)(P_2-z)^{-1}\Vert ={\mathcal O}(h^{-1/2}).
$$
As a consequence,
$$
\Vert (-ik\partial_{x_n}V_2(x) +k^2W)(P_2-z)^{-1}\Vert ={\mathcal O}(k h^{-1/2}+ k^2h^{-1})={\mathcal O}(\sqrt{h} \ln \frac1{h}),
$$
and the result follows by taking the inverse in (\ref{V2k}).
\end{proof}

\section{A geometric estimate}\label{appD}

Here we prove the following result, that we use in Section 5:
\begin{lemma}
The function $\widetilde\varphi_1$ defined in Section 4 is such that, for
any $s\in [0,1]$ and for any $y\in\Gamma_+$,
$$
\widetilde p_2(y,is\nabla\widetilde\varphi_1(y)) >0.
$$
\end{lemma}
\begin{proof} We first observe that
$$
\frac{\partial}{\partial s} \widetilde p_2(y,is\nabla\widetilde\varphi_1(y))
= \partial_\eta \widetilde p_2(y,is\nabla\widetilde\varphi_1(y))\cdot
i\nabla\widetilde\varphi_1(y)\leq -\frac1{C}s(\nabla\widetilde\varphi_1(y))^2\leq 0
$$
and therefore, it is enough to prove the result for $s=1$. Moreover,
since we already
know that 
$\widetilde p_2(y,i\nabla\widetilde\varphi_1(y))\not = 0$ outside $\Pi_y\Gamma$, it
is enough to prove it for $y\in\Gamma_+$ close enough to $\Pi_y\Gamma$, that
is,  for $y\in\Gamma_+$ such that
$$
(y,i\nabla\widetilde\varphi_1(y))= 
\exp tH_{\widetilde p_1}(\widetilde y,i\nabla\widetilde\varphi_2(\widetilde y))
$$
for some $t\in\C\backslash \{0\}$ small enough and $\widetilde y\in 
\Pi_y\Gamma\cap\R^n$.
In particular, we have
$$
(y,i\nabla\widetilde\varphi_1(y)) = (\widetilde y,i\nabla\widetilde\varphi_2(\widetilde y))
-t(2i\nabla\widetilde\varphi_2(\widetilde y), (0,\dots ,1))+{\mathcal O}(\vert t\vert^2)
$$
and thus, for $y$ to be real it is necessary that $t=it'$ for some real $t'$.
Moreover, in order that $y$ is in $\Gamma_+$, we need that 
$y_n>\widetilde y_n$ (at
least up to ${\mathcal O}(\vert t\vert^2)$), that is,
$$
2t'\partial_n\widetilde\varphi_2(\widetilde y)>0
$$
and thus: $t'>0$ (since $\partial_n\widetilde\varphi_2(\widetilde y) >0$).

Now, since $\widetilde p_2(\widetilde y,i\nabla\widetilde\varphi_2(\widetilde y))=0$, 
we obtain by a Taylor expansion,
\begin{eqnarray*}
\widetilde p_2(y,i\nabla\widetilde\varphi_1(y))
&=& 2t'(\partial_y\widetilde p_2)(\widetilde y,i\nabla\widetilde\varphi_2(\widetilde y))\cdot
\nabla\widetilde\varphi_2(\widetilde y) -it'
(\partial_{\eta_n}\widetilde p_2)(\widetilde y,i\nabla\widetilde\varphi_2(\widetilde y))
+{\mathcal O}(\vert t\vert^2)\\
&\geq & \frac{t'}{C}\left( \widetilde y\cdot \nabla\widetilde\varphi_2(\widetilde y)
+\partial_n\widetilde\varphi_2(\widetilde y)\right) >0
\end{eqnarray*}
This finishes the proof. 
\end{proof}

\section{An estimate on $\widetilde\phi_0$}\label{appE}

We have,
\begin{lemma}\sl
\label{lemmephi0}
Let $\widetilde\phi_0 $ be the function defined in (\ref{defphi0tilde}). Then, there exists $\varepsilon'>0$ such that, for all $y\in\R^n$ with $ y_n\leq \mu-\varepsilon$, one has,
$$
\widetilde\phi_0  (y_n)\leq  {\rm Min} \{\widetilde \phi (y), S(\mu) -\varepsilon'\}.
$$
\end{lemma} 
\begin{proof} 
By definition, for $\mu-2\varepsilon \leq y_n < g(y')$ (where $g(y')$ is the function defining the caustic set of $\widetilde \phi$), we have $(\partial_{y_n}\widetilde \phi_0(y_n))^2 =|\nabla \widetilde \phi (y)|^2=\mu -y_n$, and $\partial_{y_n}\widetilde \phi_0\geq 0$. Therefore, $|\partial_{y_n}\widetilde \phi (y)|\leq \partial_{y_n}\widetilde \phi_0 (y_n)$, and we also have $\widetilde \phi (y', g(y') \geq S(\mu)\geq \widetilde \phi_0( g(y'))$. Therefore, integrating from $y_n=g(y')-t$ to $y_n =g(y')$ (with $0\leq t\leq g(y') -\mu +2\varepsilon$), we deduce,
$$
\widetilde \phi_0( y_n)\leq \widetilde \phi (y) \,\,{\rm if}\,\, \mu-2\varepsilon \leq y_n < g(y').
$$
In particular, $\widetilde \phi_0( y_n)\leq \widetilde \phi (y)$ on $\{ \widetilde \phi (y)<S(\mu)\} \cap \{ y_n\geq \mu -2\varepsilon\}$. Since we also have $\sup_{y_n\leq\mu -\varepsilon} \widetilde \phi_0(y_n) < S(\mu)$, and $\widetilde \phi_0(y_n)=0$ for $y_n\leq \mu-2\varepsilon$, the result follows.
\end{proof}

\section{Proof of Lemma \ref{estuetalarge}}\label{appF}

If $c=c(x)$ satisfies (\ref{ass2}), we have,
$$
{\mathcal U}c{\mathcal U}^{-1}={\rm Op}^L_h (\widetilde c ),
$$
with,
\begin{eqnarray*}
\widetilde c(y,\eta )&=&\frac1{(2\pi h)^n}\int e^{i(y-x)(\xi-\eta)/h}\times \\
&&\times c(x' -2\xi_n(\xi'+\eta'),x_n-\frac23(\eta_n^2+\eta_n\xi_n+\xi_n^2) -2|\eta'|^2)dxd\xi.
\end{eqnarray*}
Here, the oscillatory integral is well defined thanks to (\ref{ass2}), and  it defines a smooth function on $\R^{2n}$, holomorphic with respect to $y$ in a complex strip around $\R^n$. Moreover, observing that,
$$
|\eta|\cdot |\nabla \la y_n\ra^{1/2}|\leq |\eta|\la y_n\ra^{-1/2} ={\mathcal O}(\la y_n-2|\eta|^2\ra^{1/2}),
$$
we see that, by (\ref{ass2}),  we can make a deformation of contour of integration in order to obtain,
$$
e^{-4k\la y_n\ra^{1/2}/h}{\mathcal U}c{\mathcal U}^{-1}e^{4k\la y_n\ra^{1/2}/h}={\rm Op}^L_h (\check c ),
$$
where $\check c$ has properties similar to those of $\widetilde c$.
Then, by standard properties of $T_1$ (see, e.g., \cite{Ma6}), we have,
$$
T_1e^{-4k\la y_n\ra^{1/2}/h}{\mathcal U}c{\mathcal U}^{-1}e^{4k\la y_n\ra^{1/2}/h} ={\rm Op}^L_h (\check c (y-\eta^*, y^*))T_1,
$$
where $(y^*,\eta^*)$ stands for the dual variable of $(y,\eta)$. 
Next, using the fact that $\check c (y-\eta^*, y^*)$ is holomorphic with respect to $y$, for any $\varepsilon >0$, we find,
\begin{eqnarray*}
e^{-\varepsilon k^{2/3}\la \eta\ra/h}T_1e^{-4k\la y_n\ra^{1/2}/h}{\mathcal U}c{\mathcal U}^{-1}e^{4k\la y_n\ra^{1/2}/h}\\
={\rm Op}^L_h ( c_\varepsilon (\eta, y-\eta^*, y^*))e^{-\varepsilon k^{2/3}\la \eta\ra/h}T_1,
\end{eqnarray*}
where $c_\varepsilon (y, \eta,\eta^*, y^*)$ can be computed explicitly, and has essentially the same properties as $\check c (y-\eta^*, y^*)$, with estimates uniform with respect to  $\varepsilon >0$ sufficiently small. 

Then, taking a cut-off function $\chi_1 =\chi_1 (|\eta^*|)$ that is equal to 1 near 0, and using the fact that $T_1^*T_1=1$, we can write,
\begin{eqnarray}
\label{cutetastar}
&&{\rm Op}^L_h ( c_\varepsilon (\eta, y-\eta^*, y^*))e^{-\varepsilon k^{2/3}\la \eta\ra/h}T_1\\
&&=
{\rm Op}^L_h ( \chi_1 (|\eta^*|) (c_\varepsilon (\eta, y-\eta^*, y^*))e^{-\varepsilon k^{2/3}\la \eta\ra/h}T_1\nonumber\\
&&\hskip 2cm+ {\rm Op}^L_h ((1- \chi_1 (|\eta^*|) )(c_\varepsilon (\eta, y-\eta^*, y^*))K_\varepsilon e^{-\varepsilon k^{2/3}\la \eta\ra/h}T_1,\nonumber
\end{eqnarray}
with $K_\varepsilon:= e^{-\varepsilon k^{2/3}\la \eta\ra/h}T_1T_1^*e^{\varepsilon k^{2/3}\la \eta\ra/h}$.

A simple computation (see, e.g., \cite{Ma6}, Section 3) shows that the distributional kernel of $K_\varepsilon$ is given by,
$$
K_\varepsilon (y,\eta; z,\zeta):= (\pi h)^{\frac{n}2}e^{i(y-z)(\eta +\zeta)/2h -(y-z)^2/4h -(\eta -\zeta)^2/4h+ \varepsilon k^{2/3}(\la \zeta\ra -\la\eta\ra)/h}.
$$
Therefore, the distributional kernel of $ {\rm Op}^L_h ((1- \chi_1 (|\eta^*|) )(c_\varepsilon (\eta, y-\eta^*, y^*))K_\varepsilon$ is,
$$
I_1(y,\eta; z,\zeta) := \frac1{2^n(\pi h)^{\frac{3n}2}}\int e^{ \Phi_1/h}(1- \chi_1 (|\eta^*|) )c_\varepsilon (\eta, y-\eta^*, y^*)dy'd\eta'dy^*d\eta^*,
$$
with,
\begin{eqnarray*}
\Phi_1:= i(y-y')y^* + i(\eta -\eta')\eta^*+\frac{i}2(y'-z)(\eta' +\zeta) -\frac14(y'-z)^2\\
 -\frac14(\eta' -\zeta)^2
+ \varepsilon k^{2/3}(\la \zeta\ra -\la\eta\ra),
\end{eqnarray*}
while the distributional kernel of ${\rm Op}^L_h ( \chi_1 (|\eta^*|) (c_\varepsilon (\eta, y-\eta^*, y^*))$ is,
$$
I_2(y,\eta; z,\zeta) := \frac1{(2\pi h)^{n}}\int e^{i(y-z)y^*/h + i(\eta -\zeta)\eta^*/h} \chi_1 (|\eta^*|) c_\varepsilon (\eta, y-\eta^*, y^*)dy^*d\eta^*.
$$
In $I_1$, we can perform the change of contour of integration,
$$
\eta' \mapsto \eta' -i\delta \eta^*/\la\eta^*\ra,
$$
with $\delta >0$ small enough and independent of $\mu$. Since $k^{2/3}(\la \zeta\ra -\la\eta\ra)\leq k^{4/3} + (\eta -\zeta)^2$ and $k^{4/3}/h << 1$, we see that the resulting phase factor can be estimated by,
$$
e^{-\delta |\eta^*|^2/2h\la\eta^*\ra -(y-z)^2/8h -(\eta -\zeta)^2/8h}.
$$
As a consequence, since $\eta^*$ stays away from 0 on the support of $1- \chi_1 (|\eta^*|)$, and since, for any $N\geq 1$ one has $|y^*|^{2N}\la \eta_n^*+ |y^*|^2\ra^{-N} e^{-\delta \la\eta^*\ra/4h} ={\mathcal O}(1)$, by the Calder\'on-Vaillancourt theorem we conclude that, for any $\chi\in C_0^\infty (\R^n)$, we have,
\begin{eqnarray*}
&& \Vert \chi (y){\rm Op}^L_h ((1- \chi_1 (|\eta^*|) )(c_\varepsilon (\eta, y-\eta^*, y^*))K_\varepsilon e^{-\varepsilon k^{2/3}\la \eta\ra/h}T_1\check{\mathbf u}\Vert \\
&& \hskip 4cm ={\mathcal O}(e^{-\delta_1/h})\Vert e^{-\varepsilon k^{2/3}\la \eta\ra/h}T_1\check{\mathbf u}\Vert,
\end{eqnarray*}
where we have set,
$$
\check{\mathbf u}(y) := e^{-4k\la y_n\ra^{1/2}/h}\widetilde{\mathbf u}(y),
$$
and where  $\delta_1>0$ is independent of $\mu$. In particular,  if $\psi \in C_b^\infty (\R^{2n};\R_+)$ is such that  $\sup|\psi | \leq\delta_1/2$, we obtain,
\begin{eqnarray}
&& \Vert e^{\psi/h}\chi (y){\rm Op}^L_h ((1- \chi_1 (|\eta^*|) )(c_\varepsilon (\eta, y-\eta^*, y^*))K_\varepsilon e^{-\varepsilon k^{2/3}\la \eta\ra/h}T_1\check{\mathbf u}\Vert \nonumber \\
&& \hskip 4cm ={\mathcal O}(e^{-\delta_1/2h})\Vert e^{\psi/h-\varepsilon k^{2/3}\la \eta\ra/h}T_1\check{\mathbf u}\Vert,
\end{eqnarray}

On the other hand,  for $y$ in some arbitrarily large fixed compact set $K$,  in $I_2$ we can successively perform the two changes of contour,
\begin{eqnarray}
&& \R^n \ni\eta^*\mapsto \eta^* +i\delta \chi_2(|\eta^*|)(\eta-\zeta)/|\eta -\zeta|\nonumber\\
\label{chcontourangul}
&&S^{n-1}\ni  \omega^*:= \eta^*/|\eta^*| \mapsto \omega^*+i\delta (\eta-\zeta)/|\eta -\zeta|,
\end{eqnarray}
where $\chi_2\in C_0^\infty$ is supported inside $\{ \chi_1 =1\}$, $\chi_2 =1$ near 0, and $\delta >0$ depends on $K$ only. This makes appear the exponential factor,
$$
e^{-\delta (|\eta^*|+\chi_2(|\eta^*|))|\eta -\zeta|/h }.
$$
Moreover, since $|\eta^*|$ remains bounded on the support of $\chi_1(|\eta^*|)$, if $\psi \in C_b^\infty (\R^{2n};\R_+)$ is such that  $\sup|\nabla\psi| <<\delta$, writing: $\psi(y,\eta)=\psi(z,\zeta) +(y-z)\phi_1 + (\eta-\zeta)\phi_2$, then we see that we can also perform the change of contour,
$$
\R^n\ni y^*\mapsto y^* +i\phi_1(y,z,\eta,\zeta),
$$
and we easily conclude that,
\begin{eqnarray*}
\Vert e^{\psi /h} \chi (y)\chi_1(|\eta^*|){\rm Op}^L_h ( c_\varepsilon (\eta, y-\eta^*, y^*))e^{-\varepsilon k^{2/3}\la \eta\ra/h}T_1\check{\mathbf u}\Vert \\
={\mathcal O}\left(\Vert e^{\psi /h-\varepsilon k^{2/3}\la \eta\ra/h}T_1\check{\mathbf u}\Vert\right).
\end{eqnarray*}
 Here, $\chi \in C_0^\infty (\R^n)$ is an arbitrary cut-off function supported in $K$. 
 
 Summing up, we have proved that, for any  $\chi \in C_0^\infty (\R^n)$ and  for any  $\psi\in C_b^\infty (\R^{2n}; \R_+)$ with $\sup|\psi| + \sup |\nabla\psi| $ small enough (independently of $\mu$), there exists a constant $C>0$ (independent of $\mu$) such that, for any $\varepsilon >0$, we have,
 \be
 \label{fondestopc}
 \Vert \chi (y)e^{\psi/h-\varepsilon k^{2/3}\la \eta\ra/h}T_1{\mathcal U}c{\mathcal U}^{-1}\check{\mathbf u}\Vert \leq C\Vert e^{\psi/h-\varepsilon k^{2/3}\la \eta\ra/h}T_1\check{\mathbf u}\Vert.
 \ee
 It is important to observe that, in all these computations, the estimates are uniform with respect to $\varepsilon >0$. Actually, the presence of $\varepsilon$ is required only to insure that the norms are finite.

\begin{remark}\sl
One can extract from the previous proof the following  (maybe already well-know) result: If $a(x,\xi)$ is holomorphic with respect to $\xi$ is a complex strip $\Gamma$ around $\R^n$ ($n\geq 2$), and is bounded together with all its derivatives in $\R^n\times\Gamma$, then, for any cut-off function $\chi\in C_0^\infty (\R_+)$ such that $\chi =1$ near 0, and for any $\varphi \in C_b^\infty (\R^n)$ with $\sup(|\varphi| +|\nabla\varphi|)$ small enough, the operator $e^{\varphi/h}{\rm Op}_h^W(\chi (|\xi|)a(x,\xi))e^{-\varphi/h}$ is uniformly bounded on $L^2(\R^n)$. (Indeed, this can be seen by first making integrations by parts with respect to $\xi$, and then by performing   changes of contour similar to those of (\ref{chcontourangul}).)
\end{remark}

In the sequel, we also need,
\begin{lemma}\sl 
\label{continctilde}
There exists a constant $C>0$ such that, for all $\varepsilon, h>0$ small enough, one has,
\begin{eqnarray*}
&& \Vert e^{-\varepsilon k^{2/3}\la\eta\ra/h}T_1e^{-4k\la y_n\ra^{1/2}/h}{\mathcal U}c{\mathcal U}^{-1}e^{4k\la y_n\ra^{1/2}/h} \check{\mathbf u}\Vert\\
&&  \leq C\Vert (1+\frac{\la\eta\ra^{2C}}{\la y_n-2\eta^2\ra^C})e^{-\varepsilon k^{2/3}\la\eta\ra/h}T_1\check {\mathbf u}\Vert.
\end{eqnarray*}
\end{lemma}
\begin{proof} Using (\ref{ass2}), one can check that the symbol $c_\varepsilon$ satisfies,
$$
\partial_{y, \eta,  \eta^*}^\alpha \partial_{y^*}^\beta c_\varepsilon={\mathcal O}(1+\la y^*\ra^{|\beta|}\la y_n-\eta_n^* -2|y^*|^2\ra^{-|\beta|/2}),
$$
uniformly with respect to $\varepsilon$.
Therefore, we learn from the Calder\'on-Vaillancourt Theorem the existence of some constant $C>0$ such that, 
\begin{eqnarray*}
&&\Vert {\rm Op}^L_h ( c_\varepsilon (\eta, y-\eta^*, y^*))e^{-\varepsilon k^{2/3}\la \eta\ra/h}T_1 \check{\mathbf u}\Vert\\
 &&\leq C\Vert (1+\la hD_y\ra^{2C}\la y_n-hD_{\eta_n}-2(hD_y)^2\ra^{-C})e^{-\varepsilon k^{2/3}\la \eta\ra/h}T_1 \check{\mathbf u}\Vert.
\end{eqnarray*}
Then, the result follows from standard properties of the transform $T_1$, that permit us to replace $hD_y$ by $\eta$, and $hD_\eta$ by 0, in the previous estimate (see, e.g., \cite{Ma6}).
\end{proof}

Now, we set,
$$
\check P:= e^{-4k\la y_n\ra^{1/2}/h}\widetilde Pe^{4k\la y_n\ra^{1/2}/h},
$$
and we
observe that, for any integer $N\geq 1$, the operator,
$$
(\check P -\rho)^N =e^{-4k\la y_n\ra^{1/2}/h}{\mathcal U}(P-\rho)^N{\mathcal U}^{-1}e^{4k\la y_n\ra^{1/2}/h},
$$
is a sum of terms of the type: $a(y, hD_y)e^{-4k\la y_n\ra^{1/2}/h}{\mathcal U}c{\mathcal U}^{-1}e^{4k\la y_n\ra^{1/2}/h}$, where $a(y, hD_y)$ is a $2\times 2$ matrix of partial differential operator with polynomial coefficients, and $c=c(x)$ satisfies (\ref{ass2}). Moreover, the leading term (that is, the one corresponding to the operator $a(y, hD_y)$ with higher degree in $hD_y$) is just $(-h^2\Delta)^N\left(\begin{array}{cc} (-1)^N & 0\\ 0 & 1\end{array}\right)$. Therefore,
combining (\ref{fondestopc}) with standard microlocal exponential weighted estimates (see \cite{Ma6}, Section 3.5), we see that the equation $(\check P -\rho)^N\check{\mathbf u}=0$ implies that,  for any  $\chi \in C_0^\infty (\R^n)$, for any  $\psi\in C_b^\infty (\R^{2n}; \R_+)$ with $\sup|\psi| + \sup |\nabla\psi| $ small enough,  and for any $\varepsilon >0$ small enough, one has,
\be
\label{estellipeta}
\Vert |\eta|^{2N}\chi (y)e^{\psi /h-\varepsilon k^{2/3}\la\eta\ra/h}T_1 \check{\mathbf u}\Vert \leq C\Vert \la \eta\ra^{2N-1}e^{\psi /h-\varepsilon k^{2/3}\la\eta\ra/h}T_1 \check{\mathbf u}\Vert,
\ee
where the constant $C$ does not depend on $\varepsilon$, and depends on $\psi$ through the quantities $\sup| \partial^\alpha\psi|$ only. In particular, we can take $\psi$ such that $\psi =0$ on $\{\chi (y)\not= 1\}\cup \{ |\eta|^{2N}\leq 2C\la \eta\ra^{2N-1}\}$, and $\psi \geq \delta_0$ on $K\times \{|\eta | \geq C_0\}$ with $\delta_0, C_0>0$ independent of $\varepsilon$ and $\mu$, and where $K$ is an arbitrary compact subset of the interior of $\{\chi =1\}$. In this case, we deduce from (\ref{estellipeta}),
\be
\label{estmic1}
\Vert |\eta|^{2N}e^{\psi /h-\varepsilon k^{2/3}\la\eta\ra/h}T_1 \check{\mathbf u}\Vert_{{\rm Supp}\psi} \leq 2C\Vert \la \eta\ra^{2N-1}e^{-\varepsilon k^{2/3}\la\eta\ra/h}T_1 \check{\mathbf u}\Vert_{\R^{2n}\backslash{\rm Supp}\psi}
\ee
On the other hand, using Lemma \ref{continctilde} and  the equation $(\check P -\rho)^N\check{\mathbf u}=0$, we also deduce,
\begin{eqnarray*}
\Vert |y_n -\eta^2|^{N}e^{-\varepsilon k^{2/3}\la\eta\ra/h}T_1 \check{\mathbf u}_1\Vert +\Vert |\eta|^{2N}e^{-\varepsilon k^{2/3}\la\eta\ra/h}T_1 \check{\mathbf u}_2\Vert\\
\leq C\Vert (\la \eta\ra^{2N-1}+ \la y_n -\eta^2\ra^{N-1})(1+ \frac{\la \eta\ra^{2C}}{\la y_n-2\eta^2\ra^C})e^{-\varepsilon k^{2/3}\la\eta\ra/h}T_1 \check{\mathbf u}\Vert, 
\end{eqnarray*}
with some other constant $C>0$, still independent of $\varepsilon$. Now, on $\{\eta^2\geq 2|y_n|\}$, we have $\eta^2-y_n\geq \frac12\eta^2$, while $\la \eta\ra^{2C}\la y_n-2\eta^2\ra^{-C}$ is bounded. Therefore, we obtain,
\begin{eqnarray*}
&&\Vert |\eta|^{2N}e^{-\varepsilon k^{2/3}\la\eta\ra/h}T_1 \check{\mathbf u}\Vert_{\eta^2\geq 2|y_n|}\\
&&\leq 
C\Vert (\la \eta\ra^{2N-1}+ \la y_n -\eta^2\ra^{N-1})(1+ \frac{\la \eta\ra^{2C}}{\la y_n-2\eta^2\ra^C})e^{-\varepsilon k^{2/3}\la\eta\ra/h}T_1 \check{\mathbf u}\Vert_{\eta^2\leq 2|y_n|}\\
&&\hskip 5cm +C\Vert |\eta|^{2N-1}e^{-\varepsilon k^{2/3}\la\eta\ra/h}T_1 \check{\mathbf u}\Vert_{\eta^2\geq 2|y_n|},
\end{eqnarray*}
and thus, if $C_1>0$ is sufficiently large (in order to have $|\eta|^{2N}\geq \la\eta\ra^{2N-1}$ when $|\eta|^2\geq C_1$),
\begin{eqnarray}
\label{F7}
&&\Vert |\eta|^{2N}e^{-\varepsilon k^{2/3}\la\eta\ra/h}T_1 \check{\mathbf u}\Vert_{\eta^2\geq \max(C_1,2|y_n|)}\\
&&\leq 
C_N\Vert  \la y_n\ra^{N+C}e^{-\varepsilon k^{2/3}\la\eta\ra/h}T_1 \check{\mathbf u}\Vert_{\eta^2\leq 2|y_n|}+C_N\Vert e^{-\varepsilon k^{2/3}\la\eta\ra/h}T_1 \check{\mathbf u}\Vert_{ \eta^2\leq C_1}, \nonumber
\end{eqnarray}
with $C_N>0$ independent of $\varepsilon$. Setting $\theta_1=\theta_1(y_n;h):= 4k\la y_n\ra^{1/2}/h$ and $\theta_2=\theta_2(\eta;h):=  k\la \eta\ra/h$, for $\eta^2\leq 2|y_n|$, we write,
\begin{eqnarray*}
\la y_n\ra^{N+C}T_1 \check{\mathbf u}&=&\la y_n\ra^{N+C}T_1 e^{-\theta_1}\widetilde{\mathbf u}\\
&=& \la y_n\ra^{N+C}e^{\theta_2-\theta_1}(e^{\theta_1-\theta_2}T_1 e^{-\theta_1} T_1^*e^{\theta_2}) e^{-\theta_2}T_1\widetilde{\mathbf u},
\end{eqnarray*}
and since $\theta_2(\eta)-\theta_1(y_n)\leq -\frac12\theta_1(y_n)$, and $\la y_n\ra^{N+C}e^{-\frac12\theta_1(y_n)}={\mathcal O}(1)$, we obtain,
\be
\label{F8}
\Vert  \la y_n\ra^{N+C}e^{-\varepsilon k^{2/3}\la\eta\ra/h}T_1 \check{\mathbf u}\Vert_{\eta^2\leq 2|y_n|} \leq C'_N\Vert (e^{\theta_1-\theta_2}T_1 e^{-\theta_1} T_1^*e^{\theta_2}) e^{-\theta_2}T_1\widetilde{\mathbf u}\Vert,
\ee
with $C'_N>0$ independent of $\varepsilon$.
Then, we have,
\begin{lemma}\sl 
\label{commutpoidsT}
The operator $J:= e^{\theta_1-\theta_2}T_1 e^{-\theta_1} T_1^*e^{\theta_2}$ is bounded on $L^2(\R^{2n})$, uniformly with respect to $h>0$ small enough.
\end{lemma}
\begin{proof} The distributional kernel of $J$ is given by,
\begin{eqnarray*}
J(y,\eta; z,\zeta)=\int_{\R^n} e^{i(y-x)\eta/h+i(x-z)\zeta/h -(y-x)^2/2h -(x-z)^2/2h} \\
\times e^{\theta_2(\zeta)-\theta_2(\eta)+\theta_1(y_n)-\theta_1(x_n)}dx.
\end{eqnarray*}
In this integral, we can perform the change of contour,
$$
\R^n \ni x\mapsto x+ik^{2/3} \frac{\zeta -\eta}{|\zeta -\eta|},
$$
and, since $|\nabla\theta_1|\leq \frac12 k^{2/3}$ and  $|\nabla\theta_2|\leq \frac14k^{2/3}$, we obtain,
\begin{eqnarray*}
|J(y,\eta; z,\zeta)|\leq \int_{\R^n}e^{-k^{2/3}|\zeta -\eta|/h-(y-x)^2/2h -(x-z)^2/2h+k^{4/3}/h} \\
\times e^{k^{2/3}|\zeta -\eta|/4h + k^{2/3}|y-x|/2h}dx.
\end{eqnarray*}
Therefore, using the fact that $k^{2/3}|y-x|\leq 2k^{4/3} + (y-x)^2/2$ and $k^{4/3}/h\rightarrow 0$, for $h$ small enough, we find,
$$
|J(y,\eta; z,\zeta)|\leq 2\int_{\R^n}e^{-3k^{2/3}|\zeta -\eta|/4h-(y-x)^2/4h -(x-z)^2/2h}dx,
$$
and the result follows from the Schur Lemma (see, e.g., \cite{Ma6}).
\end{proof}
We deduce from Lemma \ref{commutpoidsT} and (\ref{F7})-(\ref{F8}) that we have,
\begin{eqnarray*}
&&\Vert |\eta|^{2N}e^{-\varepsilon k^{2/3}\la\eta\ra/h}T_1 \check{\mathbf u}\Vert_{\eta^2\geq \max(C_1,2|y_n|)}\\
&&\leq 
C_N\Vert  e^{-k\la\eta\ra/h}T_1 \widetilde{\mathbf u}\Vert+C_N\Vert e^{-\varepsilon k^{2/3}\la\eta\ra/h}T_1 \check{\mathbf u}\Vert_{ \eta^2\leq C_1}, 
\end{eqnarray*}
with $C_1,C_N$ independent of $\varepsilon$, and thus,
\begin{eqnarray*}
&&\Vert |\eta|^{2N}e^{-\varepsilon k^{2/3}\la\eta\ra/h}T_1 \check{\mathbf u}\Vert_{\eta^2\geq \max(C_1,2|y_n|)}\\
&&\leq 
C_N\Vert  e^{-k\la\eta\ra/h}T_1 \widetilde{\mathbf u}\Vert+C_Nh^{-C_1}\Vert e^{- k\la\eta\ra/h}T_1 e^{-4k\la y_n\ra^{2/3}}\widetilde {\mathbf u}\Vert.
\end{eqnarray*}
Now, one can see as in the proof of Lemma \ref{commutpoidsT} that we have,
\be
\label{F9}
\Vert e^{- k\la\eta\ra/h}T_1 e^{-4k\la y_n\ra^{2/3}}\widetilde {\mathbf u}\Vert \leq C \Vert e^{- k\la\eta\ra/h}T_1 \widetilde {\mathbf u}\Vert,
\ee
and, in virtue of (\ref{estaprioriTu}), we deduce,
$$
\Vert |\eta|^{2N}e^{-\varepsilon k^{2/3}\la\eta\ra/h}T_1 \check{\mathbf u}\Vert_{\eta^2\geq \max(C_1,2|y_n|)} ={\mathcal O}(h^{-C_1}),
$$
uniformly with respect to $h$ and $\varepsilon$. Sending $\varepsilon$ to 0 in this estimate, we conclude,
\be
\label{F10}
\Vert |\eta|^{2N}T_1 \check{\mathbf u}\Vert_{\eta^2\geq \max(C_1,2|y_n|)} ={\mathcal O}(h^{-C_1}).
\ee
Gathering (\ref{F8}), (\ref{F9}), (\ref{F10}), and using Lemma \ref{commutpoidsT}, we obtain,
\be
\label{F11}
\Vert \la \eta\ra^{2N} T_1\check {\mathbf u}\Vert ={\mathcal O}(h^{-C}),
\ee
and, inserting (\ref{F11}) into (\ref{estmic1}), and sending again $\varepsilon$ to 0, this gives us,
\be
\label{F12}
\Vert |\eta|^{2N}e^{\psi /h}T_1 \check{\mathbf u}\Vert_{{\rm Supp}\psi} ={\mathcal O}(h^{-C}).
\ee
Therefore, thanks to our choice of the function $\psi$, we finally obtain,
\be
\label{F13}
\Vert \la\eta\ra^{2N}T_1 \check{\mathbf u}\Vert_{K\times \{|\eta|\geq C_0\}} ={\mathcal O}(e^{-\delta_0/2h}).
\ee
where $K\subset\R^n$ is an arbitrary compact set, $ C_0, N\geq 1$ are arbitraryily large and $\delta_0=\delta_0(K,C_0,N)$ is a positive constant that does not depend on $\mu$.
\vskip 0.2cm
To complete the proof, we fix some arbitrary constant $C\geq 1$, we take $K'\subset \subset K\cap \{\la y_n\ra\leq C^2\}$, and we write,
\begin{eqnarray*}
&&\Vert \la \eta\ra^{2N}T_1 \widetilde{\mathbf u}\Vert_{K'\times \{|\eta|\geq 2C_0\}}\\
 &&\leq h^{-4C}\Vert \la \eta\ra^{2N}e^{-4k\la y_n\ra^{1/2}/h}T_1 \widetilde{\mathbf u}\Vert_{K'\times \{|\eta|\geq 2C_0\}} \\
&&\leq  h^{-4C}\Vert \la \eta\ra^{2N}e^{ -4k\la y_n\ra^{1/2}/h}T_1e^{ 4k\la y_n\ra^{1/2}/h}\check{\mathbf u}\Vert_{K'\times \{|\eta|\geq 2C_0\}}\\
&&\leq  h^{-4C}\Vert J_1\left( \la\eta\ra^{2N}T_1\check{\mathbf u}\right)\Vert_{K'\times \{|\eta|\geq 2C_0\}},
\end{eqnarray*}
with,
$$
J_1:=\la \eta\ra^{2N}e^{ -4k\la y_n\ra^{1/2}/h}T_1e^{ 4k\la y_n\ra^{1/2}/h}T_1^*\la\eta\ra^{-2N}.
$$
Then, in a similar way as in Lemma \ref{commutpoidsT}, we see that the distributional kernel of $J_1$ satisfies,
$$
|J_1(y,\eta;z,\zeta)|\leq 2\int e^{-\delta |\zeta -\eta|/h-(y-x)^2/4h - (x-z)^2/2h +\delta^2/h}dx,
$$
where $\delta >0$ can be taken arbitrarily small. We deduce,
\begin{eqnarray*}
\Vert J_1\left( \la\eta\ra^{2N}T_1\check{\mathbf u}\right)\Vert_{K'\times \{|\eta|\geq 2C_0\}}\leq
2e^{\delta^2/h}\Vert  \la\eta\ra^{2N}T_1\check{\mathbf u}\Vert_{K\times \{|\eta|\geq C_0\}}\\
+e^{\delta^2/h-\delta_1/h}\Vert \la\eta\ra^{2N}T_1\check{\mathbf u}\Vert,
\end{eqnarray*}
where the positive number $\delta_1$ only depends on the distance between $K'$ and $\R^n \backslash K$. As a consequence, taking $\delta^2\leq \min (\frac12\delta_1, \frac14 \delta_0)$, and using (\ref{F12})-(\ref{F13}) we conclude,
$$
\Vert \la \eta\ra^{2N}T_1 \widetilde{\mathbf u}\Vert_{K'\times \{|\eta|\geq 2C_0\}}={\mathcal O}(h^{-4C}e^{-\delta_0/4h} + h^{-5C}e^{-\delta_1/2h}).
$$
Since $K$ was an arbitrary compact set, so is $K'$, and thus Lemma \ref{estuetalarge} is proved.

{}


\begin{thebibliography}{}

\vskip 0.5cm

\bibitem[Ag] {Ag} Agmon, S., {\it Lectures on Exponential Decay 
of Solutions of Second Order Elliptic
Equations, Princeton University Press}, Princeton (1982)

\bibitem[Ba] {Ba}  Baklouti, H., {\it Asymptotique des largeurs de r\'esonances pour un mod\`ele d'effet tunnel microlocal},  Ann. Inst. H. Poincar\'e Phys. Th\'eor.  68  (1998),  no. 2, 179Ð228.

\bibitem[FuLaMa] {FuLaMa} Fujiie, S., Lahmar-Benbernou, A., Martinez, A., {\it Width of shape resonances for non globally analytic potentials}, J. Math. Soc. Japan Volume 63, Number 1 (2011), 1-78.

\bibitem[GrMa] {GrMa} Grigis, A., Martinez, A. {\it Resonance widths for the molecular predissociation},  Preprint arXiv 1205.5196 (http://arxiv.org/abs/1205.5196)

\bibitem[GuSt] {GuSt} Guillemin, Sternberg, {\it Geometric Asymptotics}, Amer. Math.
Soc. Survey 14 (1977).

\bibitem[HeSj1] {HeSj1} Helffer, B., Sj\"ostrand, J., {\it Multiple Wells
in the Semiclassical Limit I}, Comm. Part. Diff. Eq., vol. 9, (4), 337-408 (1984).

\bibitem[HeSj2] {HeSj2} Helffer, B., Sj\"ostrand, J., {\it R\'esonances
en limite semi-classique }, M\'emoire Bull. Soc. Math. France, n. 24/25, tome 114 (1986).

\bibitem[Ma1] {Ma1} Martinez, A., {\it Estimations de l'effet tunnel pour
le double puits I}, J. Math. Pures et App. 66 (1987), 93-116

\bibitem[Ma2] {Ma2} Martinez, A., {\it R\'esonances dans l'approximation de 
Born-Oppenheimer I},
J. Diff. Eq., Vol. 91, No. 2, 204-234 (1991)

\bibitem[Ma3] {Ma3} Martinez, A., {\it R\'esonances dans l'approximation de 
Born-Oppenheimer II - Largeur des r\'esonances},
Comm.Math.Physics.135, 517-530 (1991)

\bibitem[Ma4] {Ma4} Martinez, A., {\it
Estimates on Complex Interactions in Phase Space}, Math. Nachr. 167, 203-254 (1994) 

\bibitem[Ma5] {Ma5} Martinez, A., {\it
Precise exponential estimates in adiabatic theory}, J. of Math. Phys. 35 (8), 
p.3889-3915 (1994)

\bibitem[Ma6] {Ma6} Martinez, A., {\it An Introduction to Semiclassical and 
Microlocal Analysis}, UTX Series, Springer-Verlag New-York (2002)

\bibitem[MaNaSo] {MaNaSo} Martinez, A., Nakamura, S., Sordoni, V. 
{\it Phase Space Tunneling in Multistate Scattering}, J. Funct. An. 191,
297-317 (2002)

\bibitem[Na] {Na}  Nakamura, S., 
{\it On an Example of Phase-Space Tunneling}, Ann. Inst. H. Poincare, phys. theo. 63, 211-229 (1995).

\bibitem[NaSo] {NaSo}  Nakamura, S., Sordoni, V., {\it A remark on exponential estimates in adiabatic theory}, Comm. Partial Differential Equations  29, 111-132 (2004)

\bibitem[Pe] {Pe} Pettersson, P., {\it WKB expansions for systems 
of Schrödinger operators
with crossing eigenvalues}, Asymptot. Anal. 14 (1997), no. 1, 1--48.

\bibitem[Sj] {Sj}  Sj\"ostrand, J., {\it Singularit\'es 
analytiques microlocales}, Ast\'erisque No 95 


\end{thebibliography}
\end{document}